\documentclass[11pt]{article}

\usepackage[utf8]{inputenc}
\usepackage{float}
\usepackage{epigraph}

\usepackage{amsmath,amsthm,amsfonts}
\usepackage{mathtools}
\usepackage{thmtools} %
\usepackage{thm-restate} %

\usepackage{centernot}
\usepackage{graphicx,xspace}
\usepackage{caption}
\usepackage[at]{easylist}
\usepackage{paralist}
\usepackage{enumitem} %

\usepackage{fullpage}
\usepackage[dvipsnames]{xcolor}
\usepackage[hidelinks]{hyperref}
\usepackage{natbib}
\newcommand{\para}[1]{\paragraph{#1}}

\usepackage{cleveref}

\usepackage[toc]{appendix}
\usepackage{xspace}
\usepackage{setspace}
\usepackage[normalem]{ulem}

\usepackage{threeparttable}
\usepackage{booktabs}
\usepackage{array} %
\newcolumntype{L}[1]{>{\raggedright\let\newline\\\arraybackslash\hspace{0pt}}m{#1}}
\newcolumntype{C}[1]{>{\centering\let\newline\\\arraybackslash\hspace{0pt}}m{#1}}
\newcolumntype{R}[1]{>{\raggedleft\let\newline\\\arraybackslash\hspace{0pt}}m{#1}}

\newcommand{\MULcites}{\citep[\dots]{cao2015towards,ginart2019making,ullah2021machine,bourtoule2021machine,sekhari2021remember,gupta2021adaptive}\xspace}

\usepackage{algorithm}
\usepackage[noend]{algpseudocode}
\usepackage{algorithmicx}

\usepackage{amssymb}
\usepackage{pifont} %
\usepackage{bm} %

\definecolor{ForestGreen}{rgb}{0.13, 0.55, 0.13}
\definecolor{Brown}{rgb}{0.59, 0.29, 0.0}

\newcommand{\greencheck}{{\textcolor{ForestGreen}{\ding{51}}}}
\newcommand{\overrestrictive}{{\textcolor{red}{{$\bm{\top}$}}}}
\newcommand{\underrestrictive}{{\textcolor{red}{{$\bm{\bot}$}}}}
\newcommand{\na}{{\textcolor{red}{\ding{55}}}}

\newtheorem{theorem}{Theorem}[section]

\newtheorem{claim}[theorem]{Claim}
\newtheorem{proposition}[theorem]{Proposition}
\newtheorem{lemma}[theorem]{Lemma}
\newtheorem{definition}[theorem]{Definition}
\newtheorem{remark}[theorem]{Remark}
\newtheorem{corollary}[theorem]{Corollary}
\newtheorem{example}[theorem]{Example}

\newcommand{\nt}[1]{{\color{blue}{#1}}}
\newcommand{\ot}[1]{{\color{red}{\sout{#1}}}}
\renewcommand{\ot}[1]{}%

\newcommand{\reals}{\mathbb{R}}
\newcommand{\eps}{\epsilon}
\newcommand{\cX}{{\mathcal{X}}}

\newcommand{\approxequiv}[1]{\ {\stackrel{{\tiny #1}}{\approx}}\ }

\newcommand{\cA}{{\mathcal{A}}}
\newcommand{\cB}{{\mathcal{B}}}
\newcommand{\cL}{{\mathcal{L}}}

\newcommand{\cY}{{\mathcal{Y}}}
\newcommand{\cZ}{{\mathcal{Z}}}

\newcommand{\bbN}{{\mathbb{N}}}

\newcommand{\set}[1]{{\left\{ {#1} \right \}}}
\newcommand{\abs}[1]{{\left| {#1} \right |}}
\newcommand{\norm}[2]{{\left\| {#1} \right \|_{#2}}}

\newcommand{\paren}[1]{{\left( {#1} \right )}}

\newcommand{\Bparen}[1]{{\Big( {#1} \Big )}}

\newcommand{\party}[1]{\ensuremath{\mathsf{#1}}\xspace}
\newcommand{\subj}{\party{Y}}
\newcommand{\env}{\party{E}}
\newcommand{\dummy}{\party{D}}
\newcommand{\cont}{\party{C}}
\newcommand{\simulator}{\party{Sim}}
\newcommand{\Sim}{\simulator}
\newcommand{\adv}{\party{A}}

\newcommand{\subjSilent}{\ensuremath{\mathcal{Y}_{\mathsf{silent}}}\xspace} %
\newcommand{\subjLift}{\ensuremath{\mathcal{Y}_{\mathsf{lift}}}\xspace} %
\newcommand{\subjDummy}{\ensuremath{\mathcal{Y}_{\mathsf{dummy}}}\xspace} %

\newcommand{\seq}{\ensuremath{\sigma}\xspace} %
\newcommand{\ins}[1]{\idkeyword{insert}(#1)}
\newcommand{\del}[1]{\idkeyword{delete}(#1)}

\newcommand{\dict}{\ensuremath{\mathcal{D}}\xspace}

\newcommand{\idkeyword}[1]{{\text{\sf {#1}}}\xspace}

\newcommand{\delete}{\idkeyword{delete}}
\newcommand{\predict}{\idkeyword{predict}}

\newcommand{\cid}{\ensuremath{\mathsf{cID}}\xspace}

\newcommand{\fail}{\ensuremath{\mathsf{fail}}\xspace}

\newcommand{\finish}{\texttt{finish}\xspace}

\newcommand{\s}{\ensuremath{\mathsf{state}}\xspace}

\newcommand{\Rand}{\ensuremath{R}\xspace}
\newcommand{\rand}{\ensuremath{r}\xspace}
\newcommand{\distrand}{\ensuremath{\mathcal{U}}\xspace} %
\newcommand{\out}{\ensuremath{\mathsf{out}}}
\newcommand{\Out}{\ensuremath{\overrightarrow{\mathsf{out}}}} %
\newcommand{\id}{{\ensuremath{id}}\xspace}
\newcommand{\msg}{\ensuremath{\mathsf{msg}}\xspace}

\newcommand{\queries}{\ensuremath{\vec{\mathsf{q}}}\xspace}
\newcommand{\answers}{\ensuremath{\vec{\mathsf{a}}}\xspace}
\newcommand{\sender}{\ensuremath{\mathsf{sender}}\xspace}
\newcommand{\receiver}{\ensuremath{\mathsf{receiver}}\xspace}
\newcommand{\trans}{\tau}
\newcommand{\exec}[1]{\ensuremath{\langle #1 \rangle}\xspace}
\newcommand{\view}{\ensuremath{V}\xspace}
\newcommand{\real}{\ensuremath{\mathit{real}}\xspace}
\newcommand{\ideal}{\ensuremath{\mathit{ideal}}\xspace}

\newcommand{\bit}[1]{\ensuremath{\{0,1\}^{#1}}\xspace}
\newcommand{\Ds}{\ensuremath{\mathbf{x}}\xspace} %
\newcommand{\ds}{\ensuremath{x}\xspace} %
\newcommand{\mech}{\ensuremath{\mathcal{M}}\xspace} %
\newcommand{\universe}{\ensuremath{\mathcal{X}}\xspace}
\newcommand{\Lap}{\mathsf{Lap}}
\newcommand{\contdp}{\ensuremath{\cont_{\mech}^{\mathsf{batch}}}}
\newcommand{\contpp}{\ensuremath{\cont_{\mech}^{\mathsf{pp}}}}

\newcommand{\datum}{x}
\newcommand{\data}{{\vec x}}

\newcommand{\univ}{\cZ}
\newcommand{\pred}{\psi}
\newcommand{\modelcoll}{\vec{\theta}}
\newcommand{\upreq}{\text{\sf UpReq}}

\newcommand{\supp}{\text{supp}}
\newcommand{\canon}[1]{\left[ {#1} \right]}

\newcommand{\IN}{\ensuremath{\mathtt{in}}\xspace}
\newcommand{\OUT}{\ensuremath{\mathtt{out}}\xspace}

\newcommand{\adt}{\ensuremath{\mathsf{ADT}}}
\newcommand{\impl}{\ensuremath{\mathsf{Impl}}\xspace}
\newcommand{\subscriptLogical}{\mathsf{adt}}
\newcommand{\subscriptPhysical}{}

\newcommand{\op}{\ensuremath{\mathsf{op}}\xspace}
\newcommand{\code}{\ensuremath{\mathsf{code}}\xspace}

\newcommand{\stateLog}{\ensuremath{{s}_{\subscriptLogical}}}

\newcommand{\statePhys}{\ensuremath{{s}_{\subscriptPhysical}}}

\newcommand{\outLog}{\ensuremath{\mathsf{out}_\subscriptLogical}}
\newcommand{\outPhys}{\ensuremath{\mathsf{out}}}
\newcommand{\logequiv}{\ensuremath{\stackrel{L}{\equiv}}\xspace}
\newcommand{\ahiJointDist}{\ensuremath{\mathfrak{D}}}

\newcommand{\side}{\ensuremath{\mathtt{side}}\xspace}

\newcommand{\chall}{\text{\sf challenge}\xspace}
\newcommand{\intru}{\text{\sf intrusion}\xspace}
\newcommand{\reg}{\text{\sf regular}}
\newcommand{\N}{\mathbb{N}}
\newcommand{\X}{\mathcal{X}}
\newcommand{\true}{\text{\sf true}\xspace}
\newcommand{\false}{\text{\sf false}\xspace}
\newcommand{\tick}{\text{\sf tick}\xspace}

\newcommand{\learn}{\mathsf{Learn}}
\newcommand{\unlearn}{\mathsf{Unlearn}}

\newcommand{\model}{\ensuremath{h}}

\title{Control, Confidentiality, and the Right to be Forgotten}

\author{Aloni Cohen\thanks{Department of Computer and Data Science, University of Chicago. \texttt{aloni@uchicago.edu}} \and Adam Smith\thanks{Department of Computer Science, Boston University. \texttt{ads22@bu.edu}}  \and Marika Swanberg\thanks{Department of Computer Science, Boston University. \texttt{marikas@bu.edu}}  \and Prashant Nalini Vasudevan\thanks{Department of Computer Science, National University of Singapore. \texttt{prashant@comp.nus.edu.sg}}}

\begin{document}
\maketitle

\begin{abstract}
    Recent digital rights frameworks give users the right to \emph{delete} their data from systems that store and process their personal information (e.g., the ``right to be forgotten'' in the GDPR).     
    
    How should deletion be formalized in complex systems that interact with many users and store derivative information?
    We argue that prior approaches fall short. Definitions of \emph{machine unlearning} \cite{cao2015towards} are too narrowly scoped and do not apply to general interactive settings. The natural approach of \emph{deletion-as-confidentiality} \cite{GGV20} is too restrictive: by requiring secrecy of deleted data, it rules out social functionalities.

    We propose a new formalism: \emph{deletion-as-control}. 
    It allows users' data to be freely used before deletion, while also imposing a meaningful requirement after deletion---thereby giving users more control.
    
    Deletion-as-control provides new ways of achieving deletion in diverse settings.
    We apply it to social functionalities, and give a new unified view of various machine unlearning definitions from the literature. This is done by way of a new adaptive generalization of {history independence}.

    Deletion-as-control also provides a new approach to the goal of \emph{machine unlearning}, that is, to maintaining a model while honoring users' deletion requests. We show that publishing a sequence of updated models that are differentially private under continual release satisfies deletion-as-control. 
    The accuracy of such 
    an algorithm  does not depend on the number of deleted points, in contrast to the machine unlearning literature.

\end{abstract}

\newpage 
{\footnotesize \tableofcontents}

\newpage

    \epigraph{Someday, this baby and other babies from her cohort will be 30 and there will be an absolutely bananas cache of data about what form and hue their poops took. Maybe someone will hack it and it will derail a presidential campaign news cycle.}{\textit{Alexandra Petri, from ``What I've been up to the last four months'' }} %

\section{Introduction}
The long-term storage of modern data collection carries serious risks, including often-surprising disclosures \citep{data-breaches,carlini2021extracting,JASONreport,Cohen_Nissim_2020,DHS-location-data}, manipulation, and epistemic bubbles.
The permanence of our digital footprints can also chill expression, with every word weighed against the risk of out-of-context blowback in the future.

Data protection laws around the world have begun to challenge this permanence. {The} EU's General Data Protection Regulation provides an individual \emph{data subject} the right to request ``the erasure of personal data concerning him or her'' and delineates when a \emph{data controller} must oblige.  
California followed suit in 2020, and similar rights take effect in Virginia, Colorado, Connecticut, and Utah in 2023~\cite{statePrivLaws}.

In the modern data ecosystem, however, it is not easy to articulate what constitutes the ``erasure'' of personal data. Data is not merely stored in databases---it is used to train machine learning models, compute and publish statistics, and drive decisions.
Such complexity and nuance challenges simplistic thinking about erasure, and the sheer number of ways data are used precludes case-by-case reasoning about erasure compliance. 

Giving users more control over data is today a central policy goal.
Decades of cryptography has given us good definitional tools for reasoning about non-disclosure of data\nt{---}enabling the development of technical solutions, informing policy decisions, and influencing practice. But we lack similar tools for reasoning about control over data and, in particular, deletion. While there has been a flurry of recent work on so-called \emph{machine unlearning}~\MULcites, often directly motivated by legal compliance,  there remain basic gaps in our understanding.

This paper sheds light on what data deletion means in complex data processing scenarios, and how to achieve it.
We provide a formulation that unifies and generalizes previous technical approaches that only partially answer this question.  Though we make no attempt to strictly adhere to any specific legal right to erasure, we aim to incorporate more of its contours than prior technical work on erasure.

Our new formulation, called \textit{deletion-as-control}, requires that after an individual Alice requests erasure, the data controller's future behavior and internal state should not depend on Alice's data \emph{except insofar as that data has already affected other parties in the world}. In this way, Alice's autonomy need not require secrecy; she has a say about how her data is used regardless of past or future disclosure.

Our definition meaningfully captures a variety of data controllers, including ones facilitating social interactions and maintaining accurate predictive models.
In contrast, prior approaches yield a patchwork of, at times, contradictory and counterintuitive interpretations of erasure.

\subsection{Touchstone Examples}
\label{sec:touchstone}

In order to understand our and prior approaches to data deletion, it is helpful to have in mind a few concrete examples of functionalities that separate our approaches from prior ones.
We describe four touchstone functionalities, then briefly discuss how they relate to prior approaches and our new notion.

    \begin{description}[leftmargin=12pt]

    \item [Private Cloud Storage.] 
    Users can upload files for cloud storage and future download. Only the originating user may download a file and the files are never used in any other way. The existence of files is only ever made known to the originating user, and the controller publishes no other information.

    \item [Public Bulletin Board.] Users can submit posts to the public bulletin board. The bulletin board simply displays all user posts currently in the controller's internal storage, with no other functionalities (e.g. responding, messaging, etc.).

    \item [Batch Machine Learning.] Users contribute data during some collection period. At the end of the period, the data controller trains a predictive model on the resulting dataset. The data controller then publishes the model.

    \item[Public Directory + Usage Statistics.]
    Users upload their name and phone number to be listed in a directory.
    The data controller allows anybody to search for a listing in the directory, and each week reports a 
    count of the number of distinct users that have looked up a phone number so far. (Other statistics are possible too---the weekly count of new users, say.) 
    
\end{description}

\para{The touchstone examples and prior approaches}
Figure~\ref{fig:touchstone_prior_defs} summarizes how the touchstone functionalities fare under deletion-as-control and under three prior approaches to defining deletion: \textit{deletion-as-confidentiality}, \textit{machine unlearning}, and \textit{simulatable deletion} (Section~\ref{sec:prior-work}).\footnote{We introduce the terms \emph{deletion-as-confidentiality} and \emph{simulatable deletion} for the definitions of  \cite{GGV20} and \cite{godin2021deletion}, respectively, to more clearly distinguish them from each other and from deletion-as-control.  \cite{GGV20} use the term \emph{deletion-compliance}; \cite{godin2021deletion} use \emph{strong deletion-compliance} and \emph{weak deletion-compliance}, respectively.}
Together, the touchstone functionalities illustrate that prior approaches constitute an inconsistent patchwork, each falling short on at least two of the examples.

Deletion-as-confidentiality \citep{GGV20} is over-restrictive. Briefly, it requires that third parties cannot distinguish whether a data subject Alice requested erasure from the controller or simply never interacted with the controller in the first place. This implies, among other things, that Alice's data is kept confidential from all other parties even if Alice never requests its erasure. This confidentiality-style approach is well-suited for Private Cloud Storage, but deletion-as-confidentiality precludes inherently social functionalities, like the Bulletin Board and Directory. No controller implementing these functionalities could ever satisfy deletion-as-confidentiality. Why would Alice post messages if they could never be made public?

On the other hand, machine unlearning---even in its strongest incarnation, due to \citet{gupta2021adaptive}---is too narrowly scoped. It is specialized to the setting of machine learning and does not consider general interactive functionalities. The definition is not applicable to  the Cloud Storage, Bulletin Board, and Directory functionalities. Definitions from the machine unlearning literature are meaningful for the Batch Machine Learning functionality, where they correspond to versions of \emph{history independence}.

Even where multiple definitions are meaningful, they may impose different requirements. For example, both deletion-as-confidentiality and deletion-as-control admit implementations of Batch Machine Learning that are \emph{persistent} in that the published models never need to be updated. On the other hand, history independence requires that any useful model be
updated after enough deletion requests.

\para{The touchstone examples and deletion-as-control}
We show that each of the touchstones can be implemented in a manner that satisfies our new notion, deletion-as-control.

    \begin{description}[leftmargin=12pt]

    \item[Private Cloud Storage.] To remove a user, the controller deletes all the user's files from its internal storage. Such a controller satisfies deletion-as-control if its data structures are \textit{history independent} (Corollary~\ref{cor:priv_cloud_pub_bboard}).

    \item [Public Bulletin Board.]{To remove a user, the controller deletes all of the user's posts from its internal storage and, as a result, from the public-facing bulletin board.} As with cloud storage, such a controller satisfies deletion-as-control if its data structures are history independent (Corollary~\ref{cor:priv_cloud_pub_bboard}).

    \item [Batch Machine Learning.] Deletion-as-control is achieved if the dataset is deleted after training and training is done with differential privacy, e.g.,  using DP-SGD \cite{bassily2014private} (Corollary~\ref{cor:dp_ml}).  To remove a user, the controller does nothing---it simply ignores deletion requests. The resulting  deletion guarantee is parameterized by the privacy parameters $\eps$ and $\delta$. 

    \item[Public Directory + Usage Statistics.] Deletion-as-control can be achieved by combining differential privacy and history independence (Corollary~\ref{cor:directory_DP_Stats}). The statistics are computed using a mechanism that satisfies a stringent form of DP---pan-privacy under continual release---while the public directly is implemented using a history independent data structure.
    To remove a user, the controller {deletes their listing from the public directory and its associated data structures,} but leaves the data structures for the DP statistics unaltered.
\end{description}

\begin{figure*}
    \centering
    \footnotesize
    \begin{tabular}{L{4.5cm}  C{2.5cm} C{2.5cm} C{2.5cm} C{2.5cm}}
         \toprule
          & Private Cloud Storage  & Public Bulletin Board &  Batch Machine Learning  & Public Directory + Usage Stats   \\
         \midrule[\heavyrulewidth]
          Deletion-as-Confidentiality \citep{GGV20} & {\Large\greencheck}  &  {\Large\overrestrictive} & {\Large\greencheck}& {\Large\overrestrictive}\\
          \midrule
          Machine unlearning \newline\citep[\dots]{gupta2021adaptive}& {\Large\na} & {\Large\na} & {\Large\greencheck} & {\Large\na}  \\
        
          \midrule
         Simulatable deletion  \newline \citep{godin2021deletion}&  {\Large\greencheck} & {\Large\underrestrictive} & {\Large\underrestrictive}& {\Large\underrestrictive} \\
         \midrule
         
         Deletion-as-control  & 
         {\Large\greencheck}  & 
         
         {\Large\greencheck} &          %
         {\Large\greencheck}  &         %
         {\Large\greencheck}  \\ 
        (this work) 
          & (Corollary \ref{cor:priv_cloud_pub_bboard}) & (Corollary \ref{cor:priv_cloud_pub_bboard}) & (Corollary \ref{cor:dp_ml}) & (Corollary \ref{cor:directory_DP_Stats})\\
         \bottomrule
    \end{tabular}
     \begin{tablenotes}
        \item \greencheck: Definition is satisfied by implementations with meaningful deletion guarantees.%
        \item \underrestrictive: Definition is under-restrictive: allows \emph{vacuous} implementations with no meaningful deletion of any kind.
        \item \overrestrictive: Definition is over-restrictive: no implementation of the functionality satisfies the definition.
        \item \na: Definition does not apply to the functionality.
    \end{tablenotes}
    \caption{Application of deletion definitions to touchstone functionalities.}
    \label{fig:touchstone_prior_defs}
\end{figure*}

\subsection{Contributions}

    \para{Defining deletion-as-control}
    Our primary contribution is a formalization of \emph{deletion-as-control}, an important step towards providing individuals greater control over the use of personal data.
    The new notion applies to general data controllers and interaction patterns among parties, building on the modeling of \citep{GGV20}. 
    As described below, it unifies existing approaches within a coherent framework and captures all the touchstone examples of Section~\ref{sec:touchstone}.
 
    The goal embodied by deletion-as-control is not so much to hide data from others as to exercise control over how the data is used. Until Alice requests erasure, deletion-as-control should not limit the controller's usage of her data. But after erasure, the data controller's future behavior and internal state should not depend on Alice's data \emph{except insofar as that data has already affected other parties in the world}. 
    In this way, Alice's autonomy need not imply secrecy; she has a say about how her data is used regardless of past or future disclosure. 

\bigskip\noindent
Our approach provides new ways of achieving meaningful deletion in diverse settings.

\para{Capturing social functionalities via history independence}
    Deletion-as-control applies to a wide range of controllers that provide ``social'' functionalities {where prior approaches fall flat (e.g., the Public Bulletin Board touchstone)}.
    {Along the way, we give a new unified view of the various machine unlearning definitions from the literature.}
    
    {Both flow from a theorem roughly stating that deletion-as-control is implied by} \emph{adaptive history independence}, a generalization of the cryptographic notion of \textit{history independence}~\citep{micciancio1997oblivious, naor2001anti} that we introduce. 
    An implementation of a data structure is history independent if its memory representation reveals nothing more than the logical state of the data structure.
    That history independence is related to deletion is intuitive, and appears in \citep{GGV20, godin2021deletion}.
    Machine unlearning imposes a similar requirement in the specific context of machine learning. Oversimplifying, a learned model (akin to the memory representation) must reveal nothing more than a model retrained from scratch (akin to the logical state).
    {We make these connections precise.}

    \para{New algorithms for machine learning via differential privacy}
Deletion-as-control provides a new approach for machine learning in the face of modern data rights. Very roughly, differential privacy (DP) provides deletion-as-control for free. Intuitively, if a person has (approximately) no impact on a trained model, mitigating that impact is trivial.
In particular, if using an adaptive pan-private algorithm to maintain the model, it does not need to be updated in response to deletion requests, unlike machine unlearning algorithms. 
For the first time, this approach enables a meaningful deletion guarantee while bounding the worst-case loss compared to deletion-free learning.

Specifically, we describe two ways of compiling DP mechanisms into  controllers satisfying
deletion-as-control. 
The first applies to DP mechanisms that are run in a batch setting on a single, centralized dataset  (e.g., the Batch Machine Learning touchstone). The second applies to mechanisms satisfying an adaptive variant of \emph{pan-privacy under continual release} \citep{chan2011private,dwork2010pan,jain2022price}, including controllers that periodically update a model on an ongoing basis. 
The compilation from differential privacy is by way of deletion-as-confidentiality \citet{GGV20}, which we prove implies deletion-as-control.

We combine this result with existing algorithms for private learning under continual release \citep{KairouzM00TX21} to obtain new controllers that maintain a model with accuracy essentially identical to that of a model trained on the entire set of added records. 
    
    \para{Capturing complex mechanisms via composition} 
    We show that deletion-as-control captures more functionalities than those collectively captured by history independence, differential privacy, and deletion-as-confidentiality. Specifically, we show how to implement the Public Directory + Usage Statistics touchstone, a functionality that cannot satisfy any of the above three properties. To do so, we prove that deletion-as-control enjoys a limited form of \emph{parallel composition}.

\begin{figure}   
    \centering
    \begin{minipage}[c]{0.55\textwidth}
\includegraphics[width=\textwidth]{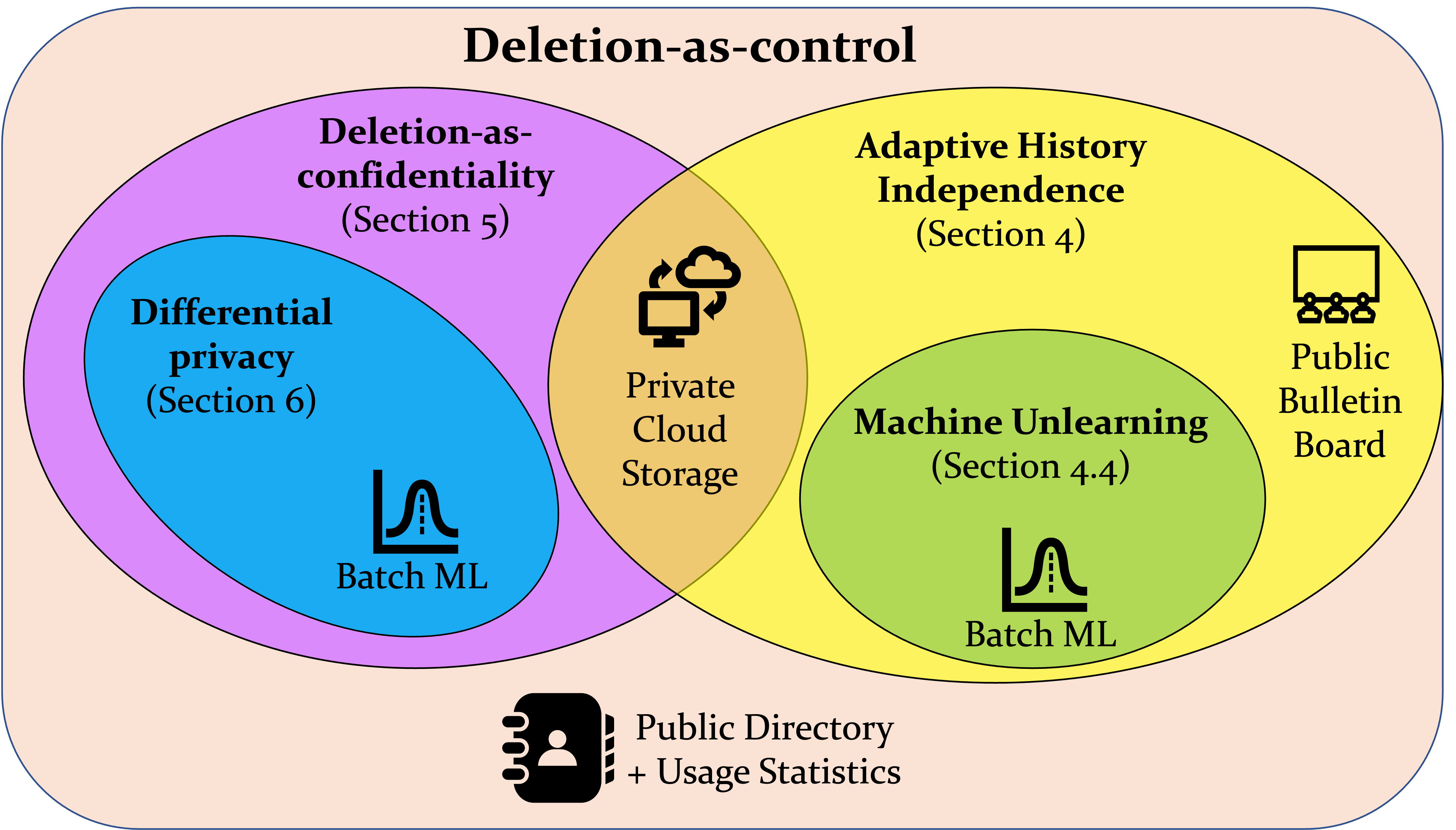}    
\end{minipage} \hfill
\begin{minipage}[c]{0.4\textwidth}
\caption{\small Overview of the relationships between definitions and the roles of specific example controllers. Each oval represents a definition, and each icon represents an implementation of a touchstone functionality satisfying the corresponding definition. 
    Oval containment corresponds to an implication between the definitions, though only  roughly and subject to technicalities that we do not attempt to capture in this figure.}
    \label{fig:overview}
\end{minipage}
\hfill

\end{figure}

\subsection{Defining deletion-as-control}
We define deletion-as-control in a way that allows arbitrary use of a person's data before deletion, but not after. 
The challenge is to provide meaningful privacy guarantees despite this feature, even in a general setting with an adaptive and randomized data controller, data subject (Alice), and an environment (representing all parties other than the data controller and Alice). 

Our proposal is that after deletion, the data controller should be able to produce a plausible alternate explanation for its current state without appealing to Alice's participation.
More specifically, \emph{the state of the controller after deletion can be plausibly attributed to the interaction between the controller and the environment alone}. By this we mean that the state is about as likely---with probability taken over the controller's random coins---in the real world as in the hypothetical `ideal' world where the environment's messages to the controller are unchanged but where Alice does not interact with the controller whatsoever. The result is that, after Alice's deletion, the controller's subsequent states depend on the interaction with Alice only insofar as that interaction affected other parties' interactions with the controller. 

Importantly, we only require that the controller's state is plausible in the ideal world \emph{given} the messages sent by the environment. \emph{We do not require the environment's messages themselves to be plausible in the ideal world.}
For example, suppose that on a public bulletin board, Bob simply copies and reposts Alice's posts. Bob's messages in the ideal world would still contain the content of Alice's posts, even though Alice is completely absent in the hypothetical ideal world. 
This is unavoidable---we want to allow the controller and the environment to use the subject's data arbitrarily before deletion. So the environment's queries may depend on the data subject's inputs, directly or indirectly.

In a bit more detail, our definition compares real and ideal worlds defined in a non-standard way.%
        \footnote{The non-real world is more `hypothetical' or `counter-factual' than `ideal.' Regardless, we use the term `ideal world' for continuity with \citep{GGV20} and decades of cryptography.}
The real world execution, denoted $\exec{\cont,\env,\subj}$, involves three parties: a data controller \cont, a special data subject \subj, and an environment \env representing all other parties. (The notation $\exec{\cdots}$ denotes the transcript of an execution between the parties listed.)
The execution ends after $\subj$ requests deletion and $\cont$ processes the request.
The real execution specifies (i) \cont's state $\s_\cont$, (ii) the \emph{queries} $\queries_\env$ sent by $\env$ to $\cont$, and (iii) the randomness $\Rand_\cont$ used by $\cont$. 
The ideal world execution, denoted  $\exec{\cont(\Rand'_\cont), \dummy(\queries_\env)}$, involves the  same controller \cont and a dummy environment \dummy that simply replays the queries  $\queries_\env$. 
The controller's ideal world randomness 
 $\Rand'_\cont$ is sampled by a simulator $\Sim(\queries_\env, \Rand_\cont, \s_\cont)$.
The ideal execution specifies (i) \cont's ideal state $\s'_\cont$, and (ii) the simulated randomness $\Rand'_\cont$ used by $\cont$.

\begin{definition}[$(\eps,\delta)$-deletion-as-Control (simplified)]\label{def:TPDP-del-as-cont}
Given $\cont$, $\env$, $\subj$, and $\Sim$, consider the following  experiment. Sample $\Rand_\cont \gets \bit{\bbN}$; run the real execution $(\queries_\env, \s_\cont) \gets \exec{\cont(\Rand_\cont), E, \subj}$;  sample $\Rand'_\cont \gets \Sim(\queries_\env, \Rand_\cont, \s_\cont)$; and run the ideal execution $(\queries'_\env,\s'_\cont) \gets \langle \cont(\Rand'_\cont), \dummy(\queries_\env)\rangle$.

We say a controller $\cont$ is \emph{$(\eps, \delta)$-deletion-as-control compliant}
if
there exists $\Sim$ such that
for all $\env$ and $\subj$: (i) $\Rand'_\cont \approxequiv{\eps,\delta} \Rand_\cont$; (ii)  $\s'_\cont = \s_\cont$ with probability at least $1-\delta$. 
\end{definition}

The notation $\approxequiv{\eps,\delta}$ denotes approximate indistinguishability parameterized by $\epsilon$ and $\delta$ (as in the definition of differential privacy). We give a complete version of Definition~\ref{def:TPDP-del-as-cont} in Section~\ref{sec:del-as-cont}.

\subsection{Prior Work}
\label{sec:prior-work}
We give a brief discussion of prior definitions of deletion. Machine unlearning and deletion-as-confidentiality are discussed in detail in Sections \ref{sec:MUL-2-AHI} and \ref{sec:ggv}, respectively.

\para{Deletion-as-confidentiality} \citet{GGV20} define \emph{deletion-as-confidentiality}.\footnote{Deletion-as-confidentiality is called \emph{deletion-compliance} in \cite{GGV20} and \emph{strong deletion-compliance} in \cite{godin2021deletion}.}
It requires that the deleted data subject Alice leaves (approximately) no trace: the whole view of the environment along with the state of the controller after deletion should be as if Alice never existed. 
As a result, no third party may ever learn of Alice's presence --- even if she never requests deletion.
The strength of this definition is its strong, intuitive, interpretable guarantee. 

But deletion-as-confidentiality is too restrictive (Figure~\ref{fig:touchstone_prior_defs}). The stringent indistinguishability requirement precludes any functionality where users learn about each other.
Implementations of the Private Cloud Storage and Batch Machine Learning functionalities can satisfy deletion-as-confidentiality, using history independence (cf. Section~\ref{sec:HI}). But the Bulletin Board and Directory are ruled out. If Bob ever looks up Alice's messages, the confidentiality required by the definition is impossible.

\para{Simulatable deletion} \citet{godin2021deletion} introduce \emph{simulatable deletion} as a relaxation of deletion-as-confidentiality, motivated by the observation that deletion-as-confidentiality rules out social functionalities.\footnote{Simulatable deletion is called \emph{weak deletion-compliance} in \cite{godin2021deletion}.}
Roughly, simulatable deletion requires that after a data subject Alice is deleted, the resulting state of the controller is simulatable given the environment's view. This means that any information about Alice that is present in the controller's state is already present in the view of other parties. 

Simulatable deletion is too permissive. A controller may indefinitely retain any information that has ever been shared with any third party.
For example, the Public Bulletin Board need not delete Alice's posts if they have ever been read!\footnotemark~
As a result, simulatable deletion is essentially vacuous for functionalities where the controller's state need not be kept secret. The controller can simply publish its state, making simulation trivial.
Turning to our touchstone examples (Section~\ref{sec:touchstone}), while simulatable deletion imposes a meaningful requirement for the Private Cloud Storage functionality, it allows implementations with no meaningful deletion of any kind for the other three functionalities. 

    \footnotetext{\citet{godin2021deletion} actually consider a version of the Bulletin Board functionality for which simulatable deletion is meaningful. Crucially, their functionality does not track whether a post has been read. Hence their controller must actually delete Alice's posts. But if the bulletin board keeps read receipts or actively pushes new messages out to users, say, it would not have to delete these posts.}

\para{Machine unlearning} This recent line of work specializes the question of deletion to the setting of machine learning~\MULcites.
Given a model $\modelcoll \gets \learn(\data)$ and a data point $\datum^*\in\data$, that literature requires sampling a new model $\theta'\gets \unlearn(\modelcoll, \datum^*, \data)$ approximately from the distribution $\learn(\data\setminus\{x^*\})$ (i.e., approximating retraining from scratch). As we explain in Section~\ref{sec:MUL-2-AHI}, these definitions correspond to versions of \emph{history independence}.

A drawback of machine unlearning is that it specialized to the setting of machine learning. It does not apply to general data controllers and interaction patterns among parties, including the Cloud Storage, Bulletin Board, and Directory functionalities (Figure~\ref{fig:touchstone_prior_defs}).

History independence-style definitions of machine unlearning do impose a meaningful requirement for the Batch Machine Learning functionality. In fact history independence is a conceptually stricter requirement than deletion-as-control, setting aside many technicalities (Section~\ref{sec:MUL-2-AHI}). Differentially private (DP) learning illustrates the difference. Roughly, DP learning provides deletion-as-control for free; the resulting model can be published once and never  updated.
In contrast, history independence \textit{require} updating the model when there are many deletions.
To see why, consider $\data$ of size $n$ and suppose all $n$ people request deletion. These definitions would require the final model $\theta^*$ be essentially trivial---it should perform about as well as a the model $\theta_0$ trained on an empty dataset. With the DP learner described above, $\theta^*$ performs just as well as the initial model $\theta$.

\subsection{Paper Structure}
In Section~\ref{sec:del-as-cont} we define deletion-as-control. In Section~\ref{sec:HI} we define Adaptive History Independence and prove that it implies deletion-as-control (Theorem~\ref{thm:AHI_to_del}). 
We also show that the machine unlearning definition of \cite{gupta2021adaptive} is a special case of Adaptive History Independence (Proposition~\ref{prop:AHI_MUL}). 
In Section~\ref{sec:ggv} we show that deletion-as-confidentiality is a strengthening of our definition (Theorem~\ref{thm:conf-cont}). In Section~\ref{sec:DP}, we relate differential privacy to our definition  (Theorem~\ref{thm:PPshi_filter}) and in the process we define adaptive pan-private in the continual release model. Lastly, in Section~\ref{sec:simul-comp} we prove a narrow composition result for our definition.  %

\section{Deletion-as-control}
\label{sec:del-as-cont}

We define deletion-as-control in a way that allows arbitrary use of a person's data before deletion, but not after.
Under such a definition, an adversary might completely learn the data before it is deleted, and even make it available after it is deleted!
The challenge is to provide a meaningful guarantee despite this limitation, even in a general setting with adaptive and randomized data controllers, data subjects, and environments (representing all parties other than the data controller and distinguished data subject).

\label{sec:prelims}
\subsection{$(\eps,\delta)$-indistinguishability}

We consider a notion of similarity of distributions closely related to differential privacy (Appendix~\ref{app:dp-prelims}).

We present some further technical tools in the appendices: a novel coupling lemma for this notion of indistinguishability (Appendix~\ref{app:coupling-lemma}) as well as background  on differential privacy (Appendix~\ref{app:dp-prelims}).

\begin{definition}
Given parameter $\eps\geq 0$ and $\delta\in [0,1)$, we say two probability distributions $P$ and $Q$ on the same set $\cX$ (with the same $\sigma$-algebra of events $\Sigma_\cX$) are \emph{$(\eps,\delta)$-indistinguishable} and write $P\approxequiv{\eps, \delta} Q$ if, for every event $E \in \Sigma_\cX$, 
 \begin{equation*}
 P(E)\leq e^\eps Q(E)+\delta\   \text{and}\   Q(E)\leq e^\eps P(E) + \delta \, .
 \end{equation*}
 Slightly overloading this notation, we say two random variables $X$ and $Y$ taking values in the same measurable space $(\cX,\Sigma_\cX$) are \emph{$(\eps, \delta)$-indistinguishable}, denoted $X \approxequiv{\eps, \delta} Y$ if, for every event $E \in \Sigma_\cX$, 
 \begin{align*}
 &\Pr[X \in E]\leq e^\eps \Pr[Y \in E] +\delta\   \text{and,}\\
 &\Pr[Y \in E]\leq e^\eps \Pr[X \in E] + \delta \, .
 \end{align*}
\end{definition}

Because algorithms in our model run in unbounded time, their (countably infinite) random tapes belong to an uncountably infinite set. This means that not all sets of random tapes have well-defined probability. The $\sigma$-algebra $\Sigma_\cX$ captures the set of events $E$ for which $P(E)$ and $Q(E)$ are defined. This issue does not affect most proofs and definitions; we only make the $\sigma$-algebra of events explicit when necessary.

In our case, the set $\cX$ of random tapes for a single machine is $\bit{\bbN}$. Any execution that terminates reads only a finite prefix of the tape, and so the natural $\sigma$-algebra $\Sigma_\cX$ is the smallest one containing the sets $E_w = \left\{w \| x: x\in \bit{\bbN}\right\}$ for all $w\in \bit{*}$, where $\|$ denotes string concatenation (that is, $E_w$ is the set of infinite tapes with a particular finite prefix $w$). $\Sigma_\cX$ contains every event $E$ that depends on only finitely many bits of the tape. This is the standard $\sigma$-algebra for an infinite set of fair coin flips~(see, e.g., \cite{sigmaalgebras}).

\subsection{Parties and simplified execution model}
Our definition uses the real/ideal cryptography, though not with a typical indistinguishability criterion. See Section~\ref{sec:execution-model-detailed} for complete details.

The real world execution, denoted $\exec{\cont,\env,\subj}$, involves three (possibly randomized) parties: a data controller \cont, an environment \env, and a special data subject \subj. The real interaction is arbitrary, limited only by the execution model described below. Parties have authenticated channels over which they may interact freely. While $\subj$ has only a single channel to $\cont$, the environment has an unbounded number of channels (representing unbounded additional parties). While these channels are authenticated, the controller cannot distinguish the single channel to $\subj$ from those to $\env$. The interaction continues until the data subject $\subj$ requests deletion, ending after the data controller $\cont$ processes the request. We denote by $\s_\cont$ the final state of the controller.

The ideal world execution, denoted $\exec{\cont,\dummy}$, involves the interaction of the same controller \cont as well as a dummy environment \dummy. $\dummy$ takes as input the transcript from the real execution, denoted $\tau$, and simply replays only $\env$'s queries, denoted $\queries_\env$. 
If $\queries_\env$ is empty, then $D$  terminates without sending messages.
Observe that $\cont$'s responses and state in the ideal world are not fixed. They depend on \cont's ideal-world randomness, denoted $\Rand'_\cont$. Moreover, the ideal interaction is defined relative to a particular instantiation of the real world interaction. In particular, the queries $\queries_\env$ may depend on \cont's real-world randomness, denoted $\Rand_\cont$.

The controller's real-world randomness $\Rand_\cont$ consists of infinitely-many random bits sampled uniformly at random from $\bit{\bbN}$. We denote this distribution $\distrand$. 
In the ideal world, a simulator \Sim takes as input $(\queries_\env, \Rand_\cont, \s_\cont)$ and generates \cont's ideal-world randomness $\Rand'_\cont$. When we wish to emphasize the controller's randomness in the execution, we write $\exec{\cont(\Rand_\cont),\env,\subj}$ and $\exec{\cont(\Rand'_\cont), \dummy}$.

An execution involves parties sending messages to each other until some termination condition is reached.
Starting with $\env$ (real) or \dummy (ideal), parties get activated when they receive a message, and deactivated when they send a message. Only a single party is active at a time.
Parties communicate over authenticated channels. Because $\env$ represents all users besides the distinguished data subject $\subj$, $\env$ has many distinct channels to $\cont$.
Importantly, authentication allows parties to know on which channel a message was received, but not which party (i.e., $\env$ or $\subj$) is on the other end of that channel.
Each party is initialized with a uniform random tape which may only be read \emph{once} over the course of the whole execution. If a party wishes to re-use bits from its randomness tape, it must store them in its internal $\s$.

The real execution $\exec{\cont,\env,\subj}$ ends when $\subj$ requests deletion from $\cont$. The data subject's \delete message activates the controller, who can then remove \subj's data.
The ideal execution $\exec{\cont(\Rand'_\cont), \dummy(\queries_\env)}$ ends after \dummy sends its last query from $\queries_\env$ to \cont.
In both cases, the execution ends after the final activation of \cont. We consider the controller's state at the end of the execution: $\s_\cont = \cont(\queries_\env;\Rand_\cont)$ in the real world, and $\s'_\cont = \cont(\queries_\env;\Rand'_\cont)$ in the ideal world. If an execution never ends, the state is defined to be $\bot$ and the transcript $\tau$ and its subset $\queries_\env$ are defined to be empty. For example, the real execution ends if and only if $\subj$ requests deletion.

\begin{remark}[Keeping time]\label{remark:env-ticks}
In Section~\ref{sec:DP}, we need a global clock. Balancing modelling simplicity with generality, we allow the environment to control time. Specifically, we introduce a special query $\tick$ that $\env$ can send to $\cont$ thereby incrementing the clock. We do not allow $\subj$ to query $\tick$.
\end{remark}

\subsection{Defining Deletion-as-Control}

We require that the internal state of the controller is about as likely in the real world and the ideal world, where probability is taken over \cont's random coins. Let $\s_\cont$ and $\s'_\cont$ be the internal states in the real and ideal executions. Consider a random variable $\Rand'_\cont$ which is sampled uniformly conditioned on $\s'_\cont = \s_\cont$ in the ideal execution where $\cont$ uses randomness $\Rand'_\cont$.
Informally, our definition requires that the distributions of $\Rand'_\cont$ and $\Rand_\cont$ are close: $\Rand'_\cont \approxequiv{\eps,\delta} \Rand_\cont$.
We do not require that the real and ideal executions are themselves $(\eps,\delta)$-indistinguishable. Instead, we require that the ``explanations'' in the real and ideal executions are $(\eps,\delta)$-indistinguishable---viewing the controller's randomness as the explanation of its state (relative to the environment's queries $\queries_\env$).

We extend this idea by considering ways of sampling $\Rand'_\cont$ other than the conditional distribution described above (which may not always be defined). In general, we allow a simulator $\Sim$ to sample $\Rand'_\cont$ as a function of the queries $\queries_\env$ from $\env$ to $\cont$, the real-world randomness $\Rand_\cont$, and the real-world state $\s_\cont$. (Although we view the simulator's output as an infinite-length bit sequence,  it actually only needs to output a finite prefix.) We require that $\Rand'_\cont \approxequiv{\eps,\delta} \Rand_\cont$ and that $\s'_\cont = \s_\cont$ (or $\s_\cont = \bot$) except with probability $\delta$. 
Sampling $\Rand'_\cont$ conditioned on $\s'_\cont = \s_\cont$ is a useful default simulation strategy that we use throughout the paper,
but there are sometimes much simpler ways to sample $\Rand'_\cont$.
\begin{figure*} \label{fig:real_ideal_diagrams}
    \centering
    \includegraphics[scale=0.3]{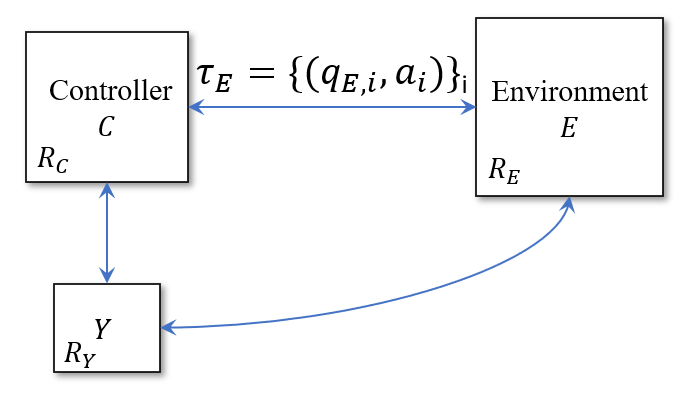}
    \hspace{0.2in}\rule{0.5pt}{1.7in}\hspace{0.2in}
    \includegraphics[scale=0.3]{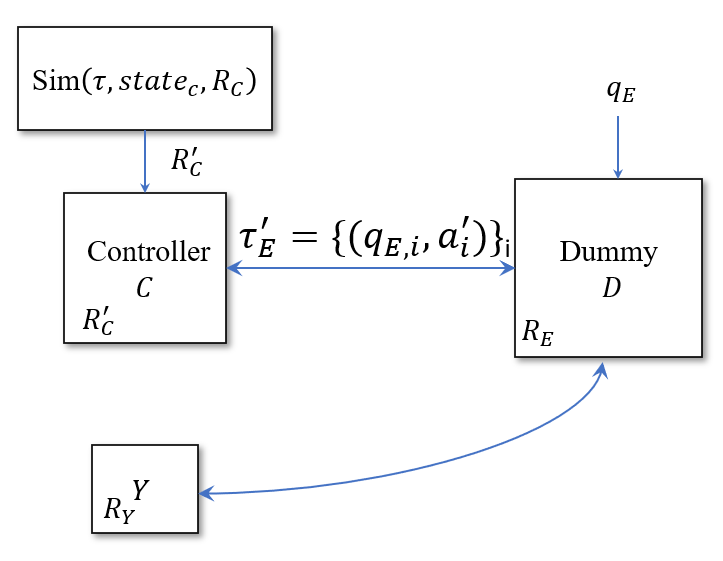}

    \caption{The real (left) and ideal (right) executions. The communications with the controller \cont are meant to be ``public API'' calls (e.g. accessing a public website, possibly using login credentials). Notice that the dummy environment in the ideal world sends message $\queries_{\env, i}$ regardless of what the previous responses were.}
\end{figure*}

\begin{definition}[Deletion as control]\label{def:del-as-cont}
Given a controller $\cont$, an environment $\env$, a data subject $\subj$ and a simulator $\Sim$, we consider the following experiment: 
\begin{itemize}[itemsep=0.3pt, topsep=4pt, partopsep=2pt]
    \item $\Rand_\cont \gets \distrand$ ($\Rand_\cont$ is a uniform random tape)
    \item $(\trans, \s_\cont) \gets \exec{\cont(\Rand_\cont), E, \subj}$,  where $\queries_\env \subseteq \trans$ are the messages from $\env$ to $\cont$. 
    \item $\Rand'_\cont \gets \Sim(\queries_\env, \Rand_\cont, \s_\cont)$
    \item $(\trans',\s'_\cont) \gets \langle \cont(\Rand'_\cont), \dummy(\queries_\env)\rangle$.
\end{itemize}
We say a controller $\cont$ is \emph{$(\eps, \delta)$-deletion-as-control compliant}
if
there exists $\Sim$ such that
for all $\env$ and $\subj$ and for $\Rand'_\cont$, $\Rand_\cont$ sampled as above:
\begin{enumerate}[itemsep=0.3pt, topsep=4pt, partopsep=2pt]
    \item $\Rand'_\cont \approxequiv{\eps,\delta} \Rand_\cont$ (i.e., $\Rand_\cont$ and $\Rand_\cont'$ are similarly distributed), \\ and \label{enum:del_cond_1}
    \item With probability at least $1-\delta$,  either\\ $\s'_\cont = \s_\cont$ or $\s_\cont = \bot$. \label{enum:del_cond_2}
\end{enumerate}
For a particular class of data subjects $\mathcal{Y}$, we say a controller $\cont$ is \emph{$(\eps, \delta)$-deletion-as-control compliant for $\mathcal{Y}$} if there exists \Sim such that the above holds for all \env and for all $\subj \in \mathcal{Y}$.
\end{definition}

\begin{example}[XOR Controller]
\label{example:XOR}
\newcommand{\opRead}{\mathsf{read}}
\newcommand{\opWrite}{\mathsf{write}}
Consider a controller $\cont_\oplus$ which maintains a $k$-bit state $\s\in\{0,1\}^k$ which is initialized uniformly at random. Upon receiving a message $x \in \{0,1\}^k$, $\cont_\oplus$ updates its state $\s\gets \s \oplus x$, sending nothing in return. If it receives any other message, including $\delete$, it does nothing. 

$\cont_\oplus$ satisfies $(0,0)$-deletion-as-control. To see why, consider an execution of the real world, which ends when the data subject sends $\delete$. At this point, $\s = R \oplus x_\subj \oplus x_\env$, where $R\in \{0,1\}^k$ is the random initialization, $x_\subj$ is the XOR of all messages $x$ sent by $\subj$, and $x_\env$ is the XOR of all messages $x$ sent by $\env$. 
Let $\Sim$ compute $x_\env$ from the queries $\queries_\env$, and output $\Rand' = \s\oplus x_\env = \Rand\oplus x_\subj$. This satisfies the definition: (1) $\Rand'$ is uniformly distributed because $x_\subj$ is independent of $\Rand$; (2) $\s' = \Rand' \oplus x_\env  = \s$.
\end{example}

\begin{definition}[Default simulator]
\label{def:default-sim}
The default simulator $\Sim^*$ samples $\Rand'_\cont$ as follows:
\begin{align*}
    &\Sim(\queries_\env, \Rand_\cont, \s_\cont):\\
    &\qquad\text{Return }\Rand'_\cont \sim \distrand\big|_{\cont(\queries_\env;\Rand') = \s_\cont} \text{if
    such an $\Rand'$ exists;}\\
    &\qquad\text{Otherwise, return $\Rand'_\cont = \Rand_\cont$}.
\end{align*}
\end{definition}

\noindent
As we show next, \env and \subj can be assumed to be deterministic without loss of generality.

\begin{lemma}[Deterministic environments and subjects]\label{lem:derandomize}
Consider a controller \cont. Suppose that for some $\eps,\delta \geq 0$, $\cont$ satisfies Definition~\ref{def:del-as-cont} for all \textit{deterministic} environments $\env$ and data subjects $\subj$ with simulator $\Sim$. Then $\cont$ satisfies Definition~\ref{def:del-as-cont} (even for randomized $\env$ and $\subj$) with the same simulator $\Sim$.
\end{lemma}

\begin{proof}
Fix  a controller $\cont$, and let $\Sim$ be the simulator which show that  $\cont$ satisfies Definition~\ref{def:del-as-cont} for all deterministic environments and data subjects. Now consider a pair of randomized ITMs $\env$ and $\subj$ with random coins denoted by $\Rand_\env$ and $\Rand_\subj$.

For a fixed string $\rand$, let $E_{\rand}$ denote the deterministic environment in which $\env$'s random tape is fixed to $\Rand_\env = \rand$; similarly for $\subj_{\rand}$. Given strings $\rand_\env$ and $\rand_\subj$, we can consider the execution of the deletion game in Definition~\ref{def:del-as-cont} with the deterministic machines $\env_{\rand_\env}$ and $\subj_{\rand_\subj}$. Let $\Rand'_{\cont, \rand_\env, \rand_\subj}$ denote the random coins output by the simulator in that game, and let $P_{\rand_\env, \rand_\subj}$ denote the distribution of $\Rand'_{\cont, \rand_\env, \rand_\subj}$. By hypothesis, $\Rand'_{\cont, \rand_\env, \rand_\subj} \approx_{\eps,\delta} \Rand_\cont$ for every $\rand_\env$ and $ \rand_\subj $.

Now consider the deletion game with the randomized machines $\env,\subj$. Let $\Rand'_\cont$ denote the output of the simulator in that game. Since the simulator $\Sim$ does not depend on the environment or data subject, the distribution of $\Rand'_\cont$ conditioned on the event that $\Rand_\env = \rand_\env$ and $\Rand_\subj = \rand_\subj$ is exactly $P_{\rand_\env, \rand_\subj}$. Thus, we have
$$ (\Rand_\env, \Rand_\subj, \Rand'_\cont) \equiv (\Rand_\env, \Rand_\subj, \Rand'_{\cont, \Rand_\env, \Rand_\subj}) \approx_{\eps,\delta} (\Rand_\env, \Rand_\subj, \Rand_\cont)\, , $$
where $\equiv$ denotes that two random variables have identical distributions and we have used the convexity of the $\approx_{\eps,\delta}$ relation. Dropping the first two components, we get that $\Rand'_C \approx_{\eps,\delta} \Rand_\cont$, as desired.

Furthermore, by hypothesis, we have that $\s = \s'$ with probability at least $1-\delta$ for each setting of $\rand_\env$ and $\rand_\subj$ (since $\s$ depends only on $\Rand_\env$ and $\Rand_\subj$ via the queries $\queries_\env$). Averaging over $\rand_\env$ and $\rand_\subj$ shows that the overall probability that $\s = \s'$ is also at least $1-\delta$.
\end{proof}

\subsection{Discussion of the definition}
\label{sec:definition:discussion}

\para{On constraining $\cont$'s state}
Our definition imposes a condition on the internal state of the controller at a moment in time. Namely that, immediately after the data subject $\subj$ is deleted, the actual state of the controller $\s_\cont$ can be plausibly attributed to the interaction between the controller and the environment alone.
This in turn provides a guarantee for anything the controller may do in the future. Namely,
if the real controller was replaced by the ideal controller at the moment of \subj's deletion, the environment would never know.

As an alternative, one might consider  restricting future behavior directly. We prefer to restrict the state as well. It is simpler to describe---for instance, because there is a natural termination condition. It is also future-proof: Any controller that satisfied the future-behavior version but not the state version could choose to, in the future, violate the guarantee by publishing the state at the time of \subj's erasure. Imposing the condition on the state directly makes it impossible for the future behavior of the real and ideal controllers to deviate, rather than merely possible for them not to.

\para{Differential privacy and deletion}
As we will see later in this paper, differential privacy (DP) can in some cases provide deletion-as-control almost automatically, with no additional action required of the data controller (Prop.~\ref{prop:one-shot-DP} and Thm.~\ref{thm:PPshi_filter}).
We believe that this makes sense both from the point of view of what DP means and 
the spirit of data protection regulations.
When greater protection is warranted, deletion-as-control should not serve as the sole basis of analysis---nor should, perhaps, a right to erasure.%

We're guided by a simple intuition: If a single individual has almost
no influence on the result of  data processing---the condition guaranteed by differential privacy---then nothing needs to be done to 
remove that individual's influence. 
This intuition closely tracks some 
prior approaches to deletion.
For instance, to show that  differentially private controllers satisfy  deletion-as-\textit{control}, we actually show that they meet  the much stricter requirements of (approximate) deletion-as-\emph{confidentiality}~\cite{GGV20}; see Lem.~\ref{lem:PP_to_conf}.
Some existing machine unlearning algorithms embody the same intuition, leaving the trained model unaltered as long as
the deleted data points had no effect on the resulting model \cite{ginart2019making}.\footnote{In contrast, \citet{thudi2022necessity} argue that any definition where a data controller ``do[es] not need to do anything and can claim the unlearning is done''---including machine unlearning definitions based on approximate history independence (Section~\ref{sec:MUL-2-AHI})---is ``not well-defined''. We disagree.}

The fact that DP can provide deletion-as-control fits well with the data protection regulations, like GDPR and CCPA, that inspire our work. Generally, these laws give individuals rights regarding the processing of \emph{personal data} relating to them.%
But these rights, including the right to erasure, do not extend to data that have been sufficiently anonymized.\footnote{Recital 26 states this explicitly: ``The principles of data protection should therefore not apply to anonymous information, namely information which does not relate to an identified or identifiable natural person or to personal data rendered anonymous in such a manner that the data subject is not or no longer identifiable.''}
If one believes that in some cases, DP anonymizes data for the purposes of GDPR, say, then in such cases the data controller need not take any further action when a data subject requests deletion.
Whether DP releases constitute personal data 
is explored in recent work bridging computer science formalisms with legal analysis \cite{nissim2017bridging,altman2021hybrid}. 
Though the general question remains unresolved, DP has been used to argue compliance
with privacy laws for several high-profile data releases, including by the U.S.\ Census Bureau~\cite{census-dp-tech-report}, Facebook~\cite{fb-urls-dataset}, and Google~\cite{google-mobility-reports}.

Of course, DP is not always the answer. For example, if a model was trained using data collected without proper consent, one might require that no benefit derived from the ill-gotten data remains. The Federal Trade Commission first adopted this type of \emph{algorithmic disgorgement} in a 2021 settlement with photo sharing app Everalbum \citep{slaughter2020algorithms}.
Differential privacy should not  shield against such algorithmic disgorgement.\footnote{\citet{achille2023ai} seem to disagree, writing: ``In many ways, differential privacy (DP) can be considered the `gold standard' of model disgorgement''.}

\para{Deleting groups}

Our definition provides a guarantee for an individual data subject. What about groups?
If many people request to be deleted, then each individual person enjoys the individual-level guarantee provided by deletion-as-control. %
But the group does not necessarily enjoy an analogous group-level guarantee. For example, the group-level deletion guarantee for the DP-based controllers in Section~\ref{sec:DP} decays linearly with the group size. For large groups, the group \emph{as a group} doesn't enjoy meaningful protection. 

This seems unavoidable in contexts where 
(useful) statistics are published once and not subsequently updated. It  reflects a fundamental difference between deletion-as-control and the history independence-style definitions in the machine unlearning literature, discussed in Section~\ref{sec:MUL-2-AHI}. 
Suppose, for example, that a controller trains a model $\theta$ using data from $n$ people. Then all $n$ people request deletion, leaving the controller with a model $\theta^*$. History independence would require that $\theta^*$ be essentially trivial: $\theta^*$ should perform about as well as a the model $\theta_0$ trained on an empty dataset (Section~\ref{sec:MUL-2-AHI}). On the other hand, 
the DP-based controllers in Section \ref{sec:DP} allow
$\theta^*$ to perform very well on the learning task.

An individual-level guarantee is in line with data privacy laws.
To whit, the GDPR grants a right to erasure to ``the data subject'' who is ``[a] natural person'' (Art.~4,~17). 
Even so, group-level deletion may be more appropriate in some settings (e.g., algorithmic disgorgement discussed above). Exploring 
group deletion is an important direction for future work.

\para{Composition}
Composition is an important property of good cryptographic definitions. We do not yet have a complete picture of how deletion-as-control composes. Theorem~\ref{thm:simul-comp} states a limited composition theorem that applies to parallel composition of two controllers at least one of which satisfies a very strong guarantee (specifically, it must implement a deterministic functionality with perfect, as opposed to approximate, deletion-as-control). By induction, this extends to the parallel composition of $k$ controllers if all but one satisfy the strong guarantee. This can be used to reason about complex interactive functionalities built from multiple strongly history-independent data structures. 

Proving more general composition for deletion-as-control is an important question for future work. Addressing it seems  challenging since it is closely related to still-open questions about composition for differential privacy. For example, it was shown  only very recently that differential privacy composes when mechanisms are run concurrently with adaptively interleaved queries~\citep{vadhan2022concurrent}. While that result allows adaptive query ordering, the dataset itself is fixed in advance. Deletion-as-control allows both  queries and  data to be specified adaptively. Proving composition of deletion-as-control seems only harder than the analogous question for differential privacy.

\para{Other limitations of our approach} 
We touch on two limitations of our approach.
First, there is no quantification of ``effort.'' 
The EU's right to be forgotten stems from \emph{Google v Costeja}, where the Court of Justice for the European Union ruled that a search engine may be required to remove certain links from search results~\cite{Costeja-press-release}.
But there are limits. Today, Google will only remove the result from search queries related to the name of the person requesting deletion, but not from other search queries~\cite{GoogleRTBF}.
This suggests a definition in which results are hidden from a low-resource adversary who only makes general searches, but not from an adversary with more side information or time, carrying out more targeted or exhaustive searches respectively. Modeling that sort of subtlety appears to require fundamental changes from all existing approaches, ours included.

Second, a failure of deletion as we formulate it doesn't map to an explicit attack on a system. It corresponds instead to a disconnect between the real execution and a counterfactual one in which Alice's data never existed but her effect on others' data remains. In this sense the definition is quite different from standard cryptographic ones, and it doesn't obviously correspond to an adversarial model nor combine well with other cryptographic definitions. This is also true of the history-independence approach, including the definitions in prior work on machine unlearning. Deletion-as-confidentiality \citep{GGV20} does have a more straightforward cryptographic flavor but, as we argue, its strict requirement is ill suited for many application.

\subsection{Real and ideal executions in detail}

\label{sec:execution-model-detailed}
Our definition involves the interaction of three parties in the real execution ($\cont$, $\env$, and $\subj$) and two parties in the ideal execution ($\cont$ and $\dummy$). Formally, the parties in our definition are interactive Turing machines (ITM) with behavior specified by code. An execution involves sequentially activating and deactivating the ITMs  until a termination condition is reached. Activations are tied to message-passing as described below. At any time, at most one ITM is active.
We denote the executions using $\exec{\cont,\env,\subj}$ or  $\exec{\cont,\dummy}$, respectively.

We require that each ITM eventually terminates for every setting of its tapes. (Note that this does not imply termination of an execution involving multiple such ITMs, who may for example pass messages back and forth forever.)
One way to enforce termination is to ask that every party come with a time bound limiting the number of steps it executes when activated. We do not model this particular detail explicitly.

Each ITM has five tapes: work, input, output, randomness, and channel ID.
An ITM's \s at a given time includes only the contents of its work tape at the time.
An ITM may freely read from and write to the work tape. The input tape is read-only and is reset when an ITM is deactivated. (At the start of the execution, input tapes are initialized with channel IDs as described below.) The output tape is write-only and is reset when an ITM is activated.

The randomness tape is initialized with uniform random bits.\footnote{%
        One could instead  view the controller's randomness as consisting of i.i.d.\ samples from an efficiently-sampleable distribution that may depend on $\cont$ (e.g., Laplace noise or Gaussian noise). Uniform bits is without loss of generality, as the simulator $\Sim$ defined below may be inefficient.}
To avoid a priori bounds on running time or randomness complexity, the randomness tape is countably infinite.
It may only be read \emph{once} over the course of the whole execution. That is, reading a bit from the randomness tape automatically advances that tape head one position, and there is no other way to move that tape head; if an ITM wishes to re-use bits from its randomness tape, it must copy them to its work tape.\footnote{%
    \label{footnote:randomness}This is a departure from the model of \cite{GGV20}, wherein the randomness tape may be read multiple times without counting as part of the state. That would allow a controller to avoid deleting anything by, for instance, encrypting its state using its randomness as a secret key---a detail that was overlooked in \cite{GGV20}. Our fix is to make the randomness read-once. As a result, a party must store in the work tape any randomness that it will later reuse.}

Messages between parties are passed over authenticated channels, represented by a channel ID $\cid$ associated with two parties $A$ and $B$. A single pair of parties may have many associated channels. Party $A$ sends a message by writing $(\cid,\msg)$ to its output tape, where $\msg \in \{0,1\}^* \cup \{\delete\}$, where $\delete$ is a special message (i.e., not in $\{0,1\}^*$). When $A$ finishes writing its output, $A$ is immediately deactivated. 
If $\cid$ corresponds to a channel ID between $A$ and another party $B$, then $(\cid, \msg)$ is written to $B$'s input tape and $B$ is activated.
Otherwise, the special message $\fail$ is written to $A$'s input tape and $A$ is activated. Importantly, $B$ does not learn the \emph{party} that sent $\msg$, only the \emph{channel} $\cid$ over which it was sent.

Each ITM's channel ID tape is initialized with the \cid's for channels over which the ITM can communicate. The channel ID tape is read-once, just like the randomness tape.
In the real execution, $\env$'s tape is initialized with $\cid$s (countably infinite, as with the randomness tape). In the ideal execution, $D$'s tape is similarly initialized with $\cid$s. The first is for communication with $\subj$ and the remainder are for communication with $\cont$. $\subj$'s input tape is initialized with only one $\cid$ for communication with $\cont$ and one \cid for communication with \env. $\cont$'s input is initially empty. It must learn the $\cid$s from messages it receives, and cannot distinguish its channel with $\subj$ from its channels with $\env$ or $\dummy$.

The executions begin with the activation of $\env$ or $\dummy$. 
Parties are deactivated when they write a message to their output tape and the recipient is activated. 
If a party halts without writing to its output tape, $\env$ or $\dummy$ is activated.
\para{The real execution} The real execution involves three (possibly randomized) parties: $\cont$, $\env$, and $\subj$ using randomness $\Rand_\cont$, $\Rand_\env$, and $\Rand_\subj$ in $\bit{\bbN}$, respectively. 
 The execution ends after $\subj$ requests deletion from $\cont$.  More precisely, after 
$\subj$ sends $\delete$ to $\cont$, $\cont$ is activated one final time. The execution then ends when $\cont$ halts or writes to its output tape. Note that the real execution may never terminate if $\subj$ never sends $\delete$. We denote the real execution as $\exec{\cont(\Rand_\cont), \env(\Rand_\env), \subj(\Rand_\subj)}$. When appropriate we omit the randomness, writing $\exec{\cont,\env, \subj}$.

The execution generates a (possibly empty) \emph{transcript} $\trans$ of all messages sent to and received by $\cont$: $\trans = ((t, \sender_t, \allowbreak\receiver_t, \cid_t, \msg_t))_t.$ For each $t$: one of $\sender_t$ or $\receiver_t$ is always \cont; $\cid_t$ is an id of a channel between $\sender_t$ and $\receiver_t$; and $\msg_t \in \{0,1\}^* \cup \{\delete\}$. 
Though they may freely communicate, we omit messages between $\env$ and $\subj$ from the transcript. The transcript can be further divided into \emph{queries} and \emph{answers}. The queries are messages sent to $\cont$ (by either $\env$ or $\subj$). 
We denote by $\queries$ the sub-transcript containing all queries,  and by $\queries_\env$ only those messages sent by $\env$. 
The answers \answers are messages to someone \textit{from} $\cont$. We denote by $\answers_\env \subseteq \trans$ the ordered messages from $\cont$ to \env.

The real execution $\exec{\cont(\Rand_\cont), \env(\Rand_\env), \subj(\Rand_\subj)}$ defines a transcript $\trans$, a state $\s_\cont$, and randomness $\Rand_\cont$. The transcript $\trans$ is defined above. The state $\s_\cont$ consists of the contents of $\cont$'s work tape at the end of the execution. 
If the execution does not terminate, we define $\s_\cont = \bot$ and the transcript $\trans$ and its subset $\queries_\env$ are defined to be empty.

\para{The ideal execution}
The ideal execution involves two parties: $\cont$ and $\dummy$. The dummy party $\dummy$'s input tape is initialized with the same channel IDs as $\env$'s tape in the real execution.  The dummy simply replays the queries in $\queries_\env$ if any. At every activation, it sends the next query in the sequence to $\cont$ using the same $\cid$ as in the real execution. Note that the \emph{answers} \dummy receives from \cont may be different from the real-world answers; however, \dummy sends the same queries $\queries_\env$ regardless.

The controller $\cont$ is exactly as in the real execution except that its randomness tape is initialized using \emph{simulated} randomness $\Rand'_\cont\gets \Sim(\queries_\env, \Rand_\cont, \s_\cont)$ instead of \emph{uniform} randomness $\Rand_\cont\sim \distrand$.
The simulator is an (inefficient) algorithm $\Sim$ that takes as input $\queries_\env$, $\s_\cont$, and $\Rand_\cont$ and produces output $\Rand'_\cont$.
Note that $\Rand_\cont$ and $\Rand'_\cont$ consist of countably infinitely-many bits; hence $\Rand_\cont$ and $\Rand'_\cont$ cannot be treated as conventional inputs and outputs to $\Sim$. Instead, we give $\Sim$ access to a special tape on which $\Rand_\cont$ is written. $\Sim$ may overwrite finitely-many bits of this tape before it halts. We denote by $\Rand'_\cont$ the final contents of this tape.

We denote the ideal execution as $\exec{\cont(\Rand'_\cont), \dummy(\queries_\env)}$.
The ideal execution ends after $\dummy$ sends its last query from $\queries_\env$ to $\cont$. When this happens, $\cont$ is activated one final time. The execution ends when $\cont$ halts or writes to its output tape. $\s'$ is defined as the final state of $\cont$.
Note that the ideal-world execution terminates if the real-world execution terminates. If the ideal execution does not terminate, we define $\s' = \bot$. In this case, $\queries_\env$ must be infinitely long and hence the corresponding real execution did not terminate.

\section{History Independence and Deletion-as-control}\label{sec:HI}

History independence (HI) is concerned with the problem that the memory representation of a data structure may reveal information about the history of operations that were performed on it \cite{micciancio1997oblivious, naor2001anti, hartline2005characterizing}. 
HI requires that the memory representation reveals nothing more than the current logical state of the data structure.
Setting aside a number of technical subtleties, the conceptual connection to machine unlearning is immediate:  if we consider a machine learning model as a representation of a dictionary data structure with insert and remove (i.e., unlearn) operations, a model is HI if and only if it satisfies machine unlearning.

In this section, we state the definition of (non-adaptive) history independence (Section~\ref{sec:HIbasic}). We then define a more general notion that allows for implementations that satisfy the conditions of HI \emph{approximately} and \emph{adaptively} (Section~\ref{sec:adaptive-HI}). 
The generalization is complex since we must explicitly model {adaptivity} in the interactions between a data structure and those issuing queries to it. 
Briefly, an adaptive adversary $\adv$ interacting with the data structure produces two equivalent query sequences.
Adaptive history independence (AHI) requires that the joint distribution of \adv's view and the data structure's state is the same under both sequences. Approximate AHI requires these distributions to be $(\eps,\delta)$-close.

We show that data controllers that satisfy approximate AHI also satisfy our notion of deletion-as-control (Section~\ref{sec:AHI-2-control}) with the same parameters. 
Finally, we show how existing definitions  of \textit{machine unlearning} and the corresponding constructions are all (weakenings of) our general notion of history independence (Section~\ref{sec:MUL-2-AHI}). 

\subsection{History independence}
\label{sec:HIbasic}

An \emph{abstract data type (ADT)} is defined by a universe of operations $\{\op\}$ and a mapping $\adt:(\op,\stateLog) \mapsto (\stateLog',\outLog)$.
We call $\stateLog,\stateLog' \in \{0,1\}^*$ the \emph{logical states} before and after operation, where $\stateLog'$ is a deterministic function of $\stateLog$ and $\op$. 
We call $\outLog \in \{0,1\}^* \cup \{\bot\}$ the \emph{logical output}, which may be randomized.
In subsequent sections we will assume without loss of generality that operations $\op(\id)$ are tagged by $\id \in \{0,1\}^*$. We omit the tags where possible to reduce clutter.

Given an initial logical state $\stateLog^0$ and a sequence of operations $\seq = (\op^1,\op^2,\dots)$, the ADT defines a sequence of logical states $(\stateLog^1, \stateLog^2, \dots)$ and a sequence of outputs $\Out = (\outLog^1, \outLog^2, \dots)$ by iterated application of $\adt$. 
We denote by $\adt(\sigma).\s$ the final logical state that results from this iterated application. When no initial state is specified, it is assumed to be the empty state. 

\begin{definition}
We say two sequences of operations $\seq$ and $\seq'$ are \emph{logically equivalent}, denoted $\seq\logequiv \seq'$, if $\adt(\seq).\s = \adt(\seq').\s$. Logical equivalence is an equivalence relation, and we denote by $[\seq]$ a canonical sequence in the equivalence class of sequences $\{\seq' : \seq' \logequiv \seq\}$.
\end{definition}

An \emph{implementation} (e.g., a computer program for a particular architecture) is a possibly randomized mapping $\impl:(\op, \statePhys) \mapsto (\statePhys', \outPhys)$. 
We call $\statePhys$ the \emph{physical state} and $\outPhys$ the \emph{physical output}.
Both may be randomized. Given an initial state and sequence of operations, $\impl$ defines a sequence of physical states and outputs by iterated application. 
When no initial state  is specified, it is assumed to be the empty state.

\begin{definition}[History independence \citep{naor2001anti}] \label{def:HI}
$\impl$ is a \emph{weakly history independent} implementation of $\adt$ \emph{(WHI-implements $\adt$)} if
\begin{equation}
    \seq \logequiv \seq' \implies  \impl(\seq).\s \equiv \impl(\seq').\s, \label{eqn:WHI:perfect}
\end{equation}
where $\equiv$ denotes equality of distributions.
$\impl$ is a \emph{strongly history independent} implementation of $\adt$  (SHI-implements ADT) if  for all initial states $\statePhys$
\begin{align}
    \seq \logequiv \seq' \implies 
    \impl(\seq,\statePhys).\s \equiv \impl(\seq',\statePhys).\s . \label{eqn:SHI:perfect}
\end{align}
\end{definition}

One can obtain approximate, \textit{nonadaptive} versions of history independence by replacing $\equiv$ with $\approxequiv{\eps,\delta}$ in Definition~\ref{def:HI}. However, because the sequence of queries is specified ahead of time, such a definition's guarantees are not meaningful in interactive settings---see Example~\ref{ex:nonadaptiveHI}.

We do not define correctness of an implementation of an ADT. 
Thus every ADT trivially admits a SHI implementation (e.g., $\impl$ always outputs $\bot$). Omitting correctness simplifies the specification of the ADT while allowing flexibility---for example, if approximate correctness suffices for an application. 
The usefulness of an implementation requires separate analysis.
However, history independence simplifies this step: it suffices to analyze the utility for the canonical sequences of operations $[\seq]$ instead of arbitrary sequences $\seq$.

\subsubsection{Strongly History Independent Dictionaries}
\label{sec:SHI_DICT}
To illustrate history independence, consider the dictionary ADT, which models a simple key-value store.  
Looking ahead, we will use history independent dictionaries to build deletion-compliant controllers from differential privacy.
For our purposes, the keys will be party IDs $\id$; the values can be arbitrary. The ADT supports operations $\ins{\id}$, $\del{\id}$, $\mathsf{get}(\id)$, and $\mathsf{set}(\id,value)$, where $\mathsf{set}$ associates the key corresponding to $\id$ with $value$, and $\mathsf{get}$ returns the most recently set value. We assume that $\ins{\id}$ is equivalent to $\mathsf{set}(\id,\top)$ where $\top$ is a special default value.

Dictionaries are typically implemented as hash tables, but such data structures are generally \textit{not} history independent.

For example, in hashing with open addressing, deletions are typically done lazily (by marking the cell for a deleted key with a special, ``tombstone'' value). So one can tell whether an item was deleted but also learn the deleted item's hash value.
 Storing a dictionary as a sorted list is inefficient---updates generally take time $\Omega(n)$, where $n$ is the current number of keys---but it enjoys strong history independence. The sorted representation depends \textit{only} on the logical contents of the dictionary, not on the order of insertions nor which items were inserted then deleted.  
 Since we do not focus on efficiency here, the reader may think of the sorted list as our default implementation of a dictionary. \footnote{In fact, a strongly history independent hash table implementation with constant expected-time operations was described by \citet{blelloch2007strongly}.}

\subsection{Adaptive History Independence (AHI)}
\label{sec:adaptive-HI}

The history independence literature gives no guarantees against adaptively-chosen sequences of queries, because the two sequences $\seq$ and $\seq'$ are fixed before the implementation's randomness is sampled.

\begin{example} \label{ex:nonadaptiveHI}

Consider an implementation $\impl$ of a dictionary with two operations: $\ins{\id}$ and $\del{\id}$. Upon initialization, $\impl$ outputs the first $n$ bits of its uniform randomness tape, denoted $r$, and stores it in state $\statePhys^1 = r$.
If the first operation is $\ins{r}$, $\impl$ will store the subsequent sequence of operations it its state. Otherwise, $\impl$ ignores all subsequent operations, setting $\statePhys = r$. $\impl$ produces no outputs other than the initial $\outPhys^1 = r$.
$\impl$ satisfies $(0,2^{-n})$-approximate \textit{nonadaptive} strong history independence: for any sequences $\seq$ and $\seq'$ fixed in advance, and for any initial $\statePhys$, $\impl(\seq,\statePhys).\s = r = \impl(\seq', \statePhys)$, except with probability $2\cdot 2^{-n}$.
However, an adaptive adversary that sees $\outPhys^1 = r$ can easily produce distinct $\seq \logequiv \seq'$ such that $\impl(\seq).\s =\seq \neq \seq' = \impl(\seq').\s$.
\end{example}

Inspired by \citep{gupta2021adaptive}, we extend the well-studied notion of history independence to the \emph{adaptive} setting, where the sequence $\seq$ of operations is chosen adaptively by an algorithm interacting with an implementation of an ADT. 

We consider an interaction $\exec{\impl(\Rand), \adv}$ between an algorithm $\adv$ and the implementation $\impl$ with random tape $\Rand \sim \distrand$ (Figure~\ref{fig:AHI}).

In the interaction, $\adv$ adaptively outputs an operation\\
$\op^{i} \gets \adv(\op^1, \outPhys^1,\dots,\op^{i-1},\outPhys^{i-1})$, and receives the output $\outPhys^{i}$ in return. 
The interaction defines a sequence of operations $\seq$ and corresponding outputs $\Out$.
Eventually, \adv outputs a sequence $\seq^*$ that is logically equivalent to the sequence $\seq$ of operations performed so far.
$\impl$ is executed on $\seq^*$ and alternate randomness $\Rand^*$, resulting in $\statePhys^*  = \impl(\seq^*; \Rand^*).\s$. 
We consider two variants: $\Rand^* = \Rand$, or $\Rand^*\sim \distrand$ independent of $\Rand$.

We consider the adversary's ability to distinguish the real state $\statePhys = \impl(\seq;\Rand).\s$ and the logically equivalent state \\$\statePhys^*  = \impl(\seq^*; \Rand^*).\s$, given its view $\view_{\impl, \adv} = (\seq, \seq^*, \Out)$.
Our definition of adaptive history independence requires that the joint distributions of $(\view_{\impl,\adv}, \statePhys)$ and $(\view_{\impl,\adv},\statePhys^*)$ be $(\eps,\delta)$-close.
We restrict ourselves to adversaries \adv such that $\exec{\impl,\adv}$ always terminates, which we call \emph{valid} adversaries.

\begin{definition}[Adaptive (weak) history independence]\label{def:AHI}
An implementation $\impl$ of an $\adt$ is $(\eps, \delta)$-history independent (AHI) if for all valid adversaries $\adv$:
$$
(V_{\impl, \adv},\statePhys)
\approxequiv{\eps, \delta}
(V_{\impl, \adv}, \statePhys^*)
$$

    where the distributions are given by the following probability experiment. The experiment has two versions (identical or independently drawn randomness):
        
    \begin{algorithmic}[1]
        \State $\Rand \sim \distrand$ 
        \State Either $\Rand^* = \Rand$ (identical randomness), or $\Rand^* \sim \distrand$ (independent randomness) \label{step:AHI-choose-rand}
        \State $(\seq,\Out ,\statePhys) \gets \exec{\impl(\Rand),\adv}$  
        \State \adv outputs $\seq^*$
        \State $\statePhys^* \gets \impl(\seq^*; \Rand^*).\s$
        \State If $\seq^* \not\logequiv\seq$, then $\statePhys^* \gets \statePhys$
        \State $V_{\impl, \adv} \gets (\seq, \seq^*, \Out)$
    \end{algorithmic}
    We say an implementation \impl is $(\eps, \delta)$-AHI if it satisfies either version of the definition (identical or independently drawn randomness).

\end{definition}

    Note that for both versions of Definition~\ref{def:AHI}, \Rand and $\Rand^*$ are identically distributed. One can generalize the definition to allow for any joint distribution over $(\Rand, \Rand^*)$ such that the two marginal distributions are identical (all of our proofs would still go through for this general notion). However, we avoid this generality since the two versions that we present are sufficient for our needs.

    The AHI definiton here corresponds to the approximate, adaptive version of \textit{weak} history independence (Definition~\ref{def:HI}), which considers only the state state at a single point in time. This suffices for our purposes. One could generalize Definition~\ref{def:AHI} to allow an adversary to see the full internal state in the real execution at multiple points in time; this would correspond to strong history independence. We leave such an extension for future work.

\begin{figure}
    \centering
    \includegraphics[scale=0.4]{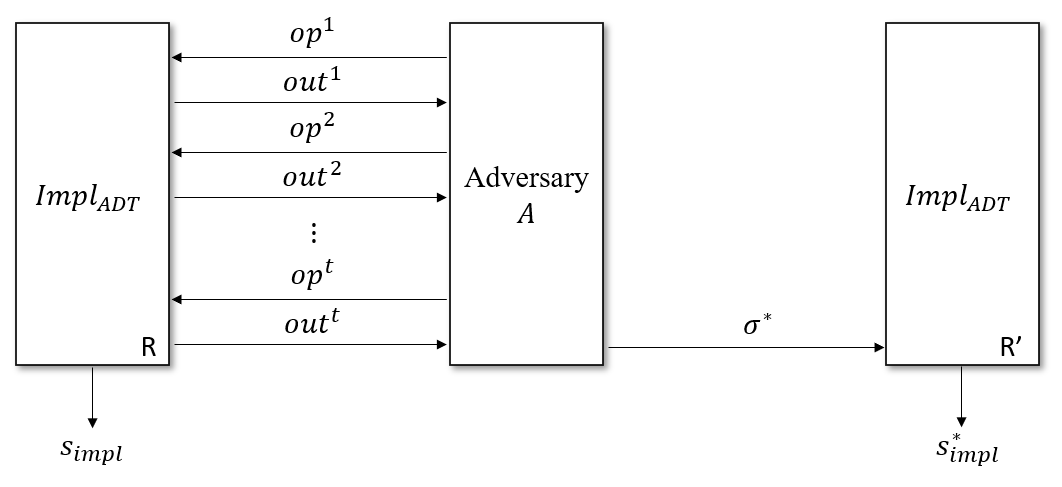}
    \caption{Adaptive history independence game, from Definition~\ref{def:AHI}}
    \label{fig:AHI}
\end{figure}

\para{From Strong Nonadaptive HI to Adaptive History Independence}
    \label{sec:adaptiveHIexamples}
    To illustrate the definition, consider the case of strongly history independent dictionaries. The hash table construction of \citet{blelloch2007strongly} is a randomized data structure with the following  property: for every setting of the random string $\Rand=r$, and for every two logically equivalent sequences $\seq,\seq^*$, the data structure stores the same table $\statePhys(\seq,r) = \statePhys(\seq^*,r)$. (In fact, essentially all $(0,0)$-strongly HI data structures can be modified to satisfy this strong guarantee \citep{hartline2005characterizing}). %
    
    The data structure's answers will thus depend only on the logical data set. However, they might also depend on the randomness $r$---for instance, the answers might leak the length of the probe sequence needed to find a given element or other internal information. In an adaptive setting, queries might depend on $r$ via the previous answers. To emphasize this dependence, we write the realized query as $\seq(r)$. 
    
    This construction satisfies $(0,0)$-adaptive HI with identical randomness (where $\Rand'=\Rand$ in Step~\ref{step:AHI-choose-rand} of Definition~\ref{def:AHI}) but \textit{not} the version with fresh randomness ($\Rand' \perp \Rand$). To see why it satisfies the definition with identical randomness, consider the tuples $(\seq(\Rand), \seq(\Rand)^*, \Out, \statePhys(\seq,\Rand))$ and $(\seq(\Rand), \seq(\Rand)^*, \Out, \statePhys(\seq^*,\Rand))$ resulting from the game in Def.~\ref{def:AHI}. By logical equivalence, the tuples are the same, and hence identically distributed. On the other hand, if we select $\Rand'$ independently of $\Rand$, then the adversary will, in general, see a difference between $\seq(\Rand)$  and $\statePhys(\seq^*,\Rand')$ (for instance, the hash functions corresponding to $\Rand$ and $\Rand'$ will be different with high probability).

\subsection{AHI and Deletion-as-Control}
\label{sec:AHI-2-control}

We prove two theorems showing that data controllers that implement history independent ADTs satisfy deletion-as-control. Before that, two technicalities remain. First, the statement only makes sense if the ADT itself supports some notion of deletion. Second, there is a mismatch between the syntax of ADTs/implementations and the deletion-as-control execution. Next we define \emph{ADTs supporting logical deletion} and \emph{controllers relative to an implementation} to handle the two issues, then state and prove the main theorems of this section.

Consider an ADT where operations $\{\op(\id)\}$ are tagged with an identifier $\id\in\{0,1\}^*$ (e.g., the channel IDs). %
For a sequence $\seq$ of operations and an $\id^*$, let $\seq_{-\id^*}$ be the sequence of operations with every operation with identifier $\id^*$  removed (that is, $\op(\id^*)$ for all values of $\op$).

\begin{definition}[Logical Deletion] \label{def:logical_deletion}
ADT supports \emph{logical deletion} if there exists an operation $\delete$ such that for all sequences of operations $\seq$ and for all IDs $\id^*$:  $\left(\seq\|\delete({\id^*})\right) \logequiv \seq_{-\id^*}$.
\end{definition}

Logical deletion is important for ensuring the correctness property of deletion compliance (that the states in the two worlds are identical). Conversely, if an ADT does not support logical deletion, then one would expect any controller that faithfully implements the ADT to violate the deletion requirement.

The notion of logical deletion applies to a wide variety of abstract data types, not only dictionaries. For example, consider a public bulletin board, where users can post messages visible to everyone. The natural ADT would allow creation of new users, insertion and maybe deletion of specific posts, and deletion of an entire user account. This latter operation would undo all previous actions involving that account. 
We consider a family of data controllers $\cont$ that are essentially just implementations of ADTs that can interface with the execution.

\begin{definition}[Controller relative to an implementation] \label{def:cont_impl}
Let $\impl$ be an implementation of an $\adt$. We define the controller $\cont_\impl$ relative to $\impl$ as the controller that maintains state $\s$ and works as follows:
\begin{itemize}[itemsep=0.3pt, topsep=4pt, partopsep=2pt]
    \item On input $(\cid, \msg)$:
    \begin{itemize}[topsep=1pt]
        \item $(\s',\outPhys)\gets \impl(\op(\cid), \s)$, where $\op \gets \msg$.
        \item Write $(\cid, \outPhys)$ to the output tape.
    \end{itemize} 
    \item On input $\fail$: Halt.
\end{itemize}

We say $\cont$ is history independent (adaptively/strongly/weakly) if $\impl$ is history independent (adaptively/strongly/weakly).

\end{definition}

\begin{restatable}{theorem}{SHIDEL}
\label{thm:perfect_shi_del}
For any ADT that supports logical deletion and any SHI implementation $\impl$, the controller $\cont = \cont_\impl$ satisfies $(0,0)$-deletion-as-control with the simulator that outputs $\Rand'_\cont = \Rand_\cont$.
\end{restatable}

Without the condition on the simulator, this is a corollary of the theorem below. 
The proof in  Appendix~\ref{app:perfect_shi_del}
uses a foundational result in the study of history independence. Roughly, that SHI implies canonical representations for each logical state of the ADT \citep{hartline2005characterizing}.

\begin{theorem}\label{thm:AHI_to_del}
For any ADT that supports logical deletion and any $\impl$ of the ADT satisfying $(\eps, \delta)$-AHI (with either variant of Definition~\ref{def:AHI}), the controller $\cont = \cont_\impl$ is $(\eps, \delta)$-deletion-as-control compliant.
\end{theorem}

The proof uses a simple, novel result on indistinguishability, dubbed the Coupling Lemma, which we present and prove in Appendix~\ref{app:coupling-lemma}.
\begin{proof}[Proof of Theorem~\ref{thm:AHI_to_del}]
Fix any deterministic \env and \subj. By Lemma~\ref{lem:derandomize} this is without loss of generality.
Throughout this proof, we drop the subscript $\impl$ from the controller $\cont_\impl$, and drop the subscript $\cont_\impl$ from the controller's state $\s_{\cont_\impl}$ and randomness $\Rand_{\cont_\impl}$.

Definition~\ref{def:AHI} considers two variants of AHI, depending on how the randomness $\Rand$ used in the initial $\impl$ evaluation (with sequence $\seq$) relates to the randomness $\Rand^*$ used in the logically equivalent evaluation (with sequence $\seq^*$). (Either $\Rand$ and $\Rand^*$ are sampled i.i.d.\ or $\Rand^* = \Rand$.)
Let $\ahiJointDist$ be the joint distribution over $(\Rand,\Rand^*)$ for which $\impl$ enjoys the $(\eps,\delta)$-AHI guarantee. Observe that the marginal distribution of $\Rand^*$ is $\distrand$.  

To prove that the controller $\cont$ is deletion-as-control compliant, we will use the following simulator $\Sim$:
$\Sim(\queries_\env,\Rand,\s)$ samples $(\widetilde{\Rand},\widetilde{\Rand}^*)$ from $\ahiJointDist$ conditioned on the following event: $$\queries_\env(\widetilde{\Rand}) = \queries_\env \quad \land \quad \cont(\queries_\env;\widetilde{\Rand}^*) = \s,$$
where $\queries_\env(\widetilde{\Rand})$ are the queries from $\env$ to $\cont$ in the execution $\exec{\cont(\widetilde{\Rand}),\env,\subj}$. $\Sim$ outputs $\Rand' = \widetilde{\Rand}^*$.
The real execution's state is $\s = \cont(\queries(\Rand);\Rand)$. 
The ideal execution's state is $\s' = \cont(\queries_\env(\Rand);\Rand')$.

\begin{claim}
\label{claim:ahi-claim-1}
Define  $f$ and $g$ as follows:
\begin{align*}
f(\Rand) &= \bigl(\queries_\env(\Rand),\ \cont(\queries(\Rand);\Rand)\bigr) \\
g(\Rand, \Rand^*) &= \bigl(\queries_\env(\Rand),\ \cont(\queries_\env(\Rand); \Rand^*)\bigr),
\end{align*}
Then $(\Rand,\Rand^*) \sim \ahiJointDist \implies f(\Rand) \approxequiv{\eps,\delta} g(\Rand,\Rand^*).$
\end{claim}

The proof of the claim is below. We use the claim to complete the proof of the theorem. 
Consider a procedure that samples $\Rand\sim \distrand$, and then samples $(\widetilde{\Rand}, \widetilde{\Rand}^*) \sim \ahiJointDist\big|_{f(\Rand) = g(\widetilde{\Rand},\widetilde{\Rand}^*)}$ (if possible, otherwise uniformly).
Observe that this is exactly the distribution of $\Rand\sim\distrand$ and the intermediate vales $(\widetilde{\Rand},\widetilde{\Rand}^*)$ sampled by $\Sim$ defined above.
Applying the Coupling Lemma (Lemma~\ref{lem:approx_coupling}) to the preceding claim, we get that the joint distribution over $(\Rand,\widetilde{\Rand},\widetilde{\Rand}^*)$ satisfies $(\widetilde{\Rand},\widetilde{\Rand}^*) \approxequiv{\eps,\delta} \ahiJointDist$, and $f(\Rand) = g(\widetilde{\Rand},\widetilde{\Rand}^*)$ with probability at least $1-\delta$.

$\Sim$ outputs $\Rand' = \widetilde{\Rand}^*$. 
Because the marginal distribution of $\widetilde{\Rand}^*$ is $\distrand$, we have that $\Rand'\approxequiv{\eps,\delta} \distrand$.
The equality of $f$ and $g$ implies that 
$\s' = \cont(\queries_\env(\Rand);\Rand') = \cont(\queries(\Rand);\Rand) = \s$. Hence, $\cont_\impl$ satisfies $(\eps,\delta)$-deletion-as-control with the simulator $\Sim$.
\end{proof}

\begin{proof}[Proof of Claim~\ref{claim:ahi-claim-1}]
We will use the adaptive history independence of $\impl$.
We define an AHI adversary $\adv = \adv_{\env, \subj}$ (see Figure~\ref{fig:AHI_to_del_reduction}) 

which gets query access to the implementation $\impl$ in the AHI game. \adv emulates the execution $\exec{\cont,\env,\subj}$ internally, querying $\impl$ as needed to implement $\cont$. When the execution terminates, $\adv$ outputs $\seq^* = \queries_\env$, the queries sent from \env to $\cont$. The actual sequence of queries received by $\impl$ in the real execution is $\seq = \queries$, including all queries from both \env and \subj.

We claim that $\seq^* \logequiv \seq$. 
To see why, observe that $\seq^* = \seq_{-\subj}$ is the query sequence with all queries from $\subj$ removed. Moreover, $\seq$ ends in $\delete(\subj)$.
By the hypothesis that $\adt$ supports logical deletion $\seq \logequiv \seq_{-\subj} = \seq^*$.

\begin{figure}[H]
    \centering
    \includegraphics[scale=0.38]{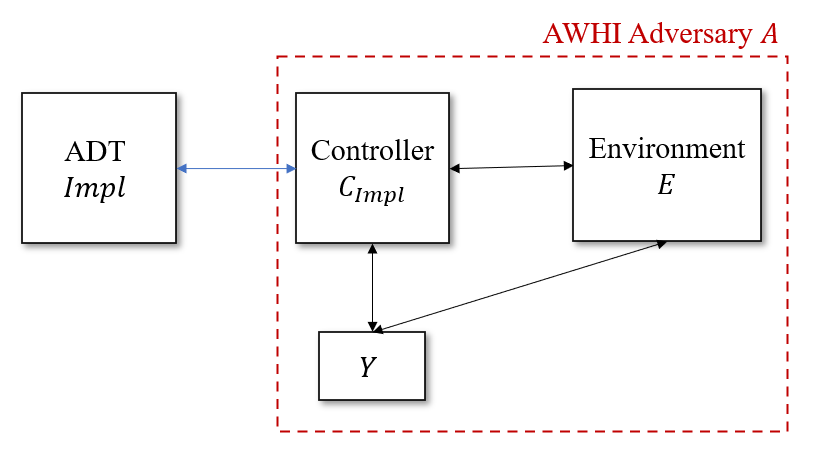}
    \caption{Reduction from $(\eps, \delta)$-AHI to $(\eps, \delta)$-deletion-as-control.}
    \label{fig:AHI_to_del_reduction}
\end{figure}

By the $(\eps,\delta)$-AHI of $\impl$, 
$(\view_{\impl, \adv}, \statePhys) 
\approxequiv{\eps, \delta} (\view_{\impl, \adv}, \statePhys^*)$.
By construction, the state $\statePhys = \impl(\seq;\Rand).\s$ is exactly the controller's state $\statePhys = \cont(\queries;\Rand)$ in the real execution $\exec{\cont(\Rand),\env,\subj}$.
Likewise, the state $\statePhys^* = \impl(\seq^*; \Rand^*)$ is exactly the the controller's state $\statePhys' = \cont(\queries_\env; \Rand^*)$ in the  ideal execution $\exec{\cont(\Rand^*), \dummy(\queries_\env)}$. 
$\adv$'s view $\view_{\impl,\adv}$ consists of $(\seq, \seq^*, \Out) = (\queries(\Rand), \queries_\env(\Rand), \Out)$.  Hence,
\begin{equation}\label{eq:AHI_eq1}
f(\Rand) = (\queries(\Rand), \statePhys) \approxequiv{\eps, \delta} (\queries(\Rand), \statePhys' ) = g(\Rand,\Rand^*).
\qedhere 
\end{equation}
\end{proof}

As a corollary to Theorem~\ref{thm:AHI_to_del}, any history independent implementation of the Private Cloud Storage touchstone or the Public Bulletin Board examples (from Section~\ref{sec:touchstone}) satisfy deletion-as-control. 

The Private Cloud Storage functionality works like a key-value store: each user has their own dictionary $D_\id$. To upload a file, a user sends $Upload(\id, filename, file)$ to the controller, which internally adds $(filename, file)$ to $D_\id$ (if $filename$ already exists in $D_\id$ then nothing happens). To download one of their own files, the user sends $Download(\id, filename)$ to the controller and receives the corresponding $file$ from $D_\id$, if one exists. To delete their account, the user sends $Delete(\id)$, in which the controller removes dictionary $D_\id$. Two sequences of operations are logically equivalent if, for every user $\id$, the dictionary $D_\id$ has the same logical content after applying the two sequences. 

We define the Public Bulletin Board functionality as follows: users can post a message by sending $Post(\id, msg)$ to the controller; they can receive all messages currently on the board by sending $Read()$; and they can delete all of their messages by sending $Delete(\id)$ to the controller. Internally, the Public Bulletin Board stores an ordered list of $(\id, msg)$ pairs (ordered by insertion time). Two sequences of operations are logically equivalent if they yield the same ordered list of $(\id, msg)$ pairs. 

\begin{corollary}
\label{cor:priv_cloud_pub_bboard}
The Private Cloud Storage and Public Bulletin Board touchstone controllers  (Section~\ref{sec:touchstone}) implemented with $(\eps, \delta)$-adaptive history independence each satisfy $(\eps, \delta)$-deletion-as-control.
\end{corollary}

\subsection{From Prior Definitions of Machine Unlearning to  History Independence}
\label{sec:MUL-2-AHI}

We claim that AHI captures the essence of existing definitions of ``machine unlearning" (that is, protocols that update a machine learning model to reflect deletions from the training data). Each definition in the literature corresponds to a special case of history independence, though each  weakens the definition in one or more ways (even when their constructions satisfy the stronger, general notion).  For illustration, we discuss the approach of \citet{gupta2021adaptive} in detail.

The basic correspondence comes via considering an abstract data type, which we dub the \textit{Updatable ML}, that extends a dictionary: it maintains a multiset $\data$ of labeled examples from some universe $\univ$. In addition to allowing $\ins{}$ and $\del{}$ operations, it accepts a possibly randomized operation $\predict$ which outputs a predictor (or other trained model) $\pred$ trained on $\data$. The accuracy requirement for $\pred$ is generally not fully specified, not least because many current machine learning methods don't come with worst-case guarantees. The literature on machine unlearning generally requires that the distribution of the final predictor $\pred$ (e.g., the model parameters) is approximately the same as it would be for a minimal sequence of operations that leads to the same training data set. In particular, \emph{deleting an individual $\subj$ should mean that $\pred$ looks roughly the same as if $\subj$ had never appeared in $\data$}. In principle one could satisfy the requirement by simply retraining $\pred$ from scratch every time the data set changes,  though this may  be practically infeasible. The literature therefore focuses on methods that allow for faster updates.

\para{Specializing history independence to updatable ML}
First, we spell out how  AHI (Definition~\ref{def:AHI}) specializes to the Updatable ML ADT. Given a sequence of updates and prediction queries $\seq = (u_1,...,u_t)$, let $\canon{\seq}$ denote a canonical equivalent sequence of updates. For example, $\canon{\seq}$ may (a) discard any insertion/deletion pairs that operate on the same element, with the insertion appearing first,  and (b) list the remaining operations in lexicographic order. In an interaction with an implementation $\impl$ that is initialized with an empty data set, an adversary makes a sequence of queries $\seq$ (of which only the updates actually affect the state) and, eventually, terminates at some time step $t$. 
Let  $\statePhys$ the resulting state of the controller based on randomness $\Rand$, and let $\data_t$ be the final logical data set. After the interaction ends, the adversary outputs some logically equivalent sequence $\seq^*$ (for example, it could choose the  canonical sequence $\seq^* = \canon{\seq}$). Finally, run $\impl$ from scratch with randomness $\Rand'$ and queries $\seq^*$ to get state $\statePhys^*$. The logical data set $\data_t$ is the same for $\seq$ and $\seq^*$ by construction. Definition~\ref{def:AHI} requires that 
\begin{equation}\label{eq:AHI2ML}
        (\seq, \statePhys) \approxequiv{\eps,\delta} (\seq, \statePhys^*) \, .
\end{equation}

\para{Definitions from the literature.} 
The exact definition of deletion varies from paper to paper (and some papers do not define terms precisely). All formulate the problem in terms of what we call the Updatable ML ADT. We claim their definitions are variations on adaptive history independence (Eq.~\ref{eq:adaptMUL}), each with one or more of the following weakenings: %
\begin{description}
    \item[Restricted queries] Some papers consider only a subset of allowed operations---e.g., in \citet{ginart2019making,ullah2021machine,sekhari2021remember} data is inserted as a batch, and then only deletions occur.  
    \item [One output versus future behavior] Some consider only the output of the system at one time (e.g., \citet{neel2020descenttodelete,sekhari2021remember,gupta2021adaptive}), while others also consider the internal state (e.g., \citet{ginart2019making}). The narrower approach \textit{does not constrain the future behavior of the system}.  
    Among those definitions that consider the full state, none discuss issues of internal representations; in this section, we also elide this distinction, assuming that datasets can be represented internally using strongly history-independent data structures. 
    \item [Nonadaptive queries] Except for \citet{gupta2021adaptive}, the literature considers only adversaries that specify the set of queries to be issued in advance. 
    For constructions that are $(0,0)$-HI (that is, in which the real and ideal distributions are identical), this comes at no loss of generality. However, the nonadaptive and adaptive versions of $(\eps,\delta)$-HI for $\delta>0$ are very different \citep{gupta2021adaptive}. Even in \citet{gupta2021adaptive}, the length $t$ of the query sequence is chosen nonadaptively.
    \item [Symmetric vs asymmetric indistinguishability] \citet{ginart2019making} and \cite{gupta2021adaptive} consider a one-sided weakening of $(\eps,\delta)$-indis\-tinguishability (discussed in Remark~\ref{rem:gupta-asymmetry}).
\end{description}

The definitions of  \citet{ginart2019making,ullah2021machine,bourtoule2021machine} consider \textit{exact} variants of unlearning (with $\eps=\delta=0$), while others \citet{guo2020certified,sekhari2021remember,neel2020descenttodelete,gupta2021adaptive, golatkar2020eternal, golatkar2020forgetting} consider approximate variants.

Despite generally formulating weaker definitions, the algorithms in the literature often satisfy  history independence. In particular, because the constructions in \citet{ginart2019making,cao2015towards} are fully deterministic, it is easy to see that they satisfy adaptive, $(0,0)$-\textit{strong history independence} (by Theorem 1 in \cite{hartline2005characterizing}). %
Some constructions satisfy additional properties, such as storing much less information than the full data set $\data_t$~(e.g., \citet{sekhari2021remember,cao2015towards}).

\para{Adaptive Machine Unlearning (\citet{gupta2021adaptive})}

We
discuss the relationship to one previous definition---that of \citet{gupta2021adaptive}---in detail, since the definition is subtle and the relationship is technically nontrivial.
In our notation, the data structure  maintains several quantities: a data set $\data$, a collection of models $\modelcoll$, some supplementary data $\supp$ used for updating, and the most recently output predictor $\pred$. We say that an implementation of the Updatable ML ADT meets the syntax of 
\citet{gupta2021adaptive} if its  state can be written $\statePhys = (\data_t,\modelcoll_t, \supp_t,\pred_t)$.
Given an initial data set $\data_0$, the execution proceeds for $t$ steps. At each step $i$, the adversary (called the ``update requester'') sends an update $u_i$, which is either an insertion or deletion. The data structure then updates its state to obtain $(\data_i,\modelcoll_i,\supp_i,\pred_i)$ and returns $\pred_i$ to the adversary. 

The condition required by \citeauthor{gupta2021adaptive} is described by three parameters $\alpha,\beta,\gamma>0$. For every time $t$ and initial dataset $\data_0$,  with probability $1-\gamma$ over the updates $\vec u = u_1,...,u_t$ that arise in an interaction between the data structure (initialized with $\data_0$ and randomness $R$) and the adversary, the models $\modelcoll_t$ should be distributed similarly to the models $\modelcoll_0^*$ one would get in a counterfactual execution in which the data structure was initialized with the data set $\data_0^* = \data_t$ (that is, $\vec u $ applied to $ \data_0$)  resulting from the real execution in canonical form (say, in sorted order), and no updates were made. They require:
\begin{equation}\label{eq:adaptMUL}
\Pr_{\vec u =(u_1,...,u_t)}\paren{ \modelcoll_t\big|_{\data_0, u_1,...,u_t} \quad \approx_{\alpha,\beta} \quad  \modelcoll_0^*\big|_{\data_0^*}} \geq 1-\gamma \,
\end{equation}
$\text{where}\ \data_0^*  =( \vec u \text{ applied to }  \data_0).$
\begin{remark}\label{rem:gupta-asymmetry}
\citeauthor{gupta2021adaptive} actually require only a one-sided version of the condition on distributions enforced by $(\alpha,\beta)$-closeness. 
Specifically, no event should be much more likely in the models resulting from a real interaction than in ones generated from scratch with the resulting data set. This condition is inherited from a definition proposed by \citet{ginart2019making}. 

As far as we know, all the actual algorithms designed for machine unlearning satisfy the stronger, symmetric version of the definition in Equation~\eqref{eq:adaptMUL}.
\end{remark}

\begin{proposition}[AHI implies Adaptive Machine Unlearning]\label{prop:AHI_MUL}
    Let $\impl$ be a data structure that implements the Updatable ML ADT and satisfies $(\eps,\delta)$-approximate history independence with fresh randomness (Definition~\ref{def:AHI}). If $\impl$ fits the syntax of machine unlearning, then for all $\beta >0$, it also satisfies \emph{$(\alpha,\beta,\gamma)$-unlearning} (Equation~\eqref{eq:adaptMUL}) with $\alpha =  3\eps$ and $\gamma = \frac{2\delta}{\beta} + \frac{2\delta}{1 - \exp(-\eps)}$.
\end{proposition}
\begin{proof} 
    Consider an Updatable ML data structure $\impl$ that satisfies approximate HI and fits the syntax of \citet{gupta2021adaptive}, that is, its state $\statePhys$ can be written $(\data_t,\modelcoll_t,\supp_t,\pred_t)$. 
    
    Let $\upreq$ be an update requester for the machine unlearning definition. $\upreq$ first chooses a time step $t$ at which it will terminate and an initial data set $\data_0$, and then executes an interaction with $\impl$ during which it makes update queries $\vec u = (u_1,...,u_t)$ and receives predictors $\pred_1,...\pred_t$. Let $\data_t$ be the final logical data set. 
    We define a corresponding adversary $\adv$ for history independence as follows:
    \begin{algorithmic}[1]
        \Procedure{\adv}{given black-box access to $\upreq$}
        \State Run $\upreq$ to obtain $t$ and $\data_0$.
        \State Compute a sequence of updates $\vec u$ such that $\vec u_0 \circ \emptyset = \data_0$ 
            \label{step:initialu}
        \State Output queries $\vec u_0$ followed by the operation $\predict$
        \State Receive $\pred_0$
        \For{$i=1$ to $t$}
            \State Run $\upreq$ on input $\pred_{i-1}$ to get next update $u_i$
            \State Output query $u_i$ followed by $\predict$ 
            \State Receive $\pred_i$
        \EndFor
        \State Output $\seq^* = [\seq] = [\vec u]$ where $ \vec u = \vec u_0\| u_1,...,u_t$, and  $\seq$ is the full realized sequence (namely,  $\vec u$ with interleaved $\predict$ queries).
        \EndProcedure
    \end{algorithmic}
    
    To apply HI, we consider an interaction between $\impl$ and $\adv$, followed by a new execution of $\impl$ with queries $\seq^*$  (which equals $[\vec u]$ for this adversary) and fresh randomness. Let $\data_t$ be the final data set (in both executions). Equation~\eqref{eq:AHI2ML}  implies that
    \[(\seq, \underbrace{\data_t, \modelcoll_t,\supp_t, \pred_t}_{\statePhys}) \approxequiv{\eps,\delta} (\seq, \underbrace{\data_t, \modelcoll^*_0,\supp^*_0, \pred^*_0}_{\statePhys^*}) \, .
    \]
    In particular, since the realized updates $\vec u$ form a subsequence of $\seq$, we have 
    \begin{equation}\label{eq:HI2MUL}
        (\vec u,\modelcoll_t) \approxequiv{\eps,\delta} (\vec u, \modelcoll_0^*) \, .         
    \end{equation}
    To draw the connection with Adaptive Machine Learning, we must define what it means to initialize $\impl$ with a nonempty data set (since such initialization occurs in the definition of Adaptive Machine Unlearning). On input a nonempty data set $\data_0$, we can compute the canonical sequence of updates $\vec u_0$ that would lead to $\data_0$ (as in Step \ref{step:initialu} in $\adv$ above) and run $\impl$ with $\vec u$ as its first queries, before real interactions occur. With this definition, the triple $(\vec u, \modelcoll_t, \modelcoll_0^*)$ is distributed exactly as in  \eqref{eq:adaptMUL}. By Lemma~\ref{lem:conditioning}, with probability at least $1-\gamma$ over $\vec u$, we have $\modelcoll_t |_{\vec u} \approxequiv{\eps',\delta'} \modelcoll_0^*$, where $\eps' = 3\eps$ and $\gamma = \frac{2\delta}{\delta'} + \frac{2\delta}{1 - \exp(-\eps)}$, as desired. 
\end{proof}

The converse implication does not hold, but we claim that it would after some modifications of the definition of \citep{gupta2021adaptive}:

    \begin{itemize}
    \item The indistinguishability condition of \citep{gupta2021adaptive} (Equation~\eqref{eq:adaptMUL}) applies only to the current collection of models $\modelcoll_t$; as such, \textit{it does not constrain the system's entire future behavior but only the distribution of the output at a particular time} (see Section~\ref{sec:definition:discussion}). History independence requires indistinguishability of the entire state $(\data_t,\modelcoll_t,\supp_t)$. 
    \item In \citep{gupta2021adaptive}, the stopping time $t$ is determined before the execution begins. In contrast, the History Independence adversary may choose the stopping point adaptively. 
    \item As discussed in Remark~\ref{rem:gupta-asymmetry}, the indistinguishability condition of \citep{gupta2021adaptive} is one-sided. Our definition of history independence requires the stronger, symmetric guarantee. 
    \end{itemize}

We do not include a formal statement of the equivalence, since by the time one spells out the strengthened model, the equivalence is nearly a tautology. The only  difference not listed above is whether one requires indistinguishability of the states $\statePhys$ and $\statePhys^*$ \textit{conditioned} on the updates $\vec u$---as in Eq.~\eqref{eq:adaptMUL}---or rather indistinguishability of the \textit{pairs} $(\vec u,\statePhys)$ and $(\vec u , \statePhys^*)$---as in Equation~\ref{eq:AHI2ML}. These turn out to be equivalent up to changes of parameters, as shown by Lemma~\ref{lem:conditioning}.

We conjecture that the protocols of \citet{gupta2021adaptive} satisfy adaptive history independence. A formal proof of this statement would require an appropriately symmetric strengthening of the max-information bound of \citep{RogersRST16} (Thm 3.1 in ver.\ 2 on \textit{arXiv}).

\section{Deletion-as-confidentiality and Deletion-as-control}
\label{sec:ggv}

We study the relationship between the notion of deletion as confidentiality from \citet{GGV20} (hereafter ``the GGV definition'') and our notion of deletion-as-control (Definition~\ref{def:del-as-cont}).
Similar to Definition \ref{def:del-as-cont}, deletion-as-confidentiality  also considers two executions, real and ideal. At the end of the two executions, they compare the \emph{view} of the environment  $\view_\env$ and the state of the controller $\s_\cont$. 
Simplifying away (important but technical) details, the real GGV execution is roughly the same as in deletion-as-control. The ideal GGV execution simply drops all messages between $\subj$ and $\cont$. Informally, deletion-as-confidentiality requires that $\env$'s view and $\cont$'s final state are indistinguishable in the real and ideal worlds.

In Section~\ref{sec:ggv:definition}, we define \emph{deletion-as-confidentiality}---the definition closest in spirit to the GGV definition but in our execution model. 
(We cannot directly compare our definition to GGV as they are defined in different models of interaction, stemming from a need in deletion-as-control to `sync' the real and ideal executions.) 
In Section~\ref{sec:ggv:cont_vs_conf}, we show that for many data subjects \subj, deletion-as-confidentiality for \subj implies deletion-as-control for \subj (Theorem~\ref{thm:conf-cont}). Finally, we show that the implication cannot hold for all data subjects due to the differences in the execution models  (Theorems~\ref{thm:conf-cont-sep-one-way} and~\ref{thm:conf-cont-sep-termination}).

\subsection{Definition}
\label{sec:ggv:definition}

Below, we define deletion-as-confidentiality by describing how it differs from deletion-as-control. 
After the definition, we briefly explain how our adaptation of deletion-as-confidentiality differs from that in \cite{GGV20}.

Similar to Definition \ref{def:del-as-cont}, deletion-as-confidentiality also considers two executions---real and ideal. The real execution is almost the same as before, involving the three parties $\cont$, $\env$, and $\subj$. We will be interested in the \emph{view} of the environment $\env$, denoted $V_\env^{real}$ and the state of the controller $\cont$, denoted $\s_\cont^{real}$, at the end of the execution. $V_\env^{real}$ consists of $\env$'s randomness $\Rand_\env$ and the transcript of its messages $\tau_\env^{real} = (\queries_\env^{real}, \answers_\env^{real})$. (Both the view and state are  $\bot$ if the execution does not terminate.)
The GGV definition requires that controller's state and the environment's view are indistinguishable in the real world and in the ideal world wherein the data subject never communicates with anybody.
This definition inherently requires that $\cont$ can never reveal any information about one user's data or participation to another data or participation to another. Any information about user $\subj$ that $\cont$ reveals to the $\env$ becomes part of $\env$'s view in the real world. If $\subj$ later requests deletion, this information would not be part of $\env$'s view in the ideal world.

The first major difference between GGV and deletion-as-control is the ideal execution. The ideal GGV execution also involves the same three parties $\cont$, $\env$, and $\subj$---no dummy party $\dummy$ as in deletion-as-control. The difference between the real and ideal GGV executions is that  ideal GGV execution drops all messages between $\subj$ and $\cont$. This execution results in some view $V_\env^{ideal}$ and some state $\s_\cont^{ideal}$. 

The second major difference between the GGV execution and deletion-as-control concerns how the execution terminates.
Recall that in the deletion-as-control game, the real execution ends as soon as $\cont$ processes $\subj$'s first $\delete$ message (i.e., the next time $\cont$ halts). And the ideal execution ends when all of $\env$'s queries from the real execution have been replayed by the dummy $\dummy$.
In contrast, the real and ideal GGV executions ends when $\env$ sends a special $\finish$ message to $\cont$. At that point, one of two things happens. If $\subj$'s last message to $\cont$ was $\delete$, then the execution immediately terminates. Otherwise, $\subj$ sends $\delete$ to $\cont$, and the execution ends as soon as $\cont$ processes that message.
This difference means that the execution's end time can depend on \subj's view in deletion-as-control, but not in GGV. Theorem~\ref{thm:conf-cont-sep-termination} leverages this gap to construct a controller that satisfies GGV but not deletion-as-control.
This difference reflects our two fundamentally different approaches to defining the ideal execution.
\subj has no affect on the GGV ideal execution, and hence \env must end the execution.
Deletion-as-control introduces a mechanism to ``sync'' the real and ideal executions, thereby allowing $\subj$ to end both executions.

The last major difference between GGV and deletion-as-control is the indistinguishability requirement between the real and ideal executions. 
The GGV definition requires that the joint distribution of the variables $(V_\env^{real},\s_\cont^{real})$ is close to that of $(V_\env^{ideal},\s_\cont^{ideal})$.  That is, the view of the environment $\env$ and the state of the controller $\cont$ should be nearly the same irrespective of whether the subject $\subj$ sent its data to the controller and then deleted it, or did not send its data at all. In contrast, deletion-as-control imposes no requirement on $V_\env$.
Moreover, GGV only requires indistinguishability for $\subj$'s that do not send any messages to $\env$. Without this restriction the definition would be far too limiting. For most useful controllers, $\subj$ would be able to tell whether its messages are being delivered to $\cont$ or not (for example, by receiving an acknowledgement of data receipt), and if it conveys this information to $\env$, $\env$'s view will become different in real and ideal executions.

\begin{definition}[$(\eps, \delta)$ Deletion-as-Confidentiality for $\subj$]
    \label{def:del-as-conf-subj}
    For a data subject $\subj$, a controller $\cont$ is $(\eps, \delta)$-\emph{deletion-as-confidentiality compliant for $\subj$} if for all $\env$, in the executions involving $\cont$, $\env$, and $\subj$, 
    \begin{equation*}
        (V_\env^{real},\s_\cont^{real}) \approxequiv{\eps, \delta} (V_\env^{ideal},\s_\cont^{ideal}).
    \end{equation*}
\end{definition}

\begin{definition}[Deletion-as-Confidentiality, adapted from \cite{GGV20}]
    \label{def:del-as-conf}
    Let $\subjSilent$ be the set of data subjects $\subj$ that never send any messages to $\env$.
    $\cont$ is \emph{$(\eps, \delta)$-deletion-as-confidentiality compliant} if it is $(\eps, \delta)$-deletion-as-confidentiality compliant for all $\subj \in \subjSilent$. 
\end{definition}

The version of deletion-as-confidentiality  in Definition~\ref{def:del-as-cont} differs from the original definition of \citet{GGV20} in a few important ways.
First, \citep{GGV20} allows users to request deletion of the information shared in specific interactions between $\subj$ and $\cont$.
To do this, they define protocols that produce protocol-specific deletion tokens. Simplifying, we take deletion of a user's data to be all or nothing. Second, \citep{GGV20} allows users to delete many times, whereas we focus on a single deletion.
Third, \cont's randomness tape is not read-once in \citep{GGV20}. As explained in  Footnote~\ref{footnote:randomness}, this allows a controller to evade deletion by encrypting its state with its randomness. Finally, the execution model of \citep{GGV20} does not have authenticated channels. As they show, authentication is necessary for non-trivial functionalities. In light of this, we chose to build authentication into our execution model  (in the form of channel IDs).

\subsection{Control vs.\ Confidentiality} \label{sec:ggv:cont_vs_conf}
In spirit, deletion-as-confidentiality imposes a stronger indistinguishability requirement than deletion-as-control. The former requires that no information about the deleted data is ever revealed, whereas the latter only requires that the effect of the deleted data is not present \emph{after} the deletion happens. One might thus expect that any $\cont$ that satisfies Definition~\ref{def:del-as-conf} would also satisfy Definition~\ref{def:del-as-cont}. 
But deletion-as-control captures more general environments and data subjects than deletion-as-confidentiality. First, deletion-as-control allows $\subj$ to communicate freely with $\env$, whereas deletion-as-confidentiality does not (i.e., $\subj\in\subjSilent$). Second, deletion-as-control imposes a requirement as soon as $\subj$ deletes, whereas deletion-as-confidentiality only requires indistinguishability after \env terminates the execution (which may be much later).

Let $\subjSilent$ be as in Definition~\ref{def:del-as-conf} and $\subjDummy$ be the set of $\subj$ that only \delete when instructed by $\env$. Let $\subjLift = \subjSilent \cap \subjDummy$.

\begin{theorem} \label{thm:conf-cont}
    For any $\cont$ and any $\subj \in \subjLift$, if $\cont$ is $(\eps, \delta)$-deletion-as-confidentiality compliant for \subj, then it is also $(\eps, \delta)$-deletion-as-control compliant for \subj. 
\end{theorem}

Theorems~\ref{thm:conf-cont-sep-one-way} and~\ref{thm:conf-cont-sep-termination} show that the restriction in Theorem~\ref{thm:conf-cont} that $\subj \in \subjLift = \subjDummy \cap \subjSilent$ is necessary.

\begin{proof}[Proof of Theorem~\ref{thm:conf-cont}]

Fix any controller \cont that is $(\eps, \delta)$-deletion-as-confidentiality compliant for subset of $\subj$'s that call \delete only when instructed to by \env. Then, in particular, \cont is deletion-as-confidentiality compliant for \emph{deterministic} environments \env and data subjects \subj in this class. By Lemma~\ref{lem:derandomize}, it is sufficient to consider deterministic $(\env, \subj)$ pairs to prove deletion-as-control compliance. 

Fix a deterministic environment \env and data subject \subj that calls \delete only when instructed to by \env. Then, by Definition~\ref{def:del-as-conf-subj}, we have
\begin{equation}\label{eq:GGV}
        (V_\env^{real},\s_\cont^{real}) \approxequiv{\eps, \delta} (V_\env^{ideal},\s_\cont^{ideal}),
    \end{equation}
where the randomness in the two distributions is only over the controller's randomness $\Rand_\cont$. Furthermore, since \subj only deletes when \env instructs it to, the GGV game and deletion-as-control game result in the same $\s_\cont$ upon termination.

We will construct a simulator \Sim that satisfies Definition~\ref{def:del-as-cont} for \cont, \env, \subj. Recall that \Sim's task is as follows. For some randomly sampled value $\Rand_\cont$ and the resulting transcript $\tau = (\queries, \answers)$ and controller state $\s_\cont$, the simulator is given $\Rand_\cont, \s_\cont$, and $\queries_\env$ (the sequence of messages sent by \env to \cont). The simulator is required to produce a $\Rand'_\cont$ such that when \cont interacts with the dummy environment $\dummy(\queries_\env)$  using randomness $\Rand'_\cont$, it results in the same state $\s_\cont$ with probability at least $1-\delta$; further, $\Rand'_\cont \approxequiv{\eps, \delta} \Rand_\cont$ (with randomness coming from sampling the two).

In our case, we will use the ``default simulator'' (from Definition~\ref{def:default-sim}), which works as follows:
\begin{align*}
    &\Sim(\queries_\env, \rand_\cont, \s_\cont):\\
    & \qquad \text{Return }\Rand'_\cont \sim \distrand\big|_{\cont(\queries_\env;\Rand') = \s_\cont} \text{ if such an $\Rand'_\cont$ exists}\\
    &\qquad \text{Otherwise, return $\Rand'\sim \distrand$}.
\end{align*}
If such an $\Rand'_\cont$ exists, then the final state of \cont when interacting with $\dummy(\queries_\env)$ using randomness $\Rand'_\cont$ will be equal to $\s_\cont$. This is because, in \cont's view, this interaction is identical to the ideal execution of Definition~\ref{def:del-as-conf-subj} with $\cont$ using randomness $\Rand'_\cont$ and $\env, \subj$. It remains to prove that  $\Rand'_\cont$ exists with probability at least $1-\delta$ and that the distribution of $\Rand'_\cont$ is $(\eps, \delta)$-close to that of $\Rand_\cont$. 

We can now apply Lemma~\ref{lem:approx_coupling} to prove both of the statements. Define the following deterministic functions:
\begin{align*}
        f(X) &= (V_\env^{real}(X), \s_\cont^{real}(X)), \text{ and}\\
        g(Y) &=  (V_\env^{ideal}(Y), \s_\cont^{ideal}(Y)),
    \end{align*}
where $V_\env^{real}(X)$ denotes the view of \env when $X$ is used as the randomness for $\cont$ in the real world. Let $X = \Rand_\cont$, $Y = \Rand_\cont$, and let $P\equiv Q$ be the distribution that both $X$ and $Y$ are sampled from. Then, by Equation~\eqref{eq:GGV}, $f(X) \approxequiv{\eps, \delta} g(Y)$. Lastly, the simulator \Sim defined samples $\Rand_\cont'$ from the conditional distribution described in Lemma~\ref{lem:approx_coupling}. Thus, by the lemma, we have that $\Rand_\cont' \approxequiv{\eps, \delta} \Rand_\cont$ and that $(V_\env^{real}(X), \s_\cont^{real}(X)) = (V_\env^{ideal}(Y), \s_\cont^{ideal}(Y))$ with probability at least $1-\delta$. In particular, we have that $\s_\cont^{real} = \s_\cont^{ideal}$ with probability at least $1-\delta$. 

Thus, the simulator \Sim for \cont satisfies both conditions of $(\eps, \delta)$-deletion-as-control for deterministic $(\env, \subj)$ and by Lemma~\ref{lem:derandomize}, \cont also satisfies $(\eps, \delta)$-deletion-as-control for randomized $(\env, \subj)$, where \subj calls \delete only when instructed to by \env. Lastly, since \Sim does not rely on the code of \subj or \env, we can use the same simulator for \emph{all} \subj's in the class of data subjects that calls \delete only when instructed to by \env.
\end{proof}

\begin{theorem} \label{thm:conf-cont-sep-one-way}
There exists \cont and $\subj \in \subjDummy \setminus \subjSilent$ such that $\cont$ satisfies $(0,0)$-deletion-as-confidentiality, but $\cont$ does not satisfy $(\eps',\delta')$-deletion-as-control for $\subj$ any $\eps'<\infty$, $\delta'<1$.
\end{theorem}

\begin{proof}

Consider $\cont$ that implements a write-only SHI dictionary modified as follows. Recall that communications in our execution model are in the form $(\cid,\msg)$, where $\cid$ is a channel ID known only to the communicating parties. If $\cont$ receives message $\msg = \cid$ with channel ID $\cid'\neq \cid$, then it ignores all subsequent \delete messages of the form $(\cid, \delete)$. In words, the user communicating on channel $\cid'$ can instruct $\cont$ to ignore deletion requests from the user communicating on channel $\cid$.

It is easy to see that $\cont$ satisfies $(0,0)$-deletion-as-confidentiality.
First, $\view_\env$ is identical in the real and ideal executions because the dictionary is write-only and $\cont$ sends no messages.
Second $\s_\cont$ is identical in the real and ideal executions because the dictionary is SHI.

Recall $\subjDummy \setminus \subjSilent$ is the set of data subjects $\subj$ that send messages to $\env$ and only $\delete$ when instructed by $\env$.
Consider $\subj$ that inserts itself into the dictionary, sends its $\cid_\subj$ to \env, and then calls \delete. 
Consider $\env$ that sends $\cid_\subj$ to \cont as soon as it receives it from \subj.
By construction, $\s_\cont$ contains $\subj$'s data at the end of the real execution, but $\s'_\cont$ in the dummy execution does not. 
Thus, if the implementation $\dict$ is perfectly correct, there are no $\Rand_\cont$, $\Rand'_\cont$ such that $\s_\cont = \s'_\cont$. This violates Definition~\ref{def:del-as-cont} for any $\delta<1$ and any $\eps<\infty$.
\end{proof}

\begin{theorem} \label{thm:conf-cont-sep-termination}
For any $\delta>0$, there exists \cont and $\subj\in \subjSilent$ such that (i) \cont satisfies  $(0,\delta)$-deletion-as-confidentiality for \subj, and (ii) $\cont$ does not satisfy $(\eps',\delta')$-deletion-as-control for $\subj$ for any $\eps<\infty$, $\delta'<1$.
\end{theorem}

\begin{proof}[Proof of Theorem~\ref{thm:conf-cont-sep-termination}]
Let $T > 1/\delta$. Consider $\cont$ implementing a write-only SHI dictionary $\dict$ modified as follows. Upon receiving the first message from a channel $\cid$, sample $t_{\cid} \in [1,T]$ uniformly at random, store it in the dictionary with key $\cid$, and reply with $t_\cid$. Upon receiving \delete from \cid, check the current size of the dictionary. If $|\dict| = t_\cid$, remove \cid the next time \cont is activated. Otherwise, remove \cid immediately.

Consider $\subj$ that (i) sends $\ins{}$ to \cont on its first activation, (ii) receives $t_{\cid_\subj}$ in response, and (iii) sends \delete to \cont after $t_{\cid_\subj}-1$ additional activations.

\begin{claim}
$C^*$ satisfies $(0,\delta)$-deletion-as-confidentiality for $\subj^*$. 
\end{claim}

\begin{proof}
Fix $\env$ and $\subj \in \subjSilent$.
Because $\dict$ is write-only, $\view_\env$ is exactly the same in the real and ideal deletion-as-confidentiality executions. 
That view consists of $k$ calls to $\ins{}$ sent to $\cont$ and the responses $t_1, \dots, t_k$. 

In contrast, $\s^{real}_\cont$ may differ at the end of the real and ideal executions. But only if $|\dict| = t_{\cid_\subj}$ when \subj calls \delete.
In particular, $\s^{real}_\cont \neq \s^{ideal}_\cont$ requires $\env$ to make exactly $t_{\cid_\subj}-1$ calls to $\cont.\ins{}$ over the course of the execution. 
But $\env$ has no information about $t_{\cid_\subj}$ besides its prior distribution. From $\cont$ it receives only the $t_i$, which are independent. From \subj it receives nothing, as $\subj \in \subjSilent$.
Therefore, the probability that $\s^{real}_\cont \neq \s^{ideal}_\cont$ is at most $1/T$.
\begin{align*}
&d_{TV}\biggl((V_\env^{real},\s_\cont^{real}), (V_\env^{ideal},\s_\cont^{ideal})\biggr) \\
& =
d_{TV}\biggl(\s_\cont^{real}, \s_\cont^{ideal}\biggr) < \frac{1}{T}
\end{align*}
Hence $(V_\env^{real},\s_\cont^{real})\approxequiv{0,\frac{1}{T}} (V_\env^{ideal},\s_\cont^{ideal}).$
\end{proof}

\begin{claim}
\label{claim:sub-claim-to-ggv-sep-thm}
$C^*$ does not satisfy $(\eps',\delta')$-deletion-as-control for $\subj^*$ for any $\eps<\infty$, $\delta'<1$.
\end{claim}

\begin{proof}
Consider $\env$ that repeats two operations: (i) activate $\subj$, and (ii) send $\ins{}$ to \cont along a new channel.
The execution $\exec{\cont(\Rand_\cont), \subj, \env}$ always ends with $\subj$ calling $\delete$ when $|\dict|= t_{\cid_\subj}$. By construction, $\dict$ does not remove $\subj$ from $\dict$. 
So at the end of the real execution, $\s_\cont$ contains the dictionary $\dict$ that includes $\subj$.

In contrast, the dummy execution $\exec{\cont(\Rand'_\cont), \dummy(\queries_\env)}$ includes no calls to $\ins{}$ from $\subj$. So at the end of the dummy execution,  $\s'_\cont$ contains the dictionary $\dict'$ that does not include $\subj$.
Thus, if the implementation $\dict$ is perfectly correct, there are no $\Rand_\cont$, $\Rand'_\cont$ such that $\s_\cont = \s'_\cont$. This violates Definition~\ref{def:del-as-cont} for any $\delta<1$ and any $\eps<\infty$.
\end{proof}

This completes the proof of Theorem~\ref{thm:conf-cont-sep-termination}.
\end{proof}

\section{Differential Privacy and Deletion-as-control}\label{sec:DP}

This section describes two ways of compiling $(\eps,\delta)$-differentially private mechanisms $\mech$ into controllers satisfying $(\eps,\delta)$-deletion-as-compliance. 
The first applies to DP mechanisms that are run in a batch setting on a single, centralized dataset: data summarization, query release, or DP-SGD, for example (\Cref{sec:one-shot-DP}).
The second applies to mechanisms satisfying pan-privacy under continual release  (\Cref{sec:implementPP}). Along the way we define non-adaptive event-level pan privacy (one intrusion) with continual release, first defined by \citet{chan2011private, dwork2010pan}, and an adaptive variant, building on the adaptive continual release definition of \citet{jain2022price}~(\Cref{sec:PPdef}). 
Both of our compilers make use of a SHI dictionary $\dict$ (\Cref{sec:SHI_DICT}).

There is a strong intuitive connection between (approximate) history independence and deletion-as-control, as illustrated by the results of the previous section.
This intuition is so strong that many prior works on machine unlearning \textit{essentially equate} deletion with history independence (\Cref{sec:MUL-2-AHI}).

However, \emph{the examples of this section show that our notion of deletion-as-control is much broader than history independence}. For instance, consider a controller as follows. At some time $t_0$ the controller computes a  differentially-private approximation $\out_{t_0}$ to the current number of users in the data set, and the controller stores  $\out_{t_0}$ and makes it available at all later times. Intuitively, this controller does not satisfy any version of history independence: even if every user request deletion at time $t_0+1$, the stored $\out_{t_0}$ is unchanged---making this history easy to distinguish from one where the all users deleted at time $t_0-1$. It could, however, still satisfy deletion-as-control.

\para{Keeping time}
Throughout this section we consider systems with a global clock. This allows for controllers that publish the number of weekly active users, say. For simplicity and generality, we allow the environment to control time. Specifically, we introduce a special query $\tick$ that only $\env$ can send to $\cont$ to increment the clock.

\subsection{Batch differential privacy}
\label{sec:one-shot-DP}

Let $\mech:\dict \mapsto \mech(\dict)$ be a non-interactive differentially private mechanism $\mech$.

\Cref{alg:one-shot-DP-controller} defines a simple controller $\contdp$ that satisfies deletion-as-control.

It works in three phases: before $\tick$, during $\tick$, and after $\tick$.
At the beginning, $\contdp$ populates a dataset $\dict$ stored as a SHI dictionary from its input stream, returning $\bot$ in response to every query. When it receives the $\tick$, it evaluates $\mech(\dict)$, stores the result as $\out$, and erases the dictionary $\dict$. For all future queries, $\contdp$ simply returns $\out$.
We assume for simplicity that the mechanism $\mech$ is a function only of the logical contents of $\dict$, and is independent of its memory representation.

\begin{algorithm}
\caption{$\contdp$ satifying deletion-as-control from batch DP mechanism $\mech$}
\label{alg:one-shot-DP-controller}
    \begin{algorithmic}[1]
        \Procedure{Initialize}{} \Comment{Run on first activation}
            \State Initialize a SHI dictionary $\dict$
            \State Initialize $\out = \bot$ \Comment{$\bot$ is  never returned by $\mech$}
        \EndProcedure
        \Statex
        \Procedure{Activate}{$\op(\id)$}
            \If{$(\op \neq \tick) \land (\out = \bot)$}
                \State \texttt{// before tick}
                \State $\dict.\op(\id)$ \Comment{Insert/modify/delete item}
                \State Return $\bot$
            \ElsIf{$(\op = \tick)$}
                \State \texttt{// during tick}
                \State Set $\out \gets \mech(\dict)$
                \State Delete $\dict$ \Comment{eg: iteratively call $\dict.\delete(\id)$}
                \State Return $\out$
            \Else
                \State \texttt{// after tick}
                \State Return $\out$
            \EndIf
        \EndProcedure
    \end{algorithmic}
\end{algorithm}

\begin{proposition}\label{prop:one-shot-DP} ~
\begin{enumerate}
    \item If $\mech$ is $(\eps, \delta)$-DP, then $\contdp$ satisfies $(\eps, \delta)$-deletion-as-control.
    \item For any $\eps>0$, suppose $\mech$ is the Laplace mechanism with parameter $\eps$ applied to a count of the number of record in its input. Then $\contdp$ satisfies $(\eps, \delta)$-deletion-as-control but is \emph{not} $(\eps',\delta')$-HI for any  $\epsilon' <\infty$ and $\delta' < 1$.
\end{enumerate}
\end{proposition}

The second part of the proposition really applies to any DP mechanism that releases useful information about its inputs—the argument relies just on the fact that $\mech$ acts differently on the empty data set than it does on a data set with many records.
In particular, the DP Machine learning example from Section~\ref{sec:touchstone}, which trains a model using DP-SGD, satisfies deletion-as-control. 

\begin{corollary}\label{cor:dp_ml} The $(\eps, \delta)$-DP Machine Learning touchstone controller (Section~\ref{sec:touchstone}) satisfies $(\eps, \delta)$-deletion-as-control.
\end{corollary}

\begin{proof}[Proof of Proposition~\ref{prop:one-shot-DP}]
(1) To reduce clutter, let $\cont = \contdp$ throughout this proof.
The simulator is described in \Cref{alg:one-shot-DP-controller:sim}.
We must show that for all deterministic $\env$ and $\subj$ (i) $\Rand'_\cont \approxequiv{\eps,\delta} \Rand_\cont$, and (ii) $(\dict', \out') = (\dict, \out)$ with probability at least $1-\delta$.

Fix deterministic $\env$ and $\subj$.
Until a $\tick$ query, $\cont$'s output is determinstically $\bot$.
Hence, all queries  that precede $\tick$ are fixed in advance---including whether the sequence contains a $\tick$ at all. If there is no $\tick$, then the whole sequence is fixed.

We consider two cases depending on whether the real sequence of queries $\queries$ includes $\tick$.
Suppose there is no $\tick$ in $\queries$. By construction, $\cont$'s state $\out = \bot$ and $\dict$ may be non-empty. $\Sim(\queries_\env,\Rand_\cont,\out)$ outputs $\Rand'_\cont = \Rand_\cont$. Consider $\dict$ in the real execution resulting from queries $\queries$ and randomness $\Rand_\cont$, and $\dict'$ in the ideal execution resulting from $\queries_\env$ and $\Rand'_\cont$.
The only difference between $\queries$ and $\queries_\env$ is $\subj$'s queries.
Because $\queries$ ends with $\delete$ from $\subj$, and dictionaries support logical deletion, the logical contents of $\dict$ and $\dict'$ are identical. Applying Theorem 1 from \cite{hartline2005characterizing}\footnote{The theorem states that reversible $(0,0)$-SHI data structures have canonical representations, up to initial randomness.}, $\dict' = \dict$. In this case, the simulation is perfect.

Suppose there is a $\tick$ in $\queries$. By construction,   $\cont$'s state $\dict$ is empty and $\out \neq \bot$. The simulator samples $\Rand'_\cont$ uniformly conditioned on $\out' = \out$ (if such $\Rand'_\cont$ exists, uniformly otherwise).
Consider $\dict$ in the real execution resulting from queries $\queries$, and $\dict'$ in the ideal execution resulting from $\queries_\env$. $\dict$ and $\dict'$ represent fixed neighboring datasets: differing only on $\subj$.
Take $f(\Rand) = \mech(\dict;\Rand)$ and $g(\Rand) = \mech(\dict'; \Rand)$. Because $\mech$ is $(\eps,\delta)$-DP, $f(\Rand) \approxequiv{\eps,\delta} g(\Rand)$. The simulator samples $\Rand'_\cont$ as in the statement of \Cref{lem:approx_coupling}. Applying that lemma completes the proof.

(2) When $\mech$ is the Laplace mechanism, $\contdp$ satisfies deletion as control by part (1). To see why it does not satisfy even weak, nonadaptive history independence, consider a (nonadaptive) sequence $\seq_n$ of $n$ distinct insertions, followed by a \tick, followed by $n$ corresponding deletions. The minimal equivalent sequence $\seq^*=[\seq_n]$ consists only of the $\tick$. The value $\out$ generated at the tick will follow $n + \Lap(1/\epsilon).$ If we consider the event $E_n=\set{\out < n/2}$, then $E_n$ occurs with probability at least $1-\exp(\eps n/2)$ in the real world. However, if we were to run a dummy execution of the controller $\contdp$ with $\seq^*$, the probability of $E_n$ would be at most $\exp(\eps n)$, regardless of whether we use the same randomness $\Rand$ as in the real world or a fresh string $\Rand'$. Suppose, for contradiction, that $\contdp$ satisfies $(\eps',\delta')$-HI. Then
\begin{align*}
1-\exp(-\eps n) \leq \Pr_{\text{real}}(E_n) &\leq \exp(\eps')\cdot  \Pr_{\text{dummy}}(E_n) \\
&\leq \exp(\eps' - \eps n) + \delta' \, .
\end{align*}
For any $\eps'>0$ and $\delta'<1$, we can get a contradiction by choosing  $n> \frac 1 \eps \ln \paren{\frac{e^{\eps'}+1}{1-\delta'}}$.
\end{proof}

\begin{algorithm}[h]
\caption{$\Sim$ for $\cont=\contdp$}
\label{alg:one-shot-DP-controller:sim}
    \begin{algorithmic}[1]
        \State \textbf{input} $\Rand_\cont$, $\queries_\env$, $\out$
        \If{$\tick \notin \queries_\env$}
            \State $\Rand'_\cont \gets \Rand_\cont$
        \Else
            \State Run the default $\Sim$ (Def.~\ref{def:default-sim})
            
            Namely, sample $\Rand'_\cont$ uniformly conditioned on $\out' = \out$, where $\out' = \cont(\queries_\env; \Rand_\cont')$ is the output of $\cont$ in the ideal execution (with queries $\queries_\env$). If no such $\Rand'_\cont$ exists, sample $\Rand'_\cont$ uniformly at random.
        \EndIf
        \State Return $\Rand'_\cont$
    \end{algorithmic}
\end{algorithm}

\subsection{Adaptive pan-privacy with continual release} \label{sec:PPdef}
In this section, we formalize the definition of adaptive pan-privacy with continual release (against a single intrusion), extending the non-adaptive versions originally defined in \cite{chan2011private, dwork2010pan}. The continual release setting concerns an \emph{online controller} that processes a stream of elements and produces outputs at regular time intervals. We say two streams are (event-level) \emph{neighboring} if they differ in at most one stream item. We consider privacy against an adversary that can adaptively choose neighboring streams (one of which is processed by the online mechanism) and inspect the internal state of the online algorithm once. This adaptive game defines the adversary's view. Informally, adaptive pan-privacy under continual release requires that, for all adaptive adversaries, the adversary's view for the two streams is $(\eps,\delta)$-close.

\paragraph{Preliminaries}
A \emph{stream} $\Ds \in  (\{\reg, \intru, \chall\} \times \universe)^*$
consists of elements of the form $(\code, \op(\id))$.
The \emph{control code} $\code \in \{\reg, \intru, \chall\}$ controls the execution of the pan-privacy continual release security game and is never seen by the mechanism itself (see \Cref{alg:adaptive_PP_game}). The \emph{data operation} $\op(\id)\in \universe$ is sent to the pan-private mechanism. The universe of data operations $\universe$ includes a special \tick operation (which requires no \id tag) which indicates to the mechanism that one time step has passed.

\begin{definition}[Online Algorithm]
An \emph{online algorithm} \mech is initialized with a time horizon $T$ and is defined by an internal algorithm $I$.
Algorithm \mech processes a stream of elements through repeated application of $I:(\op(\id),\s_\mech) \mapsto (\s'_\mech, \out)$, which (with randomness) maps a stream element and the current internal state to a new internal state and an output (which may consist of $\bot$).
For simplicity, we only consider online algorithms that produce outputs when a $\tick$ occurs. That is, $\op(\id)\neq \tick$ implies that $\out=\bot$.
The internal state of \mech includes a clock which counts the number of $\tick$s that have previously occurred in the stream. When \mech's internal number of clock ticks reaches $T$, it stops processing new stream elements and always outputs $\bot$.

\end{definition}

Any stream $\Ds$ with a single instance of $\chall$ naturally gives rise to two neighboring sequences of data operations $\Ds^{(\IN)}$ and $\Ds^{(\OUT)}$ which include or exclude the challenge operation, respectively.

\begin{definition}[Event-level neighboring sequences]
For stream $\Ds$ containing a single instance of $\chall$, we define two \emph{event-level neighboring sequences} $\Ds^{(\IN)}$ and $\Ds^{(\OUT)}$ as follows:
\begin{align*}
\Ds^{(\IN)} &:= \{\op(\id) \mid (\reg, \op(\id)) \in \Ds   \text{ or } \\
& \qquad (\chall, \op(\id)) \in \Ds\} \\
\Ds^{(\OUT)} &:= \{\op(\id) \mid (\reg, \op(\id)) \in \Ds\}.
\end{align*}
\end{definition}

\paragraph{Event-level pan privacy}

We formalize the definition of non-adaptive pan-privacy with continual release (against a single intrusion), based on \citep{chan2011private, dwork2010pan}. Although our version of the definition is identical in spirit to previous definitions, our setting has two complications which are notationally challenging but not conceptually so. First, the algorithm \mech may process any number of sequence items before producing an output. Second, we consider sequences as neighboring if they differ by insertion or deletion of items, which causes a discrepancy in sequence indices. We have chosen to write the non-adaptive definition so that it is clearly a special case of the adaptive version, defined later in the section.

Let $\Ds$ be a stream with a single  $\intru$.
For online algorithm $\mech$ and $\side \in \{\IN, \OUT\}$, we denote by
$V_{\mech,\Ds}^{(\side)} = (\s,\Out)$
the view of an attacker who
sees (i) the sequence of outputs $\Out$ produced by $\mech(\Ds^{(\side)})$, and (ii) a snapshot of $\s_\mech$ of the internal state of \mech at the time indicated by $\intru$ in $\Ds$.

\begin{definition}[Non-Adaptive Event-Level PP with CR~\citep{chan2011private,dwork2010pan}]
An online algorithm \mech satisfies $(\eps, \delta)$-non-adaptive event-level pan-privacy in the continual release model if for all streams \Ds with at most one $\chall$ and at most one $\intru$:
\begin{equation*}
    V_{\mech, \Ds}^{(\IN)} \approxequiv{\eps, \delta} V_{\mech,\Ds}^{(\OUT)}.
\end{equation*}
\end{definition}

\paragraph{Adaptive pan privacy} We define adaptive event-level pan privacy with continual release against a single intrusion, following the continual-release model of \citet{jain2022price} (which did not allow for intrusions). In this game, the adversary \adv decides when to increment timesteps (i.e., by querying $(\reg, \tick)$), when to issue a single challenge query, when to intrude, and when to terminate the game (i.e., using $\code = \bot$). Depending on whether $\side$ is $\IN$ or $\OUT$, \mech is either given the challenge query or not. Thus, the neighboring data streams we consider differ by the {insertion or deletion} of a single query.
We require that for the $\chall$ query, $\op \neq \tick$. We assume for simplicity that $\mech(\op(\id))$ produces no output unless $\op =\tick$.

\begin{algorithm}
\caption{Privacy game $\Pi_{\mech,\adv}^{(\side)}$ for the adaptive event-level pan privacy with continual release model. The game halts if any assertion fails. }\label{alg:adaptive_PP_game}
    \begin{algorithmic}[1]
        \Procedure{RunGame}{time horizon $T \in \N$}
            \State{\texttt{// initialize global variables}}
            \State $\mech.\textsc{Initialize}(T)$
            \State $\mathsf{intruded} \gets \false$
            \State $\mathsf{challenged}\gets \false$
            \Statex
            \State{\texttt{// interaction between $\adv$ and $\mech$}}
            \State $\out \gets \bot$
            \Repeat
            \State $(\code, \op(\id)) \gets \adv(out)$
            \State $\out \gets \textsc{Activate}(\code, \op(\id))$
            \Until{$\code = \bot$}
        \EndProcedure

        \Statex
        \Procedure{Activate}{$\code, \op(\id)$}

            \If{$\code = \intru$}
                \State \textbf{Assert:}  $\mathsf{intruded} = \false$  \label{line:pp:assert-intrusion}
                \State $\mathsf{intruded} \gets \true$
                \State \Return{$\s_\mech$}

            \EndIf
            \Statex
            \If{$\code = \reg$}
                \State \Return{$\mech(\op(\id))$} \label{line:pp:call-to-m-1} \Comment{If $\op\neq \tick$, then $\mech(\op(\id))=\bot$}

            \EndIf
            \Statex
            \If{$\code = \chall$ }
                \State \textbf{Assert:} $(\mathsf{challenged} = \false)$ and $(\op \neq \tick)$ \label{line:if}  \label{line:pp:assert-challenge}
                \State $\mathsf{challenged} \gets \true$
                \If{$\side = \IN$ \label{line:pp:side-in}}
                    \State Run $\mech(\op(\id))$ \Comment{By assumption $\op(\id) \neq \tick$ and $\mech(\op(\id))= \bot$}  \label{line:pp:call-to-m-2}
                \EndIf
                \State \Return{$\bot$}
            \Statex
            \Else
                \State \Return{$\bot$}
            \EndIf

        \EndProcedure
    \end{algorithmic}
\end{algorithm}

\begin{definition}
We denote by \emph{$V_{\mech,\adv}^{(\side)}$} the \emph{view of $\adv$} in the game $\Pi_{\mech,\adv}^{(\side)}$, consisting of $\adv$'s internal randomness and the transcript of all messages $\adv$ sends and receives.\footnote{One could instead define $\adv$'s view as its internal state at the end of the game. Our version of the view contains enough information to compute that internal state and is simpler to work with.}
\end{definition}

\begin{definition}[Adaptive Event-Level PP with CR] \label{def:adaptive_PP}
A mechanism $\mech$ is \emph{$(\epsilon,\delta)$-DP in the adaptive event-level pan-privacy with continual release model} if for all adversaries~$\adv$,
$$V_{\mech,\adv}^{(\IN)} \approxequiv{\eps,\delta} V_{\mech,\adv}^{(\OUT)}. $$
\end{definition}

\begin{remark}
The definition above makes sense in a setting where the state $\s_{\mech}$ does not have information about the mechanism's \textit{future} random coins. If the adversary could deduce future randomness, then security would be unachievable: an adversary could intrude at time 0, learn all the algorithm's randomness, and then easily tell whether the mechanism had received $x^*$ or not.

An alternative, which does allow for the adversary learning future randomness, is to require that the challenge is possible only \textit{before intrusion} (specifically, Line~\ref{line:if} additionally assert that  $\mathsf{intruded} = \false$). This definition is satisfiable by interesting mechanisms---for example, a mechanism that initializes a counter with Laplace noise, then adds subsequent stream elements $x$ in $[0,1]$ to the counter, and outputs the final value on the first clock tick, satisfies the weaker, alternative definition.
\end{remark}

\subsection{From pan-privacy to deletion-as-control}
\label{sec:implementPP}

\Cref{alg:event-to-user} describes a general transformation from an event-level pan private algorithm to a controller satisfying deletion-as-control.
It uses as a building block a strongly history independent dictionary $\dict$ (\Cref{sec:SHI_DICT}).
For each $\id$, the controller passes the first operations $\op(\id)$ to the underlying pan private mechanism $\mech$. It also passes all $\tick$, which produce output. $\contpp$ uses the dictionary $\dict$ to check whether a given $\id$ has already issued a query to $\mech$.
To delete, $\id$ is removed from $\dict$ but $\mech$ is unaffected. (Note that a deleted user can then issue a new query to $\mech$; we make no guarantees for such users.)

We assume that the randomness tape of $\contpp$ defined in \Cref{alg:event-to-user} is partitioned into two independent strings: $\Rand_\mech$ to be used by the mechanism \mech, and $\Rand_\dict$ to be used by $\dict$. Additionally, we assume that the state of $\contpp$ can be partitioned into one part containing $\s_\mech$ and another part with $\s_\dict$. Beyond from $\s_\mech$ and $\s_\dict$, the controller $\contpp$ uses only ephemeral state and no additional randomness.

\begin{algorithm}[H]%
\caption{$\contpp$ Deletion Compliant Controller Based on %
Event-Level Pan-Private Algorithm}
\label{alg:event-to-user}
    \begin{algorithmic}[1]
        \Procedure{Initialize}{Time horizon $T \in \N$, privacy parameters $\eps, \delta$, query access to event-level pan-private algorithm $\mech$}
            \State Initialize $\mech$ with parameters $(\eps, \delta, T)$
            \State Initialize an empty SHI dictionary $\dict$
        \EndProcedure
        \Statex
        \Procedure{Activate}{$\op(\id)$}
            \If{$\op = \tick$} \Comment{By assumption $\subj$ can never query $\tick$}
                \State \Return{$\mech(\tick)$} \label{line:tick-M}
            \EndIf
            \If{$\op = \delete$}
                \State $\dict.\del{\id}$ \Comment{If $\id \notin \dict$, nothing happens.}
                \State \Return{$\bot$}
            \EndIf
            \If{$\id \in \dict$}
                \State  \Return{$\bot$}
            \EndIf
            \If{$\id \notin \dict$}
                \State $\dict.\ins{\id}$
                \State $\mech(\op(\id))$ \label{line:calls-to-M}
                \State \Return{$\bot$} \Comment{If $\op\neq \tick$, then $\mech(\op(\id))=\bot$}
            \EndIf
        \EndProcedure
    \end{algorithmic}
\end{algorithm}

\begin{theorem}
\label{thm:PPshi_filter}
If a controller \mech satisfies $(\eps,\delta)$-adaptive event-level pan-privacy with continual-release, then the composed controller $\contpp$ as described in \Cref{alg:event-to-user} satisfies $(\eps,\delta)$-deletion-as-control.
\end{theorem}

We will prove the statement by showing that $\contpp$ satisfies $(\eps, \delta)$-deletion-as-confidentiality (\Cref{lem:PP_to_conf}) and then apply \Cref{thm:conf-cont} to prove that $\contpp$ satisfies $(\eps, \delta)$-deletion-as-control.

\begin{lemma}\label{lem:PP_to_conf}
If a controller \mech satisfies $(\eps,\delta)$-adaptive event-level pan-privacy with continual release, then the composed controller $\contpp$ as described in \Cref{alg:event-to-user} satisfies $(\eps,\delta)$-deletion-as-confidentiality  (\Cref{def:del-as-conf-subj}).
\end{lemma}

\begin{proof}
To reduce clutter, let $\cont = \contpp$ throughout this proof.
Fix a controller \mech that satisfies $(\eps, \delta)$-adaptive event-level pan-privacy, and fix \env and \subj in the deletion-as-confidentiality game. We must show that $(V_\env^{\real},\s_{\cont}^{\real}) \approxequiv{\eps,\delta} (V_\env^{\ideal}, \s_{\cont})$ using the pan-privacy hypothesis. Recall that in the deletion-as-confidentiality game, \subj does not send any messages to \env.

Below, we construct an adversary $\adv = \adv_{\env,\subj}$ for the adaptive pan-privacy game of $\mech$ (see \Cref{fig:pp_implies_ggv}). After adaptively querying $\mech$ (as in \Cref{alg:adaptive_PP_game}), \adv produces output $Z$. We will argue two key properties: First, if $\side = \IN$, then $Z= Z^{(\IN)}$ will be distributed as $(V_\env^{\real},\s_{\cont}^{\real})$. The latter is the view of the environment and state of the controller in the real confidentiality execution (\Cref{def:del-as-conf-subj}). Second, if $\side = \OUT$, then $Z^{(\OUT)}$ will be distributed as $(V_\env^{\ideal},\s_{\cont}^{\ideal})$, the corresponding quantities in the ideal confidentiality execution.
The adaptive pan-privacy of \mech implies that $\adv$'s views with $\side\in\{\IN,\OUT\}$ are $(\eps,\delta)$-indistinguishable.
Post-processing implies that
$Z^{(\IN)} \approxequiv{\eps, \delta} Z^{(\OUT)}$, completing the proof.

\begin{figure}
    \centering
    \includegraphics[scale=0.25]{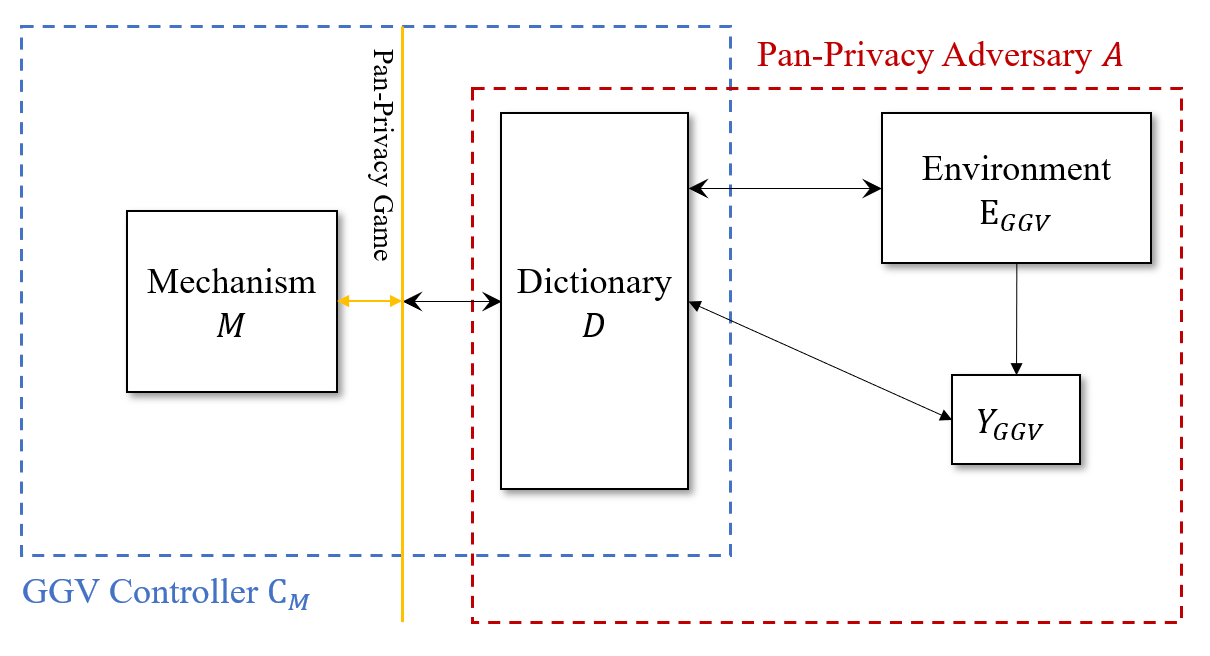}
    \caption{How \adv emulates a real GGV execution among $\env$, $\subj$, and $\contpp$.}
    \label{fig:pp_implies_ggv}
\end{figure}

Adversary $\adv$ emulates a real GGV execution among $\env$, $\subj$, and $\cont$, passing messages among them and performing each party's computations as defined in \Cref{sec:ggv} and depicted in \Cref{fig:pp_implies_ggv}. To emulate $\env$ and $\subj$, $\adv$ simply runs their code.
To emulate $\cont$, $\adv$ maintains an internal SHI dictionary $\dict$ as in \Cref{alg:event-to-user} (using state $\s_\dict$)  but does not emulate $\mech$ directly.
Instead, $\adv$ makes a queries in the pan-privacy game by calling $\textsc{Activate}$ in \Cref{alg:adaptive_PP_game} as follows. Let $\id_\subj$ be $\subj$'s $\id$. When the emulated parties \env or \subj sends message $(\id, \op)$ intended for emulated $\cont$, \adv does the following:

\begin{enumerate}
    \item If $\op$ is a special message from \env to terminate the alive phase of execution: Run $\dict.\del{\ensuremath{\id_\subj}}$, call \\ $\textsc{Activate}(\intru, \bot)$. Upon receiving $\s_\mech$, end the GGV execution and output $Z = (V_\env, \s_{\cont})$ where $V_\env$ is the view of the emulated $\env$ and $\s_{\cont} = (\s_\mech \|\s_\dict)$.
    \item Else if $\op = \delete$: Run $\dict.\del{\id}$ and send response $\bot$ to party $\id$ (on behalf of the emulated $\cont$).
    \item Else if $\id \in \dict$: Send response $\bot$ to party $\id$ (on behalf of the emulated $\cont$).
    \item Else if $\id \neq \id_\subj$: $\dict.\ins{\id}$, $\textsc{Activate}(\reg, \op(\id))$, and forward the response from $\mech$ to party $\id$ (on behalf of the emulated $\cont$).

    \item Else: Call $\textsc{Activate}(\chall, \op(\id))$ and forward the response to $\subj$ (on behalf of the emulated $\cont$).
\end{enumerate}

Let $\mu^{\mathit{GGV}}$ be the sequence of queries to $\mech$ made by the emulated $\cont$ (all calls to Line~\ref{line:calls-to-M} of \Cref{alg:event-to-user}).
Each query in $\mu^{\mathit{GGV}}$ results in a single call to $\textsc{Activate}$. These in turn may result in a query to $\mech$ on Lines~\ref{line:pp:call-to-m-1} and~\ref{line:pp:call-to-m-2} of \Cref{alg:adaptive_PP_game}.
But, due to the pan-privacy game some of these queries may not actually {reach} \mech: if $\side = \OUT$ (Line~\ref{line:pp:side-in}), or if an assertion is violated (Lines~\ref{line:pp:assert-intrusion} or~\ref{line:pp:assert-challenge}).
Let $\mu^{\mathit{PP}, \side}$ be the sequence of queries that actually reach $\mech$ for $\side \in \{\IN, \OUT\}$ (i.e., all calls to Lines~\ref{line:pp:call-to-m-1} and~\ref{line:pp:call-to-m-2}).

\begin{claim} Consider an execution of the adversary $\adv$ above in the pan-privacy game with mechanism $\mech$. Then, with probability 1:
\label{claim:queries-GGV-PP}\hfill
\begin{enumerate}
    \item $\mu^{\mathit{PP},\IN} = \mu^{\mathit{GGV},\mathit{real}}$, and
    \item $\mu^{\mathit{PP},\OUT} = \mu^{\mathit{GGV}, \mathit{ideal}}$.
\end{enumerate}
\end{claim}

\begin{proof}
\emph{Part 1 $(\side = \IN)$:}
As explained before the claim, every query to $\mech$ made by the emulated $\cont$ corresponds to a call to $\textsc{Activate}$ (and vice-versa).
Every call to $\textsc{Activate}$ results in a query in $\mu^{\mathit{PP},\IN}$ unless $\side = \OUT$ or one of the assertions in \Cref{alg:adaptive_PP_game}  (Lines~\ref{line:pp:assert-intrusion} and~\ref{line:pp:assert-challenge}) is violated. Because $\side = \IN$, it suffices to show that the assertions are never violated.
One assertion requires that $\intru$ is only called once (Line~\ref{line:pp:assert-intrusion}). This is by the construction of $\adv$ above.

The second assertion requires that $\chall$ is called only once and only with $\op \neq \tick$ (Line~\ref{line:pp:assert-challenge}).
This follows three properties of $\mu^{\mathit{GGV}}$. First, $\chall$ is only called when $\subj$ makes a non-\delete query.
Second, $\subj$ cannot make $\tick$ queries.
Finally, $\mu^{\mathit{GGV}}$ contains at most one non-\delete query from $\subj$. This follows automatically from the restrictions on \subj in the deletion-as-confidentiality game; specifically, after \subj makes a \delete query to $\cont$, it may not send any additional messages to $\cont$.

\emph{Part 2 $(\side = \OUT)$:}
By construction of $\adv$, $\subj$'s queries in $\mu^{\mathit{GGV}}$ causes $\adv$ to send a $\chall$ query in the pan-privacy game. The only difference from Part 1 is that these queries are never sent to $\mech$ in the pan-privacy game (Line~\ref{line:pp:side-in} of \Cref{alg:adaptive_PP_game}). And, although the emulated \subj receives $\bot$ responses from \adv (on behalf of $\cont$)--which it would not get in the ideal world--\subj can not send messages to \env and so can not influence the view of \env at all.
\end{proof}

Next we show that $Z^{(\IN)}$ is distributed as $(V_\env^{\real}, \s_{\cont}^{\real})$. Consider a real-world GGV execution where all parties use the same randomness as in $\adv$'s emulated execution. If $\side = \IN$, the view of the real-world GGV environment $\view_\env^{\mathit{GGV},real}$ is identical to the view of the emulated environment $\view_\env^{\mathit{PP},\IN}$.
From $\env$'s view, the only difference between the two executions is how  $\mech$ is implemented and queried. By Part 1 of \Cref{claim:queries-GGV-PP}, the queries processed by $\mech$ are identical in the two executions. Hence $\env$'s views of the two executions are identical. Likewise, the states $\s_\cont =(\s_\dict \| \s_\mech)$ are identical in the two executions.

Finally we show that $Z^{(\OUT)}$ is distributed as $(\view_\env^{\ideal}, \s_{\cont}^{\ideal})$. Consider an ideal-world GGV execution where all parties use the same randomness as in $\adv$'s emulated execution. In the ideal GGV execution, \subj cannot send messages to $\cont$, whereas in the emulated execution it can; however, by part 2 of \Cref{claim:queries-GGV-PP}, the queries processed by $\mech$ are identical in the two executions. Hence $\env$'s views of the two executions are identical. Likewise, the states $\s_{\cont} =(\s_\dict \| \s_\mech)$ are identical in the two executions.

Lastly, $Z^{(\side)}$ is a post-processing of $V_{\adv, \mech}^{(\side)}$ in the adaptive pan-privacy game. Hence, by event-level $(\eps, \delta)$-adaptive-pan-privacy with continual release, we have $Z^{(\IN)} \approxequiv{\eps, \delta} Z^{(\OUT)}$. Finally, since the dictionary \dict is SHI, the state $\s_\dict$ in the emulation is identically distributed to $\s_\dict$ in the GGV real and ideal worlds (since in all cases \subj's removal from \dict is undetectable). This implies that $(V_\env^{\real}, \s_{\cont}^{\real}) \approxequiv{\eps, \delta} (\view_\env^{\ideal}, \s_{\cont}^{\ideal})$, completing the proof.
\end{proof}

\begin{proof}[Proof of \Cref{thm:PPshi_filter}]
Our proof combines \Cref{lem:PP_to_conf} and \Cref{thm:conf-cont}.
To reduce clutter, let $\cont = \contpp$ throughout this proof.
Fix an $(\eps, \delta)$-adaptive event-level pan-private controller \mech with continual release. Then, by \Cref{lem:PP_to_conf}, the controller $\cont$ described in \Cref{alg:event-to-user} satisfies $(\eps, \delta)$-deletion-as-confidentiality.
Recall that deletion-as-confidentiality only considers the subset of data subjects $\subjLift$ satisfing two requirements. First, \subj does not send messages to \env. Second, \subj only \delete's when instructed to by \env. If we restrict to this class \subjLift of data subjects, then by \Cref{thm:conf-cont}, $\cont$ is $(\eps, \delta)$-deletion-as-control compliant for $\subjLift$.

What remains to argue is that for $\cont$ specifically, we can restrict to $\subjLift$ without loss of generality (using the same simulator \Sim as for $\subj \in \subjLift$).
It suffices to show that for any $(\env,\subj)$ there exists $(\env',\subj')$ with $\subj' \in \subjLift$ such that the transcripts in the executions $\exec{\cont(\Rand_\cont), \env, \subj}$ and $\exec{\cont(\Rand_\cont), \env', \subj'}$ are identical.
Intuitively, this holds because $\cont$ only ever sends $\bot$ to $\subj$. Recall that $\subj$ cannot make any $\tick$ queries, which is the only way for $\mech(\op(\id))$ (or $\cont$) to produce non-$\bot$ output. Hence, the view of $\subj$ can easily be simulated by $\env$, and the restriction that $\subj$ does not communicate with $\env$ is without loss of generality. From there, we can restrict \subj to be a dummy party without loss of generality, satisfying the requirements of $\subjLift$.

Fix any \env, \subj. \Cref{fig:PP_emulation} depicts
$\env'$ and $\subj'$. The subject $\subj'$ forwards messages from $\env'$  to $\cont$, but sends no messages to $\env'$. The environment $\env'$ emulates $\env$ and \subj in its head and emulates their interactions with other parties as follows:
\begin{itemize}[itemsep=0.3pt, topsep=4pt, partopsep=2pt]
    \item $\env'$ passes  messages between emulated parties \env and \subj freely.
    \item $\env'$ forwards  queries from \env intended for $\cont$; similarly, $\env'$ forwards  responses from $\cont$ to \env.
    \item $\env'$ forwards  queries from \subj intended for $\cont$ to the external $\subj'$ (who then forwards this to $\cont$).
    When reactivated, $\env'$ sends the response $\bot$ to \subj.
\end{itemize}

\begin{figure}
    \centering
    \includegraphics[scale=0.5]{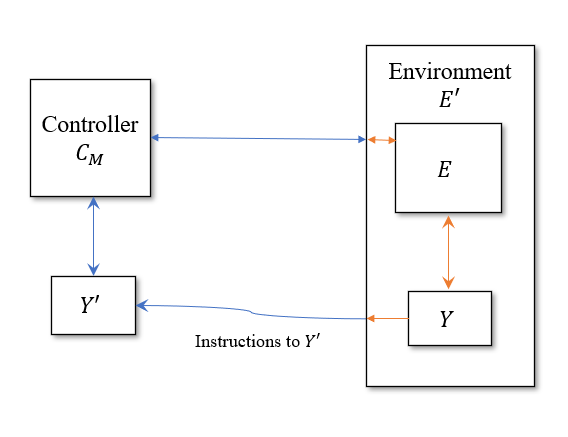}
    \caption{Construction of $\env'$ and $\subj'$ with internal emulated parties \env, \subj interacting with $\contpp$..}
    \label{fig:PP_emulation}
\end{figure}
By construction, $\subj'$ satisfies the condition that it only sends \delete if $\env'$ tells it to, and it never sends any message to $\env'$. Hence, $\subj' \in \subjLift$.
By the construction in \Cref{alg:event-to-user}, $\cont$ always sends $\bot$ in response to $\subj$'s queries.
Using this, it is easy to show that the emulation is perfect. Hence, the transcripts seen by \cont in the two executions  $\exec{\cont(\Rand_\cont), \env, \subj}$ and $\exec{\cont(\Rand_\cont), \env', \subj'}$ are identical.

\end{proof}

\subsection{Examples of pan-private controllers}
\label{sec:PP-examples}

\subsubsection{Counters}
\label{sec:PP-counters}

Perhaps the most widely studied algorithm that is pan-private under continual release is the tree mechanism of \citet{DworkNPR10,chan2011private}. Although originally formulated without pan-privacy, it relies on the composition of summations over different time intervals. Each of these can be made pan-private against a single intrusion by adding noise when the counter is initialized and again at release \citep{dwork2010pan}. The overall mechanism retains the same asymtptotic guarantees, since the noise at most doubles, and enjoys pan-privacy. We summarize its properties here:

\begin{proposition}[Combining \cite{dwork2010pan} with \cite{DworkNPR10,chan2011private}]\label{prop:pp-cr-counter}
    There is an $(\eps,\delta)$-adaptively event-level pan-private under continual-release algorithm that takes a time horizon $T$ and  a sequence of inputs $\datum_1,\datum_2, ... \in [0,1]$ and at the $t$-th \tick, for $t\in [T]$, releases $\out_t$ such that
    \[
        \out_t = \sum_{\substack{i\,:\,{x_i} \text{ received}\\ \text{by time } t}} \datum_i + Z_t,
    \]
    where $Z_t$'s are   distributed (independently of $x_i$'s but not each other) as $Z_t\sim N(0, \sigma_t^2)$ with $\sigma_t = O\left(\frac{\sqrt{\log^3 T \log(1/\delta)}}{\eps}\right)$.
\end{proposition}

In Section~\ref{sec:simul-comp}, we use the tree mechanism as a building block to construct a controller that satisfies deletion-as-control but not history independence, confidentiality, nor differential privacy. %
Moreover, one can use the tree mechanism to construct a controller that satisfies deletion-as-control but not history independence (cf.\  Proposition~\ref{prop:one-shot-DP}).
The controller could, for example, use the tree mechanism to count the number of distinct users $n$ it has seen. Applying Theorem~\ref{thm:PPshi_filter}, a controller could publish $n + N(0, \sigma_t^2)$ even after all users have requested deletion, all while satisfying deletion-as-control.
\emph{Such functionality is impossible under history independence.}

\subsubsection{Learning under Pan-Private Continual Release (and with Deletion)}
\label{sec:PP-MUL}

\newcommand{\modelset}{\Theta}
\newcommand{\loss}{\ell}
\newcommand{\Loss}{\cL}

As an example, we show how the results of this section allow one to maintain a model trained via gradient descent with relatively little noise---the algorithm's error is similar to that of state of the art federated machine learning algorithms.

Consider the following optimization problem: given a loss function $\loss: \modelset\times \universe\to \reals$ and a data set $\data = (\datum_1,...,\datum_n)$, we aim to find $\model \in \modelset$ which approximately minimizes
\begin{equation} \label{eq:loss}
    \Loss_{\data}(\model) = \sum_{i=1}^n \loss(\model; \datum_i) \, .
\end{equation}

To analyze convergence, we assume a convex loss function. Nevertheless, the method we give applies more broadly---the privacy or deletion guarantees hold regardless of convexity, and the algorithms makes sense as long as the loss function is roughly convex near its minimum.

Algorithms for this problem that are differentially private under continual release were studied for online learning and efficient distributed learning~\citep{KairouzM00TX21}. The algorithm they consider is not panprivate. However, it only access the data via  the tree-based mechanism for continual release of a sum \cite{DworkNPR10,chan2011private} (in this case, releasing the sum of the gradients of users' contributions to the overall loss function over the evaluation of a first-order optimization algorithm). We observed above that this can be made panprivate with only a constant factor increase in the added noise (\Cref{prop:pp-cr-counter})

To allow us to apply the convergence analysis of \citet{KairouzM00TX21} as a black box, we assume: 
\begin{inparaenum}
\item $\modelset$ is convex, and $\loss(\cdot ; \datum)$ is convex for every choice of $\datum$; 
\item $\loss$ is $1$-Lipschitz in $\model$ (that is, suppose $\modelset \in \reals^d$ and for all $\model, \model'$ and $\datum$, we have $\abs{\ell(\model;\datum) - \ell(\model';\datum)} \leq  \norm{\model-\model'}{2}$);
\item One new user arrives between adjacent \tick{}s, though deletions may occur at any time.
\end{inparaenum}

\begin{proposition}[Derived from Theorem 5.1 from \textit{arXiv v3} of \citet{KairouzM00TX21}]
    There is an $(\eps,\delta)$-adaptively pan-private, continual-release algorithm that takes takes a time horizon $T$, 
    a learning rate $\lambda>0$, 
    and  a sequence of inputs $\datum_1,\datum_2, ...$ and at the $t$-th \tick  
    releases a model $\model_t$ such that, for every $t \in [T]$,  
    data set $\data_t \in \univ^n$ (consisting of records received by tick $t$), 
    $\model^* \in \modelset$, and $\beta>0$, 
    with probability at least $1-\beta$ over the coins of the algorithm,
    
        \begin{align*}
            \Loss_{\data_t} (\model_t)  - \Loss_{\data_t} (\model^*)=  O\paren{\norm{\model^*}{} \cdot  \paren{\frac 1 {\sqrt{n}} +  d^{1/4}  \sqrt{\frac{\ln^2(1/\delta) \ln (1/\beta)}{\eps n}}   }  }
        \end{align*}

\end{proposition}

This result immediately yields a controller that satisfies $(\eps,\delta)$-deletion as control and maintains a model whose accuracy tracks that of the best model trained on all the arrivals so far (where repeated arrivals of the same user are ignored). The version above uses a fixed learning rate $\lambda$ which can be tuned to get error $\tilde O(d^{1/4}/\sqrt n)$ for any particular $n$; however, one could also decrease the learning rate as $1/\sqrt{t}$ to get bounds that hold for all data set sizes. 

In contrast to the HI-style algorithms proposed for machine unlearning \cite{gupta2021adaptive}, this controller need not update the model when a user is deleted.

\section{Parallel Composition}
\label{sec:simul-comp}

\newcommand{\SET}{\ensuremath{\mathsf{set}}\xspace}
\newcommand{\LOOKUP}{\ensuremath{\mathsf{lookup}}\xspace}
\newcommand{\INCREMENT}{\ensuremath{\mathsf{increment}}\xspace}
\newcommand{\GET}{\ensuremath{\mathsf{get}}}
\newcommand{\COUNT}{\ensuremath{\mathsf{getCount}}\xspace}
\newcommand{\users}{\ensuremath{\mathcal{U}}\xspace}

In this section, we show that our definition does more than unify approaches to deletion based on history independence, confidentiality, and differential privacy. We describe a controller that satisfies deletion-as-control but none of the other three notions. 
To prove that the controller satisfies deletion-as-control, we prove
that the definition enjoys a limited form of \emph{parallel
  composition}: a controller that is built from two constituent
sub-controllers, each of which  is run independently
(Figure~\ref{fig:simul-composition}), will be deletion compliant if
the component controllers are compliant and satisfy additional
conditions. 
We focus on a special case that suffices for our needs, and leave a general treatment of composition of deletion-as-control for future work.

To anchor the discussion, we consider the touchstone composed controller ``Public Directory with DP Statistics,'' described in
Algorithm~\ref{alg:directory-with-statistics}. It provides a public
directory (e.g., a phone book) which, in addition to answering
directory queries, periodically reports the total
number of users that have made queries to the directory.
We build the touchstone controller \cont from two SHI dictionaries $\dict$ and $\users$
(Section~\ref{sec:SHI_DICT}) and the pan-private tree mechanism
$\mech$ (Proposition~\ref{prop:pp-cr-counter}).  $\dict$ implements
the directory functionalities of reading and writing.  $\mech$ keeps
track of an approximate count of users---current and former---that have looked up an entry in $\dict$ for the first $T$ epochs (e.g., the time between $\tick$s). $\users$ is an auxiliary dictionary used to ensure that $\mech$ counts distinct users.

\begin{algorithm}
\caption{Public directory with DP statistics}
\label{alg:directory-with-statistics}
    \begin{algorithmic}[1]

        \Procedure{Initialize}{time horizon $T \in \N$, privacy parameters $\eps, \delta$} \Comment{Run on first activation}
       
            \State Initialize two SHI dictionaries $\dict$ for the directory and $\users$ for users.
            \State 
            Initialize the tree mechanism $\mech$ with parameters $(\eps,\delta,T)$
        \EndProcedure
        \Statex
        \Procedure{Activate}{$\cid, \op(arg)$}
            \If{$\op = \delete$}
                \State $\dict.\delete(\cid)$
                \State $\users.\delete(\cid)$
                \State \Return $\bot$
            \EndIf
            \If{$\op = \SET$}
                \State $\dict.\SET(\cid, arg)$
                \State \Return $\bot$
            \EndIf
            \If{$\op = \GET$}
                \If{$\cid \notin \users$}
                    \State $\users.\SET(\cid, 1)$
                    \State Insert the value 1 into $\mech$'s data stream
                \EndIf \State \Return{$\dict.\GET(arg)$}
            \EndIf
            \If{$\op = \COUNT$}
                \State \Return $\mech(\tick)$
            \EndIf
        \EndProcedure
    \end{algorithmic}
\end{algorithm}

    \begin{figure}
        \centering
        \includegraphics[scale=0.7]{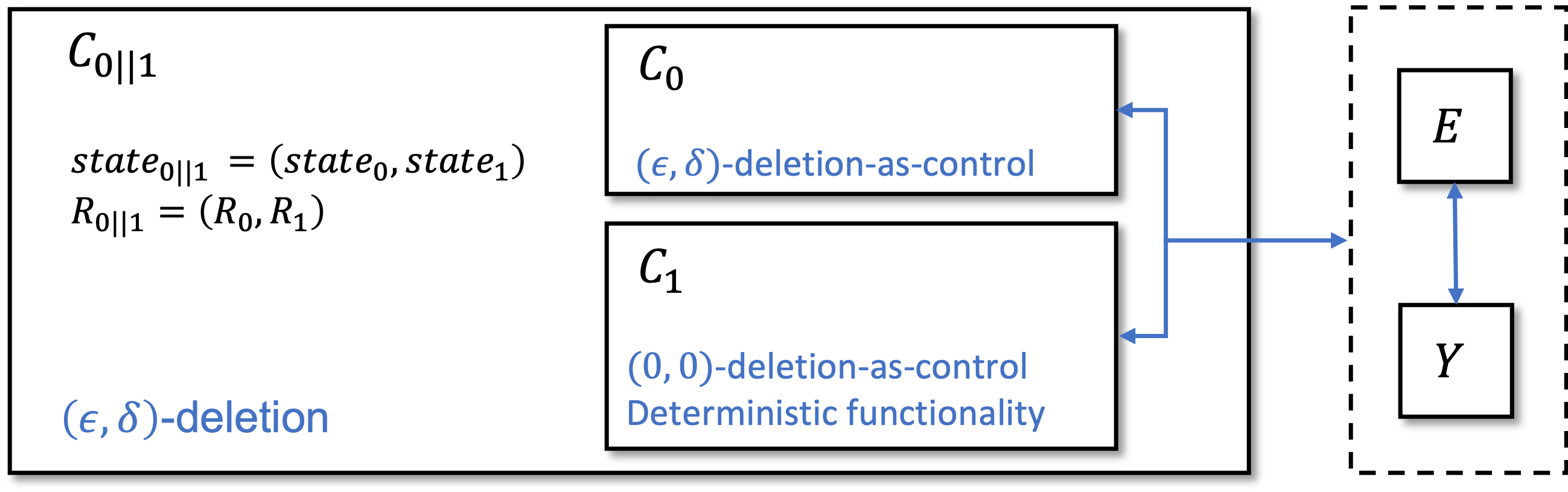}
        \caption{Parallel composition of controllers $\cont_0$ and $\cont_1$.}
        \label{fig:simul-composition}
    \end{figure}

To prove that $\cont$ satisfies deletion-as-control, we prove that one can build deletion-as-control controllers from two constituent sub-controllers by parallel composition. We consider the special case where one sub-controller satisfies $(0,0)$-deletion-as-control and implements a \emph{deterministic functionality}, and where both sub-controllers are \emph{query-response controllers}. We define these next.
Note that the controllers given in Section~\ref{sec:adaptive-HI} and Section~\ref{sec:DP} are all query-response controllers. 
While we restrict our attention to a single controller of each type, the result immediately extends by induction to many query-response controllers satisfying $(0,0)$ deletion-as-control with deterministic functionalities.

\begin{definition}[Deterministic functionality]
\label{def:deterministic-func}
A controller $\cont$ implements a \emph{deterministic functionality} if for all $\env$, $\subj$, the transcript of $\exec{\cont(\Rand_\cont), \env,\subj}$ is independent of $\Rand_\cont$.
\end{definition}

\begin{definition}[Query-response controller]
\label{def:query-response}
$\cont$ is a \emph{query-response} controller if it always replies to the same party that activated it. Namely, when $\cont$ is activated with $(\cid,\msg)$ on its input tape, it always halts with $(\cid, \msg')$ on its output tape for some $\msg'$.
\end{definition}

The example controllers we discuss are generally query-response. But
a messaging server, for example, might not be: a request from Alice to
send something Bob could result in a push from the server to
Bob.

\begin{definition}[Parallel composition]
\label{def:simul-comp}
Let $\cont_0$ and $\cont_1$ be controllers. We define their \emph{parallel composition}, denoted $\cont_{0\|1}$, to be the controller that emulates $\cont_0$ an $\cont_1$ internally. When activated with $\op_\id$, $\cont_{0\|1}$ computes $\out_0 \gets \cont_0(\op_\id)$ and $\out_1\gets \cont_1(\op_\id)$ and returns $\out_{0\|1} = (\out_0, \out_1)$. The controllers $\cont_0$ and $\cont_1$ are emulated with disjoint portions of $\cont_{0\|1}$'s randomness tape $\Rand_{0\|1} = (\Rand_0, \Rand_1)$ and work tape $\s_{0\|1} = (\s_0, \s_1)$.
\end{definition}

\begin{theorem}
\label{thm:simul-comp}
Let $\cont_0$ and $\cont_1$ be query-response controllers and let $C_{0\|1}$ be their parallel composition.
If $\cont_0$ satisfies $(\eps,\delta)$-deletion-as-control, and if $\cont_1$ satisfies $(0,0)$-deletion-as-control and has deterministic functionality, then $\cont_{0\|1}$ satisfies $(\eps,\delta)$-deletion-as-control.
\end{theorem}

\begin{proof}[Proof of Theorem~\ref{thm:simul-comp}]
Define the parallel simulator \\ $\Sim_{0\|1}(\queries_\env,\Rand_{0\|1}, \s_{0\|1})$ as outputting $\Rand'_{0\|1} = (\Rand'_0,\Rand'_1)$ where each half is computed using the corresponding simulator. That is, for each $i \in \{0,1\}$:
$$\Rand'_i \sim \Sim_i(\queries_\env, \Rand_i,\s_i).$$

\noindent
By hypothesis, all of the following hold for all $\env$, $\subj$.
\begin{itemize}[itemsep=0.3pt, topsep=4pt, partopsep=2pt]
\item $(\Rand'_0 \approxequiv{\eps,\delta} \Rand_0)$ AND $(\Rand'_1 \equiv \Rand_1)$,
\item $(\s_1 = \s'_1)$ OR $(\s_1 = \bot)$
\item With probability at least $1-\delta$:
$(\s_0 = \s'_0)$ OR $(\s_0 = \bot)$.
\end{itemize}

\noindent
We must show that for all $\env$, $\subj$, both of the following hold in the real-ideal probability experiment given in Definition~\ref{def:del-as-cont}.
\begin{itemize}[itemsep=0.3pt, topsep=4pt, partopsep=2pt]
\item $\Rand'_{0\|1} \approxequiv{\eps,\delta} \Rand_{0\|1}$
\item With probability at least $1-\delta$: $\s_{0\|1} = \s'_{0\|1} \neq \bot$ or $\s_{0\|1} = \bot$ (i.e., the real execution doesn't terminate).
\end{itemize}

\noindent
We consider the state first. First, observe the following:
\begin{align*}
(\s_0 = \bot) &\implies (\s_1 = \bot), \text{ and} \\
(\s_0 \neq \bot) &\implies (\s_1 \neq \bot) \implies (\s_1 = \s'_1)
\end{align*}
This is because both $\s_0$ and $\s_1$ result from a single real execution $\exec{\cont_{0\|1}, \env,\subj}$ that either terminates ($\s_{0\|1} = \bot$) or does not ($\s_{0\|1} \neq \bot$). Hence,

\begin{align*}
    &\Pr\left[\s_{0\|1} \in \{\s'_{0\|1}, \bot\}\right]\\
    &\qquad=
    \Pr\left[\s_0 \in \{\s'_0, \bot\}\right] \ge 1-\delta
\end{align*}

.

It remains to show that $\Rand'_{0\|1} \approxequiv{\eps,\delta} \Rand_{0\|1}$. 
Below, we prove that $\Rand'_0$ and $\Rand'_1$ are independent random variables. This suffices because $(\Rand'_0 \approxequiv{\eps,\delta} \Rand_0)$ and $(\Rand'_1 \equiv \Rand_1)$.

Let $\queries$ be the queries received by $\cont_{0\|1}$ in the real execution. By construction, these are also the queries received by the emulated $\cont_0$ and $\cont_1$.  
Observe that for $i \in\{0,1\}$ there exists randomized $S_i$ such
that $\Rand'_i \sim S_i(\Rand_i,\queries)$. Namely, $\Rand'_i$ only
depends on $\Rand_i$ and $\queries$. This is because $\Rand'_i \gets
\Sim_i(\Rand_i, \queries_\env, \s_i)$ where $\queries_\env$ and
$\s_i$ are deterministic functions of $\Rand_i$ and $\queries$
(and the code of $\cont_{0\|1}$, $\env$, $\subj$ which are
fixed).%

The above implies two useful facts. First, $\Rand'_1$ is independent of $\Rand'_0$ conditioned on $\queries$. This follows from the fact that $\Rand'_i \sim S_i(\Rand_i, \queries)$, combined with the fact that $\Rand_1$ and $\Rand_0$ are independent (by construction). Second, $\Rand'_1$ is independent of the random variable $\queries$. That is, for all $\rand'_1$ and $q$, $\Pr[\Rand'_1 = \rand'_1 | \queries = q] = \Pr[S_1(\Rand_1,\queries) = \rand'_1 | \queries = q] = \Pr[S_1(\Rand_1,q) = \rand'_1] = \Pr[\Rand'_1 = \rand'_1].$ The middle equality uses the independence of $\queries$ and $\Rand_1$, which follows from the hypothesis that $\cont_1$ has a deterministic functionality.

Fix $\rand'_0$ and $\rand'_1$. We check independence of $\Rand'_0$ and $\Rand'_1$ as follows. 

    \begin{align*}
    \Pr[\Rand'_0 = \rand'_0 \land \Rand'_1 = \rand'_1]
    &=
    \Pr[\Rand'_0 = \rand'_0] \cdot \sum_{q} \Pr[\queries = q | \Rand'_0 = \rand'_0] \cdot \Pr[\Rand'_1=\rand'_1 | \queries = q, \Rand'_0 = \rand'_0] 
    \\&=
    \Pr[\Rand'_0 = \rand'_0] \cdot \sum_{q} \Pr[\queries = q | \Rand'_0 = \rand'_0] \cdot 
    \Pr[\Rand'_1=\rand'_1 | \queries = q]
    \\&=
    \Pr[\Rand'_0 = \rand'_0] \cdot \Pr[\Rand'_1 =\rand'_1] \cdot  \sum_{q} \Pr[\queries = q | \Rand'_0 = \rand'_0]
    \\&=
    \Pr[\Rand'_0 = \rand'_0]\cdot \Pr[\Rand'_1 = \rand'_1]
\end{align*}

\noindent
The second equality uses the independence of $\Rand'_1$ and $\Rand'_0$ conditioned on $\queries$. The third equality uses the independence of $\Rand'_1$ and $\queries$. 
\end{proof}

\begin{corollary}\label{cor:directory_DP_Stats}
The Public Directory with DP Statistics (Algorithm~\ref{alg:directory-with-statistics}) satisfies $(\eps,\delta)$-deletion-as-control.
\end{corollary}

\begin{proof}
Let $\cont_0$ be the controller $\contpp$ defined relative to the pan-private tree mechanism $\mech$ (Algorithm~\ref{alg:event-to-user}).
Let $\cont_1$ be the controller relative to a SHI dictionary (Definition~\ref{def:cont_impl}).
By construction, $\cont_0$ and $\cont_1$ are query-response controllers and $\cont_1$ has a deterministic functionality. 
$\cont_0$ satisfies $(\eps,\delta)$ deletion-as-control and $\cont_1$ satisfies $(\eps,\delta)$ deletion-as-control (Theorems~\ref{thm:PPshi_filter} and~\ref{thm:AHI_to_del} respectively).
Algorithm~\ref{alg:directory-with-statistics}  is equivalent to the parallel composition $\cont_{0\|1}$ of $\cont_0$ and $\cont_1$. 
By parallel composition, Algorithm~\ref{alg:directory-with-statistics} satisfies $(\eps,\delta)$ deletion-as-control.
\end{proof}

\section{Conclusion}
\label{sec:conclusion}

Defining deletion-as-control in a way that is both expressive and meaningful is the central challenge of our work. We believe that our definition succeeds, providing a new perspective to the ongoing discussion on how to give users control over their data.

More work is needed to understand how deletion-as-control handles the complexity of real-life functionalities. For example, we are far from understanding the implications for something like Twitter, though the Public Bulletin Board serves as a starting point. Analyzing complex functionalities may involve further studying the adaptive variants of history independence and pan-privacy defined in this work.
Beyond specific functionalities, more work is needed to interpret the guarantees provided by our notion. In particular, we have a limited view of what can be said about groups of individuals and about composition. 

These are directions for future technical work, but also for normative, legal, and policy considerations. 
Consider, for example, the goal of machine unlearning: maintaining a model while respecting requests to delete. The algorithms in the machine unlearning literature seek to approximate what one would get by retraining from scratch (that is, history independence). But using adaptive pan-privacy, one can satisfy deletion-as-control without updating the model in response to deletion requests. 
Each behavior might be appropriate for a different setting, depending on the most relevant measure of model accuracy. The difference may also relate to whether one adopts an individual- or group-based view of a right to erasure. 
We hope that our work inspires further exploration of these questions.

\section*{Acknowledgements}

    We are grateful to helpful discussions with many colleagues, notably Kobbi Nissim. A.C. by the National Science Foundation under Grant No. 1915763 and by the DARPA SIEVE program under Agreement No. HR00112020021. A.S. and M.S. were supported by NSF awards CCF-1763786 and CNS-2120667, as well as faculty research awards from Google and Apple. P.V. was supported by the National Research Foundation, Singapore, under its NRF Fellowship programme, award no. NRF-NRFF14-2022-0010.

\bibliographystyle{abbrvnat}
\bibliography{bibliography}{}

\appendix

    \section{Lemmas on Couplings and Marginal Distributions}\label{app:coupling-lemma}

\begin{lemma}[Coupling Lemma]
\label{lem:approx_coupling}
Let $P,Q$ be probability distributions on sets $\X$ and $\cY$, respectively, and let $f:\X\to\cZ$ and $g:\cY\to\cZ$ be (deterministic) functions with the same codomain $\cZ$. 
Suppose that $f(X)\approx_{\eps,\delta} g(Y)$ when $X\sim P$ and $Y\sim Q$. 

Consider the collection of distributions $\{Q_x \in \Delta(\cY): x\in \X\}$ on the set $\cY$, where $Q_x$ denotes the distribution on $Y$ conditioned on the event that $g(Y)=f(x)$ (that is $Q_x = Q|_{\set{y: g(y)=f(x)}}$). In the case that $g^{-1}(f(x))$ is empty, $Q_x$ assigns probability to a default value  $\bot \in \cY$.

If we select $X\sim P$ and then sample $Y'\sim Q_X$ (so that $\Pr(Y'=y|X=x) = Q_x(y)$), then 
\begin{enumerate}[itemsep=0.3pt, topsep=4pt, partopsep=2pt]
    \item $Y'\approx_{\eps,\delta} Y$, and
    \item $f(X)=g(Y')$ with probability at least $1-\delta$ over $(X,Y)$.
\end{enumerate}
\end{lemma}

For example, consider an $(\eps,\delta)$-DP mechanism $\mech$ using randomness $\Rand\sim\distrand$, and let $\Ds,\Ds'$ be a pair of neighboring inputs. 
Take $f(\Rand) = \mech(\Ds;\Rand)$ and $g(\Rand) = \mech(\Ds'; \Rand)$.
Lemma~\ref{lem:approx_coupling} gives a way to jointly sample randomness $(\Rand,\Rand')$ such that $\Rand \sim \distrand$, $\Rand'\approxequiv{\eps,\delta} \distrand$, and with probability at least $1-\delta$: $\mech(\Ds;\Rand) = \mech(\Ds';\Rand')$. 
For simple DP mechanisms, it is easy to understand what this sampler does. For instance, let $\Lap(1/\eps;\Rand)$ be an algorithm that samples from the Laplace distribution when $\Rand \sim \distrand$ is uniform. Consider $\mech(\Ds;\Rand) = f(\Ds) + \Lap(1/\eps;\Rand)$ for some $f$ with sensitivity 1. The sampler in Lemma~\ref{lem:approx_coupling} samples $\Rand\sim \distrand$, and samples $\Rand'$ uniformly conditioned on $\Lap(1/\eps;\Rand') = f(\Ds') - f(\Ds)$.

\begin{proof}
Consider the randomized map $G^*:\cZ\to \cY$ that inverts $g$ for inputs drawn from $Q$---namely, $G^*(z)$ returns a random element according to $Q|_{\set{y: g(y)=z}}$ and returns the default element $\bot$ with probability 1 if $g^{-1}(z)$ is empty. Running $G^*$ on $g(Y)$ returns a sample distributed identically to $Y$ conditioned on $g(Y)$. That is, 
$$ \Bparen{g(Y),\ G^*(g(Y))} \equiv \Bparen{g(Y),\  Y}. $$
Since the $\approx_{\eps,\delta}$ relation is preserved by processing, and since $f(X)\approx_{\eps,\delta} g(Y)$, we have 
$$ \Bparen{g(Y),\  G^*(g(Y))} \approx_{\eps,\delta}  \Bparen{f(X),\  G^*(f(X))}. $$
Now for every $x\in \X$, the random variable $G^*(f(x))$ is distributed as $Q_x$. Therefore, $G^*(f(X))$ is distributed as $Y'$ from the theorem statement and  $Y' \approx_{\eps,\delta} Y$ (the first requirement of the theorem). Furthermore, $g(Y')$ equals $f(X)$ except when $f(X)$ lies outside the image of $g$. Since $f(X) \approx_{\eps,\delta} g(Y)$,  the value $f(X)$ lies outside the image of $g$ with probability at most $\delta$. This establishes the  theorem's second requirement.
\end{proof}

\para{A computable analogue of Lemma~\ref{lem:approx_coupling}}

To demonstrate deletion-as-control compliance, we need a simulator to sample ideal-world randomness $\Rand'_\cont$ whose marginal distribution is close to the uniform distribution while also guaranteeing that $\s'_\cont = \s_\cont$. A natural strategy is to sample $\Rand'_\cont$ from distribution conditioned on the equality of the states. Applying Lemma~\ref{lem:approx_coupling}, the resulting distribution suffices for our purposes in this paper, where $X=\Rand_\cont$ and $Y = \Rand'_\cont$. 
But there is a technicality that needs to be dealt with. In full generality, the conditional distributions $Q_x$ in that lemma may not be sampleable in finite time. This will not be a problem for us, as our applications of the lemma will be restricted to functions that depend only a finite prefix of $\Rand_\cont$ and $\Rand'_\cont$.
For completeness, we state a computable analogue of Lemma~\ref{lem:approx_coupling}. Elsewhere, we will elide this technicality and invoke the simpler Lemma~\ref{lem:approx_coupling}.

\begin{lemma}[Computable Coupling Lemma]
\label{lem:approx_coupling_computable}
Let $P$, $Q$, $\X$, $\cY$, $\cZ$, $f$, $g$, and $Q_x$ as in Lemma~\ref{lem:approx_coupling}.
Let $\X = \bit{*} = \cY$.
Suppose that $g(y)$ depends only on a finite prefix of $y\in \bit{*}$.  Suppose also that for all $x \in \bit{*}$, either $f(x)$ depends only on a finite prefix of $x \in \bit{*}$, or  $g^{-1}(f(x)) =\emptyset.$
Then for any $x$, $Y\sim Q_x$ can be written as $Y_1 \| Y_2$ where $|Y_1|$ is finite and $Y_2$ is uniform in $\bit{*}$ and independent of $Y_1$.
\end{lemma}

In our applications, $X$ and $Y$ will be $\Rand_\cont$ and $\Rand'_\cont$ respectively. $f$ and $g$ will be functions of the real and ideal executions, respectively. By construction, an ideal execution is always finite, and therefore $g(Y)$ will only depend on a finite prefix of $\Rand'_\cont$. A real execution is either finite---in which case $f(X)$ only depends on a finite prefix of $\Rand_\cont$---or it is infinite. If it is infinite the real execution (and $f$) outputs $\bot$, whereas the ideal execution (and $g$) can never output $\bot$.

\begin{lemma}[Conditioning Lemma \citep{KasiviswanathanS08}]\label{lem:conditioning}
    Let $(A,B)$ and $(A',B')$ be pairs of random variables over the same domain $\cA\times \cB$ such that $(A,B)\approxequiv{\eps,\delta}(A',B')$. Then for all $\delta'> 0$, with probability at least $1-\gamma$ over $a\gets A$, we have 
    \[ 
        B|_{A=a} \approxequiv{\eps',\delta'} B'|_{A'=a} \,,
    \]
    where $\eps' = 3\eps$ and $\gamma = \frac{2\delta}{\delta'} + \frac{2\delta}{1 - \exp(-\eps)}$.
\end{lemma}

\section{Differential Privacy Preliminaries}
\label{app:dp-prelims}

A dataset $\Ds = (\ds_1, \ldots, \ds_n) \in \universe^n$ is a vector of elements from universe \universe. Two datasets are {\em neighbors} if they differ in at most one coordinate. Informally, differential privacy requires that an algorithm's output distributions are similar on all pairs of neighboring datasets. We use two different variants of differential privacy. The first one (and the main one used in this paper) is the standard definition of differential privacy.

\begin{definition}[Differential Privacy~\citep{DMNS16,DworkKMMN06}]\label{def:differentially private} A randomized algorithm $\mech: \universe^n \rightarrow \mathcal{Y}$ is {\em $(\eps, \delta)$-differentially private} if, for every pair of neighboring datasets $\Ds, \Ds'\in \universe^n$, 
the distributions of $\mech(\Ds)$ and $\mech(\Ds')$ are defined for the same $\sigma$-algebra $\Sigma_\cY$ and
$$\mech(\Ds) \approxequiv{\eps,\delta} \mech(\Ds') \, .$$
 \end{definition}

Differential privacy protects the privacy of groups of individuals, and is closed under post-processing. Moreover, it is closed under adaptive composition. For a fixed dataset \Ds, {\em adaptive composition} states that the results of a sequence of computations satisfies differential privacy even when the chosen computation $\mech_t(\cdot)$ at time $t$ depends on the outcomes of previous computations $\mech_1(\Ds), \ldots, \mech_{t-1}(\Ds)$. Under adaptive composition, the privacy parameters add up.

\begin{lemma}[Post-Processing]\label{prelim:postprocess} If $\mech: \universe^n \rightarrow \mathcal{Y}$ is $(\eps, \delta)$-differentially private, and $\mathcal{B} : \mathcal{Y} \rightarrow \mathcal{Z}$ is any randomized function, then the algorithm $\mathcal{B} \circ \mech$ is $(\eps, \delta)$-differentially private.
\end{lemma}

\begin{lemma}[Group Privacy]\label{prelim:group_privacy} Every $(\eps, \delta)$-differentially private algorithm \mech is $\left(k\eps, \delta\frac{e^{k\eps} -1}{e^\eps-1}\right)$-differentially private for groups of size $k$. That is, for all datasets $\Ds, \Ds'$ such that $\|\Ds - \Ds' \|_0 \leq k$ and all subsets $Y \subseteq \mathcal{Y}$,
\begin{equation*}
    \Pr[\mech(\Ds) \in Y] \leq e^{k\eps} \cdot \Pr[\mech(\Ds') \in Y] + \delta \cdot \frac{e^{k\eps} -1}{e^\eps-1}.
\end{equation*}
\end{lemma}

\begin{definition}[Composition of $(\eps, \delta)$-differential privacy]\label{prelim:composition} Suppose $\mech$ is an adaptive composition of differentially private algorithms $\mech_1, \ldots, \mech_T$. If for each $t \in [T]$, algorithm $\mech_t$ is $(\eps_t, \delta_t)$-differentially private, then \mech is $\left(\sum_t \eps_t, \sum_t \delta_t\right)$-differentially private.
\end{definition}

\section{Strong History Independence and Deletion-as-control}
\label{app:perfect_shi_del}

We show that strong history independence implies deletion-as-control. Throughout, we use SHI to mean $(0,0)$-SHI.

A foundational result in the study of history independent data structures  characterizes the behavior of SHI implementations of a large class of ADTs called reversible ADTs.

\begin{definition}[Reversible ADT] 
An ADT is \emph{reversible} if, for all logical states $A$ and $B$ reachable from the initial state, there exists some finite sequence of operations $\seq$ such that $\adt(A, \seq).\s = B$. That is, all reachable states are \emph{mutually reachable}.
\end{definition}

\begin{theorem}[\cite{hartline2005characterizing}]
\label{thm:shi_canonical} 
A reversible data structure is SHI if and only if each logical state has a single canonical memory representation for every setting of the implementation's randomness tape.
\end{theorem}

\begin{claim}\label{clm:rev_not_del} If an ADT supports logical deletion, then it is reversible. The converse is not true.
\end{claim}

\begin{proof}
(Logical deletion $\implies$ reversible): To prove reversibility, it is sufficient to give a sequence of operations that takes the ADT back to its initial state. Fix a state $A$ that is reachable from the initial state. Pick any sequence $\seq$ that takes the initial state to $A$. Delete all of the $\id$'s in $\seq$ one by one. By logical deletion, the resulting state is the initial state.

(Reversible $\centernot\implies$ logical deletion): 
Consider the ADT where the set of (non-delete) operations is the space of IDs and the logical state consists of the most recent ID. Suppose also that $\outLog = \bot$ always. 
This ADT is reversible: for any IDs $A$ and $B$, $\adt((A,B)).\s = B$. 
(We note that this ADT trivially admits a SHI implementation.)

Consider two sequences of IDs $(A,C)$ and $(B,C)$, where $A\neq B$. If this ADT supports logical deletion, then we derive a contradiction: 

\begin{align*}
&(A,\bot) = \adt((A,C,\delete(C))) = \adt((C,\delete_C)) \\
&= \adt((B,C,\delete(C))) = (B,\bot),
\end{align*}
where the first and last equalities follow from logical deletion and the second and third follow by construction.
This counter example illustrates that deletion may require the data controller to keep sufficient records of all users to update the state. 
\end{proof}

We now prove that SHI implies deletion-as-control. We restate the theorem for convenience.
\SHIDEL*

\begin{proof}[Proof of Theorem~\ref{thm:perfect_shi_del}]
Fix \cont and any \env and \subj in the deletion-as-control game. By the definition of logical deletion, the sequence \queries in the real world and $\queries'$ in the ideal world are logically equivalent. 
Let $\Sim$ be the simulator that always outputs $\Rand'_\cont = \Rand_\cont$. 
By Theorem~\ref{thm:shi_canonical}, the real-world and ideal-world states are identical (since the randomness and the logical states are identical). So, we satisfy the first condition of deletion-as-control. Furthermore, since the value of $\Rand'_\cont$ equals the value of $\Rand_\cont$, they are identically distributed. Thus, this controller satisfies $(0,0)$-deletion-as-control.
\end{proof}

\end{document}